\documentclass[12pt]{article}%
\usepackage{enumitem}
\usepackage{adjustbox}
\usepackage{bm}
\usepackage{amsfonts}
\usepackage[colorlinks,citecolor=blue,urlcolor=blue,hyperfootnotes=false]{hyperref}
\usepackage{amssymb}
\usepackage[bottom]{footmisc}
\usepackage[hyperref]{ntheorem}
\usepackage{graphicx}
\usepackage{float}
\usepackage[figuresright]{rotating}
\usepackage{natbib}  
\usepackage[doublespacing,nodisplayskipstretch]{setspace}
\usepackage[margin=1in, a4paper]{geometry}
\usepackage{booktabs}
\usepackage{threeparttable}
\usepackage{amsmath}
\usepackage{lscape}
\usepackage{scrextend}
\usepackage{relsize}
\usepackage{multicol}
\usepackage{chngcntr}
\usepackage{etoolbox}
\usepackage{amssymb}
\usepackage{tabularx}
\usepackage[english]{babel}
\usepackage{multirow}
\usepackage{booktabs}
\usepackage[labelfont=bf]{caption}
\usepackage{float}
\usepackage{titlesec}
\usepackage{array}
\usepackage{color}
\usepackage{afterpage}
\usepackage{framed}
\usepackage{float}
\usepackage{afterpage}%
\providecommand{\U}[1]{\protect\rule{.1in}{.1in}}
\allowdisplaybreaks

\theoremstyle{plain}
\newcommand{\ii}{{\rm i}}
\newcommand{\ee}{{\rm e}}
\newtheorem{theorem}{Theorem}[section]
\newtheorem{assumption}{Assumption}[section]

\newtheorem{corollary}{Corollary}[section]

\newtheorem{lemma}{Lemma}[section]

\newenvironment{proof}[1][Proof]{\noindent\textbf{#1.} }{\ \rule{0.5em}{0.5em}}

\numberwithin{equation}{section}
\numberwithin{table}{section}
\numberwithin{figure}{section}

\theorembodyfont{\normalfont}

\deffootnote{1em}{1.6em}{\thefootnotemark \enskip}

\numberwithin{equation}{section}
\numberwithin{assumption}{section}

\definecolor{lightgray}{gray}{0.9}
\begin{document}

\title{Specification tests for regression models with measurement errors\footnote{The authors contributed equally to this work and are listed in alphabetical order.}}
\author{Xiaojun Song\thanks{Guanghua School of Management, Peking University. Email: \texttt{sxj@gsm.pku.edu.cn}. This work was supported by the National Natural Science Foundation of China [Grant Numbers 72373007 and 72333001]. The author also gratefully acknowledges the research support from the Center for Statistical Science of Peking University.}\\
	\and Jichao Yuan\thanks{Guanghua School of Management, Peking University. Email: \texttt{2201111030@stu.pku.edu.cn}.}
	}

\maketitle

\begin{abstract}
In this paper, we propose new specification tests for regression models with measurement errors in the explanatory variables. Inspired by the integrated conditional moment (ICM) approach, we use a deconvoluted residual-marked empirical process and construct ICM-type test statistics based on it. The issue of measurement errors is addressed by applying a deconvolution kernel estimator in constructing the residuals. We demonstrate that employing an orthogonal projection onto the tangent space of nuisance parameters not only eliminates the parameter estimation effect but also facilitates the simulation of critical values via a computationally simple 
multiplier bootstrap procedure. 
It is the first time a multiplier bootstrap has been proposed in the literature of specification testing with measurement errors. We also develop specification tests and the multiplier bootstrap procedure when the measurement error distribution is unknown. The finite-sample performance of the proposed tests for both known and unknown measurement error distributions is evaluated through Monte Carlo simulations, which demonstrate their efficacy.  
\end{abstract}

\newpage

\doublespacing

\section{Introduction}
There already exists a substantial body of literature on specification testing for regression models, see, e.g., \cite{bierens1982consistent}, \cite{hardle1993comparing}, \cite{JOHNXUZHENG1996263}, \cite{stute1997nonparametric}, \cite{stute1998bootstrap}, \cite{stinchcombe1998consistent}, \cite{stute2002model}, \cite{zhu2003testing}, \cite{zhu2005testing}, \cite{Escanciano_2006}, \cite{hall2007testing}, \cite{song2008model}, \cite{xu2015nonparametric}, \cite{sant2019specification}, \cite{otsu2021specification}, and \cite{TAN2025106113}. The list is undoubtedly not exhaustive. 
However, among these, relatively few studies have adequately developed specification tests in the presence of measurement errors. Measurement error is a common issue in data across many disciplines, including economics, finance, and medicine, see, e.g., \cite{fuller2009measurement}, \cite{alexander2009deconvolution}, and \cite{hu2017measurement}. Unfortunately, specification tests often have incorrect size and low power in the presence of measurement errors.


In an earlier influential work, \cite{otsu2021specification} proposed a specification test for regression models with measurement errors in the explanatory variables based on nonparametric local smoothing estimators. Their test belongs to the class of local smoothing tests, and so does the minimum distance test proposed by \cite{song2008model}. These tests are typically constructed from the distance between the fitted nonparametric regression and the parametric fit under the null hypothesis, and have been shown to be more powerful against high-frequency alternatives. Complementary to such approaches, global smoothing tests, developed in \cite{hall2007testing}, exhibit higher power against low-frequency alternatives. As emphasized in \cite{otsu2021specification}, local smoothing and global smoothing tests serve as complements rather than substitutes. A similar complementarity in power properties between the two types of tests, in the absence of measurement errors, was also established in \cite{fan2000consistent}.

Within the class of global smoothing tests, the integrated conditional moment (ICM)-type tests proposed by \cite{bierens1982consistent} and \cite{bierens1997asymptotic} are particularly important and are the main focus of this paper. These tests typically compare the integrated regression function with the integrated parametric regression function under the null hypothesis. Specifically, a residual-based empirical process is employed to construct the test statistics, and, in general, the asymptotic null distributions depend on the data-generating process, thereby motivating the use of bootstrap techniques to obtain the critical values. The ICM-type tests have been widely applied to specification testing when measurement errors are absent, as shown in \cite{su1991lack}, \cite{stute1997nonparametric}, and \cite{Escanciano_2006}, among others. However, in the presence of measurement errors, the ICM-type tests face challenges arising from the unobservability of the true residuals, the stringent requirements for the estimators, and the computational complexity of obtaining critical values, which motivate the present study to address these issues. In this paper, we propose ICM-type tests based on a deconvoluted residual-marked empirical process (i.e., the residuals are constructed using the deconvolution kernel estimator). We employ an orthogonal projection onto the tangent space of nuisance parameters to eliminate the parameter estimation effect. As a desirable byproduct, the projection facilitates the simulation of critical values via a computationally straightforward multiplier bootstrap. 


It is important to highlight our contribution regarding the simple multiplier bootstrap procedure for obtaining critical values. Conventional bootstrap methods face serious challenges in the presence of measurement errors because the true regressors are unobservable; consequently, we cannot resample them. Infeasibility of obtaining the critical values through the multiplier bootstrap procedure is due to the parameter estimation effect, 
initially discussed in \cite{durbin1973distribution}. We introduce a novel projection to address the parameter estimation effect by imposing an orthogonality condition on the weight function of the residual-marked empirical process. The introduction of the projection also facilitates a computationally attractive multiplier bootstrap procedure to implement our tests, whose asymptotic validity can be easily justified. Notably, compared with other bootstrap methods, the multiplier bootstrap is simpler to justify theoretically, easier to implement, and computationally more efficient. To our knowledge, this is the first time a multiplier bootstrap has been proposed in the literature of specification testing in the presence of measurement errors. Furthermore, we establish a comprehensive theoretical framework and a multiplier bootstrap-based procedure for obtaining critical values when the measurement error distribution is unknown. In addition, our tests are more robust to bandwidth choices than the local smoothing approach. 


The rest of this paper is organized as follows. In Section \ref{sec.Test}, we describe the testing framework and introduce our projection method. Section \ref{sec.Asy} discusses the asymptotic properties of the test statistics. Section \ref{sec.unknown ME} addresses the case of unknown measurement error distribution. In Section \ref{sec.boot}, we provide a multiplier bootstrap procedure and its implementation. Monte Carlo simulations are conducted in Section \ref{sec.Simulation}. We present our conclusions in Section \ref{sec.Conclusion}. Additional simulation results and the mathematical proofs are included in the online supplementary appendix.

\section{The testing framework}\label{sec.Test}

\subsection{The testing problem}

Let $(Y,X)^\top$ be a random vector in a two-dimensional Euclidean space, where $Y$ is an observed real-valued response variable and $X$ is an unobservable error-free explanatory variable. Consider the following regression model:
\begin{align}\label{mod.regression model}
Y=m(X)+U, \quad \text{with} \quad\mathbb{E}[U|X]=0\,\,\,\text{almost surely }(a.s.),
\end{align}
where $m(X)=\mathbb{E}[Y|X]$ is the regression function of $Y$ given $X$ and $U$ is the unpredictable part of $Y$. While the true explanatory variable $X$ is not directly observed, we instead observe a noisy variable $W$ through 
\begin{align}\label{ME}
W=X+\epsilon, \quad \text{with} \quad \mathbb E[\epsilon]=0,
\end{align}
where the measurement error $\epsilon\in\mathbb{R}$ is assumed to be independent of $X$ and $Y$. In addition, we assume that the density function of $\epsilon$ is known and denoted by $f_\epsilon$. The case of unknown $f_\epsilon$ but with repeated measurements available is discussed in Section \ref{sec.unknown ME}. Suppose researchers consider using the following parametric regression model:
\begin{align}\label{mod.parametric regression model}
Y=g(X;\theta)+e(\theta),
\end{align}
where $g(X;\theta)$ is the parametric specification for the regression function $m(X)$ and $e(\theta)$ is the parametric disturbance of the model. In both econometrics and statistics, it is crucial to assess the adequacy of the putative parametric model $g(X;\theta)$. That is, our null hypothesis of interest is
\begin{align}\label{hyp.null1}
H_0: \mathbb{P}[m(X)=g(X;\theta_0)]=1 \,\,\,\text{for some}\,\,\,\theta_0\in\Theta\subset\mathbb{R}^d.
\end{align}
The alternative hypothesis is the negation of $H_0$, i.e.,
\begin{align}\label{hyp.alternative1}
H_1: \mathbb{P}[m(X)=g(X;\theta)]<1 \,\,\,\text{for any}\,\,\,\theta\in\Theta\subset\mathbb{R}^d.
\end{align}
Clearly, testing the null hypothesis in \eqref{hyp.null1} is equivalent to testing
\begin{align}\label{hyp.null2}
H_0: \mathbb{E}[e(\theta_0)|X]=0\,\,\,a.s.\,\,\text{for some}\,\,\,\theta_0\in\Theta\subset\mathbb{R}^d.
\end{align}
Note that \eqref{hyp.null2} is a standard conditional moment restriction. It is well known that \eqref{hyp.null2} can be equivalently expressed as a continuum number of unconditional moment restrictions using the exponential weight function, see, for example, \cite{bierens1982consistent}, \cite{bierens1990consistent}, and \cite{bierens1997asymptotic}. Specifically, we can rewrite \eqref{hyp.null2} as follows:
\begin{align}\label{null eq1}
S(\xi,\theta_0)=\mathbb{E}\left[e(\theta_0)\ee^{\ii X\xi}\right]=0\,\,\,\forall\,\xi\in\Pi\subset\mathbb{R}\,\,\,\text{for some}\,\,\,\theta_0\in\Theta\subset\mathbb{R}^d,
\end{align}
where $e(\theta)=Y-g(X;\theta)$ is the parametric error, $\Pi$ is a properly chosen compact set with nonempty interior, and $\ii=\sqrt{-1}$ denotes the imaginary unit. 

To provide evidence of model misspecification, we can then compare a suitable estimator for $S(\xi,\theta_0)$ in \eqref{null eq1} with the zero function. Unfortunately, the classical residual-marked empirical process is clearly infeasible, as the true regressor $X_i$ is unobservable due to the measurement error $\epsilon$ in \eqref{ME}. 
Motivated by \cite{dong2022nonparametric}, however, $S(\xi,\theta_0)$ can still be consistently estimated by noting that
\begin{align*}
S(\xi,\theta_0)=\iint\left(y-g(x;\theta_0)\right)f_{Y,X}(y,x)\ee^{\ii x\xi}\,dy\,dx, 
\end{align*}
where $f_{Y,X}(y,x)$ denotes the unknown joint density function of $\left(Y,X\right)^\top$. Given a random sample $\{(Y_i,W_i)^\top\}_{i=1}^n$ of size $n\geq 1$, and motivated by the deconvolution methods developed for density estimation [see, e.g., \cite{carroll1988optimal} and \cite{stefanski1990deconvolving}], we replace $f_{Y,X}(y,x)$ by a nonparametric estimator and employ a consistent estimator $\hat{\theta}_n$ for $\theta_0$. Then, $S(\xi,\theta_0)$ is estimated by
\begin{align*}
S_{n}(\xi,\hat{\theta}_n)=\iint\left(y-g(x;\hat{\theta}_n)\right)\hat{f}_{Y,X}(y,x)\ee^{\ii x\xi}\,dy\,dx,
\end{align*}
where $\hat f_{Y,X}(y,x)$ is the deconvolution kernel density estimator for $f_{Y,X}(y,x)$, i.e.,
\begin{align*}
\hat{f}_{Y,X}(y,x)=\frac{1}{n}\sum_{i=1}^nK_b\left(\frac{y-Y_i}{b}\right)\mathcal{K}_b\left(\frac{x-W_i}{b}\right),
\end{align*}
with $K_b(a)=K(a)/b$, where $K(\cdot)$ is a symmetric kernel function and $b=b_n\in\mathbb{R}^+$ is a sequence of bandwidth parameters shrinking to zero at an appropriate rate specified later. In addition, $\mathcal{K}_b(a)$ is the univariate deconvolution kernel given by
\begin{align}\label{conv}
\mathcal{K}_b(a)=\frac{1}{2\pi b}\int \ee^{-\ii t a}\frac{K^{\text{ft}}(t)}{f_\epsilon^{\text{ft}}(t/b)}\,dt, 
\end{align}
in which $K^{\text{ft}}(\cdot)$ and $f_\epsilon^{\text{ft}}(\cdot)$ are the Fourier transforms of the kernel function $K(\cdot)$ and the measurement error density $f_\epsilon(\cdot)$, respectively. The above deconvolution kernel density estimator has been widely used to construct density and nonparametric regression estimators in the presence of measurement errors; see, e.g., \cite{alexander2009deconvolution}. Such an approach is particularly important when kernel smoothers are employed, as shown in \cite{otsu2021specification}. We note that, following this approach, $S_{n}(\xi,\hat{\theta}_n)$ can be further rewritten as
\begin{align}
S_{n}(\xi,\hat{\theta}_n)=&\frac{1}{n}\sum_{i=1}^n\int\left[\int\left(y-g(x;\hat{\theta}_n)\right)K_b\left(\frac{y-Y_i}{b}\right)\,dy\right]\mathcal{K}_b\left(\frac{x-W_i}{b}\right)\ee^{\ii x\xi}\,dx\notag\\
=&\frac{1}{n}\sum_{i=1}^n\int\left(Y_i-g(x;\hat{\theta}_n)\right)\mathcal{K}_b\left(\frac{x-W_i}{b}\right)\ee^{\ii x\xi}\,dx. \label{interm}
\end{align}
It can be shown that under the null hypothesis and mild regularity conditions on the smoothness of $g(x;\theta)$ with respect to $\theta$, the smoothness of $f_\epsilon$ and kernel function $K(\cdot)$, the $\sqrt n$-consistency assumption on the estimator $\hat\theta_n$ as well as the bandwidth $b$, 
\begin{align}
\sqrt{n}S_{n}(\xi,\hat{\theta}_n)=&\frac{1}{\sqrt{n}}\sum_{i=1}^n\int\left(Y_i-g(x;\theta_0)\right)\mathcal{K}_b\left(\frac{x-W_i}{b}\right)\ee^{\ii x\xi}\,dx\notag\\
&-\sqrt{n}\left(\hat{\theta}-\theta_0\right)^\top G(\xi,\theta_0)+o_p(1), \label{no_proj_decom}
\end{align}
uniformly in $\xi\in\Pi$, where, 
\begin{align}\label{classical_parametric effect}
G(\xi,\theta_0)=\mathbb{E}\left[\dot{g}(X;\theta_0)\ee^{\ii X\xi}\right]:=\int\dot{g}(x;\theta_0)\ee^{\ii x\xi}f_X(x)\,dx, 
\end{align}
with $f_X(x)$ denoting the density function of $X$. Here and below, a dot denotes differentiation with respect to the variable $\theta$, i.e., $\dot{g}(x;\theta)=\partial g(x;\theta)/\partial\theta$. 

However, the uniform asymptotic decomposition obtained in \eqref{no_proj_decom} requires the $\sqrt n$-consistency property of $\hat\theta_n$, which may not be satisfied in the measurement error context. It is known that $\hat\theta_n$ may have a lower convergence rate, for example, $\hat\theta_n-\theta_0=O_p(n^{\delta-1/2})$ for some $0\leq\delta<1/4$ in parametric models with measurement errors, see \cite{taupin1998estimation}. Even $\hat\theta_n$ is $\sqrt n$-consistent, due to the presence of $\sqrt{n}(\hat{\theta}_n-\theta_0)^\top G(\xi,\theta_0)$ (commonly known as the parameter estimation effect) in the decomposition of $\sqrt{n}S_n(\xi,\hat{\theta}_n)$, the asymptotic null distribution of $\sqrt{n}S_n(\xi,\hat{\theta}_n)$ depends on $\hat\theta_n$ and the computationally attractive multiplier bootstrap is also infeasible. A typical way to deal with the parameter estimation effect is to impose the asymptotically linear representation assumption of $\sqrt n(\hat\theta_n-\theta_0)$ as follows:
\begin{align}\label{param expansion}
\sqrt{n}\left(\hat{\theta}_n-\theta_0\right)=\frac{1}{\sqrt n}\sum_{i=1}^nl(Y_i,W_i;\theta_0)+o_p(1),
\end{align}
for some function $l(Y_i,W_i;\theta_0)$ with zero mean and finite variance. However, such an asymptotically linear representation is 
even more difficult to obtain in the context of measurement errors.

\subsection{A projection approach}
In this section, we propose a class of orthogonal projection-based test statistics with attractive theoretical and empirical properties. This approach has been used in the literature to address parameter estimation effects in various testing problems; for example, specification analysis of linear quantile models in \cite{escanciano2014specification}, specification tests for the propensity score in \cite{sant2019specification}, model checking in partially linear spatial autoregressive models in \cite{yang2024model}, and testing for the nonparametric component in 
partially linear quantile regression models in \cite{song2025unified}. It is worth emphasizing that the projection approach proposed here is particularly suitable for regression models with measurement errors, as it does not require assuming an asymptotically linear representation of $\sqrt n(\hat\theta_n-\theta_0)$, or even the $\sqrt n$-consistency of $\hat\theta_n$ to $\theta_0$. Specifically, we introduce a projection-based weight to eliminate the parameter estimation effect in \eqref{no_proj_decom}, thereby relaxing the requirement on the estimator $\hat\theta_n$ and preserving the parametric convergence advantage of our ICM-type tests relative to local-smoothing-based methods. To be precise, $n^{\delta}(\hat\theta_n-\theta_0)$ for some $\delta>1/4$ would suffice. Consequently, a computationally attractive multiplier bootstrap procedure is facilitated, which is particularly advantageous in the presence of measurement errors.

To motivate our approach, we consider a projection-based transformation of the weight function $\ee^{\ii x\xi}$ given by
\begin{align}\label{projection distribution}
\mathcal{P}(x;\xi,\theta)=\ee^{\ii x\xi}-\dot g^\top(x;\theta)\Delta^{-1}(\theta)G(\xi,\theta),  
\end{align}
with
\begin{align*}
G(\xi,\theta )=\mathbb{E}[\dot g(X;\theta )\ee^{\ii X\xi}]\quad\text{and}\quad\Delta(\theta) =\mathbb{E}[\dot g(X;\theta)\dot g^\top(X;\theta)].
\end{align*}
A similar modification of the weight function can also be found in \cite{escanciano2014specification} and \cite{sant2019specification}, among others. The intuition behind \eqref{projection distribution} is that $\Delta^{-1} \left( \theta \right)G\left(\xi,\theta \right) $ represents the vector of linear projection coefficients of regressing $\ee^{\ii X\xi} $ on $\dot g(X;\theta)$. Thus, it follows that $\dot g(X;\theta )^\top\Delta^{-1} \left( \theta \right) G\left(\xi,\theta
\right) $ is the best linear predictor of $\ee^{\ii X\xi} $ given $\dot g(X;\theta )$ and 
\begin{align*}
\mathbb{E}\left[\dot g(X;\theta )\mathcal{P}(X;\xi,\theta)\right] & =\mathbb{E}\left[\dot g(X;\theta )\left(\ee^{\ii X\xi}-\dot g^\top (X,\theta
)\Delta^{-1} \left( \theta \right) G\left(\xi,\theta \right) %
\right)\right] \\
& =G(\xi,\theta )-\Delta \left( \theta \right) \Delta^{-1} \left( \theta
\right) G\left(\xi,\theta \right) =0.
\end{align*}
Based on the properties mentioned above, our tests are then based on continuous functionals of the following feasible projection-based deconvoluted residual-marked empirical process,
\begin{align}\label{Test stat}
S_{n}^{pro}(\xi,\hat\theta_n)=\frac{1}{n}\sum_{i=1}^n\int\left(Y_i-g(x;\hat{\theta}_n)\right)\mathcal{K}_b\left(\frac{x-W_i}{b}\right)\mathcal{P}_n(x;\xi,\hat{\theta}_n)\,dx,
\end{align}
where $\hat{\theta}_n$ is a suitably consistent estimator for $\theta_{0}$ under the null hypothesis with required convergence rates specified later (not necessarily $\sqrt n$-consistent), and $\mathcal{P}_n(x;\xi,\theta)$ is the sample analog of projection $\mathcal{P}(x;\xi,\theta)$ in (\ref{projection distribution}); namely,
\begin{align*}
\mathcal{P}_n(x;\xi,\hat{\theta}_n)=\ee^{\ii x\xi}
-\dot g^\top(x;\hat\theta_n)\Delta_n^{-1}(\hat\theta_n)G_{n}(\xi,\hat\theta_n),
\end{align*}
where
\begin{equation*}
G_{n}(\xi,\theta)=\int\dot g(x;\theta)\hat f_X(x)\ee^{\ii x\xi}\,dx \quad\text{and} \quad \Delta_n(\theta)=\int\dot g(x;\theta)\dot g^\top(x;\theta )\hat f_X(x)\,dx,
\end{equation*}
with
\begin{equation}
\hat f_X(x)=\frac{1}{n}\sum_{i=1}^n\mathcal{K}_b\left(\frac{x-W_i}{b}\right)\label{decon density}
\end{equation}
the deconvolution kernel density estimator for $f_X(x)$, where $\mathcal{K}_b(\cdot)$ is defined in \eqref{conv}. 

With the assistance of \eqref{Test stat}, we can establish that under the null hypothesis of correct specification, the empirical process $S_{n}^{pro}(\cdot,\hat{\theta}_n)$ is expected to be close to the zero function. Under the alternative hypothesis of misspecification, $S_{n}^{pro}(\cdot,\hat{\theta}_n)$ tends to deviate from the zero function, and consequently $\sqrt nS_{n}^{pro}(\cdot,\hat{\theta}_n)$ diverges. It is therefore natural to construct the test statistic by measuring an appropriate distance between $S^{pro}_{n}(\cdot,\hat{\theta}_n)$ and the zero function, denoted by $\Gamma(S^{pro}_{n})$ for some continuous functional $\Gamma(\cdot)$. For example, we could consider the following two test statistics based on the popular Kolmogorov--Smirnov (KS) and Cram\'{e}r--von Mises (CvM) functionals, respectively,
\begin{align*}
KS_{n}=\sup_{\xi\in\Pi}\left\vert S^{pro}_{n}(\xi,\hat{\theta}_n)\right\vert
 \quad\text{and}\quad CvM_{n}=\int_\Pi\left\vert S^{pro}_{n}(\xi,\hat{\theta}_n)\right\vert^2\,d\xi.
\end{align*}
In constructing the $CvM_{n}$ statistic, we adopt the uniform integrating measure on $\Pi$, which is also recommended in \cite{dong2022nonparametric}. Previous studies have discussed the use of integrating measures that are absolutely continuous with respect to the Lebesgue measure on $\Pi$. However, due to the difficulty of selecting an appropriate measure and the estimation challenges that arise when the true regressor $X$ is unobserved, such measures have rarely been employed in models with measurement errors. The test statistics $KS_{n}$ and $CvM_{n}$ should be small if the null hypothesis \eqref{hyp.null1} is true, while \textquotedblleft
large\textquotedblright{}\ values of $KS_{n}$ and $CvM_{n}$ imply the rejection of $H_{0}$. These \textquotedblleft
large\textquotedblright{}\ values will be determined by a convenient multiplier bootstrap procedure described in Section \ref{sec.boot}, thanks to the projection $\mathcal{P}_n(x;\xi,\hat{\theta}_n)$ that eliminates the parameter estimation effect.

\section{Asymptotic theory}\label{sec.Asy}

In this section, we establish the asymptotic distributions of the test statistics $KS_n$ and $CvM_n$ under the null hypothesis $H_0$. Subsequently, we investigate the asymptotic power behavior of these statistics under a sequence of local alternative hypotheses $H_{1n}$ that converge to $H_0$ at the parametric rate $n^{-1/2}$, as well as under the fixed alternative. Denote $f^{ft}(t)=\int \ee^{\ii tx}f(x)dx$ for a generic function $f$ and let $g^{(p)}(x;\theta)$ denote the $p$-times derivative of function $g$ with respect to variable $x$. Throughout this paper, $|c|$ is used to denote the Euclidean norm of a vector $c$ and $|A|$ is used to denote the Frobenius norm of a matrix $A$.

\subsection{Asymptotic null distributions}

We first impose some regularity conditions to derive the asymptotic null distributions of the test statistics $KS_n$ and $CvM_n$.

\begin{assumption}\label{ass.D}

    \quad 

    \begin{enumerate}[label=(\roman*)]
	\item $\{(Y_i,W_i)^\top\}_{i=1}^n$ is an i.i.d sample of $(Y,W)^\top$ which satisfies \eqref{mod.regression model}, \eqref{ME}, and $\mathbb E|Y|^2<\infty$.
        
    \item The estimator $\hat\theta_n$ satisfies $\hat\theta_n-\theta_0=O_p(n^{\delta-1/2})$ for some $0\leq \delta<1/4$ in parametric regression models with measurement errors.

	\item The measurement error $\epsilon$ is independent of $(Y,X)^\top$.
    \end{enumerate}

\end{assumption}
Assumption \ref{ass.D}({\romannumeral1}) requires random sampling and the existence of the second moment of $Y$, both of which are frequently mentioned in the literature. We emphasize that Assumption \ref{ass.D}({\romannumeral2}) relaxes the requirement on the convergence rate of the estimator $\hat\theta_n$  used in our tests, making it more suitable for cases involving measurement errors. Most estimators cannot typically be expanded as \eqref{param expansion} detailed in Section \ref{sec.Test} in the presence of measurement errors. Even worse, they cannot achieve $\sqrt{n}$-rate convergence without imposing numerous restrictions [see \cite{taupin1998estimation}], making the multiplier bootstrap infeasible. Assumption \ref{ass.D}({\romannumeral3}) is common and essential, as mentioned in the classical literature on measurement error [see \cite{otsu2021specification}].

As shown in the classical measurement error literature, we categorize two separate cases by the decay rate of the tail of the characteristic function of the measurement errors: the ordinary smooth case and the supersmooth case. We first focus on the ordinary smooth case and impose the following assumptions.

\begin{assumption}\label{ass.O}

    \quad

    \begin{enumerate}[label=(\roman*)]
	\item The functions $f_X(x)$, $g(x;\theta_0)$, and $h(x;\theta_0)$ are $p$-times continuously differentiable with respect to $x$ with bounded and integrable derivatives, where $h(x;\theta)$ denotes any partial derivative of $g(x;\theta)$ with respect to $\theta$ of order up to three and $p$ is a positive integer satisfying $p>\alpha$. Furthermore, we also impose additional assumption about Lipschitz continuous properties of $f_X^{(p)}(x)$, $\left[g(x;\theta_0)f_X(x)\right]^{(p)}$, $h^{(p)}(x;\theta_0)$, and $\left[g(x;\theta_0)h(x;\theta_0)\right]^{(p)}$ for almost every $x$ as follows:
    \begin{align*}
        & \left\vert f_X^{(p)}(x+y) - f_X^{(p)}(x)\right\vert \leq L_{f_X^{(p)}}(x)\vert y\vert,\\
        & \left\vert \left[g(x+y;\theta_0)f_X(x+y)\right]^{(p)} - \left[g(x;\theta_0)f_X(x)\right]^{(p)}\right\vert \leq L_{[g f_X]^{(p)}}(x)\vert y\vert\\
        & \left\vert h^{(p)}(x+y;\tilde{\theta}) - h^{(p)}(x;\tilde{\theta})\right\vert \leq L_{h^{(p)}}(x)\vert y\vert,\\
        & \left\vert \left[g h\right]^{(p)}(x+y;\tilde{\theta}) - \left[g h\right]^{(p)}(x;\tilde{\theta})\right\vert \leq L_{[g h]^{(p)}}(x)\vert y\vert,
    \end{align*}
    for some bounded and integrable functions $L_{f_X^{(p)}}(x)$, $L_{[gf_X]^{(p)}}(x)$, $L_{h^{(p)}}(x)$, and $L_{[gh]^{(p)}}(x)$, where $\tilde{\theta}$ takes values between $\theta_0$ and $\hat{\theta}_n$.

	\item The characteristic function of measurement error $\epsilon$ is of the following form for all $t\in\mathbb{R}$, where $c_0^{os},c_1^{os},\cdots,c_{\alpha}^{os}$ are finite constants with $c_0^{os} = 1$ and $\alpha>0$,
    \begin{align*}
        f_\epsilon^{ft}(t) = \frac{1}{c_0^{os}+c_{1}^{os}t+\cdots+c_{\alpha}^{os}t^{\alpha}}.
    \end{align*}

        \item The kernel function $K$ is differentiable to order $p+1$ and satisfies the following equations for $l=1,2,\cdots,p-1$:
    \begin{align*}
        &\int K(u)du = 1, \qquad \int u^{p}K(u)du\neq 0, \qquad \int u^lK(u)du = 0.
    \end{align*}
    In addition, $K^{ft}$ is compactly supported on a compact set that is symmetric around zero.

    \item $nb^{2p}\to 0$ as $n\to\infty$.

    \item For $c_l^{os}(\xi) = (-\ii)^l\sum\limits_{j=l}^{\alpha}c_j^{os}\binom{j}{l}\xi^{j-l}$, we have
    \begin{align*}
        &\mathbb{E}\left[\sup\limits_{\xi\in\Pi}\left\vert r^{os}_{h, \infty}(Y,W;\xi,\theta_0)\right\vert^2\right]<\infty,
    \end{align*}
    where
    \begin{align*}
        &r^{os}_{h, \infty}(Y,W;\xi,\theta_0) = \sum_{l=0}^\alpha c_l^{os}(\xi) \left[Yh^{(l)}(W;\theta_0)\right] - \sum_{l=0}^\alpha c_l^{os}(\xi) \left[(gh)^{(l)}(W;\theta_0)\right].
    \end{align*}
    \end{enumerate}
\end{assumption}

Assumption \ref{ass.O}({\romannumeral1}) imposes both requirements of the Lipschitz continuity and restrictions on the smoothness of the structural functions, mainly adopted from \cite{dong2022nonparametric}. The Lipschitz continuity assumptions are necessary to analyze the moment properties of structural functions under ordinary smooth conditions, as detailed in the proofs of the lemmas. Our assumptions are slightly stronger, requiring the Lipschitz continuity of $h^{(p)}$ and $[gh]^{(p)}$, enhancements that are essential for ensuring convergence. These assumptions hold in commonly used scenarios, such as polynomial $g$, uniformly distributed $X$, normally distributed $U$, and Laplace distributed $\epsilon$. Assumption \ref{ass.O}({\romannumeral2}) is commonly known as the ordinary smooth assumption, which specifies the polynomial decay rate of $f_\epsilon^{ft}$, the characteristic function of measurement error. As shown in \cite{fan1995average} and \cite{dong2022nonparametric}, it can be generalized to $f_\epsilon^{ft}(t)=\exp(it\zeta)/(c_0^{os}+c_{1}^{os}t+\cdots+c_{\alpha}^{os}t^{\alpha})$, including typical distributions such as the Laplace distribution. The characteristic function of the measurement error is estimated through repeated sampling when it is unknown, as outlined in Section \ref{sec.unknown ME}. Assumption \ref{ass.O}({\romannumeral3}) concerns the order of the kernel function, which is essential for constructing the deconvolution kernel, ensuring its integration properties, and eliminating the asymptotic bias introduced by nonparametric estimators. Moreover, the construction of higher-order kernel functions is feasible as discussed in \cite{alexander2009deconvolution}. Assumption \ref{ass.O}({\romannumeral4}) contains the undersmoothing condition frequently used in the literature. It also guarantees the asymptotic negligibility of the bias from nonparametric estimators and establishes the integration properties of the deconvolution kernel. It is worth noting that, in the context of the specification testing problem considered in this paper, we do not impose a lower bound on the bandwidth to ensure the existence of the variance, as is commonly done in studies using kernel-based estimators. In particular, we verify in Appendix \ref{sec.AppendixC} that, under the Lipschitz continuity of the underlying functions and an undersmoothing bandwidth condition, the use of kernel smoothing does not introduce any additional non-negligible variance. Consequently, it suffices to assume the boundedness of the asymptotic variance to guarantee the existence of the limiting process, as stated in Assumption \ref{ass.O}({\romannumeral5}). Similar assumptions are made in \cite{fan1995average} and \cite{dong2022nonparametric}.

Using the assumptions above, we can characterize the limiting behavior of $S_{n}^{pro}(\cdot,\hat\theta_n)$ for the ordinary smooth case. Let \textquotedblleft $\Longrightarrow$ \textquotedblright{} denote weak convergence on $(l^{\infty}(\Pi),\mathcal{B}_{\infty})$ in the sense of Hoffmann--J{\o}rgensen, where $\mathcal{B}_{\infty}$ denotes the corresponding Borel $\sigma$-algebra, see, e.g., Definition 1.3.3 in \cite{van1996weak}. Subsequently, we derive the asymptotic null distributions of the statistics $KS_n$ and $CvM_n$, as stated in the following theorem and corollary. 

\begin{theorem}\label{theorem.ordinary smooth under H_0}
    Suppose that Assumptions \ref{ass.D} and \ref{ass.O} hold. Then, under the null hypothesis $H_0$ in \eqref{hyp.null1}, we have 
    \begin{align}\label{result of param effect ass O H0}
        &\sup_{\xi\in\Pi}\left\vert S_{n}^{pro}(\xi,\hat\theta_n)- S_{n}^{pro}(\xi,\theta_0)\right\vert=o_p\left(n^{-\frac{1}{2}}\right).
    \end{align}
    Furthermore, 
    \begin{align}\label{result of main term ass O H0}
        \sqrt n S_{n}^{pro}(\cdot,\hat\theta_n)\Longrightarrow S_{\infty}^{os}(\cdot,\theta_0),
     \end{align}
    where $S_{\infty}^{os}(\cdot,\theta_0)$ is a Gaussian process with mean zero and covariance structure
    \begin{align*}
        &Cov\left[S_{\infty}^{os}(\xi,\theta_0),S_{\infty}^{os}(\xi^\prime,\theta_0)\right] = \mathbb{E}\left[r^{os}_{\infty}(Y,W;\xi,\theta_0)\overline{r^{os}_{\infty}}(Y,W;\xi^\prime,\theta_0)\right],
    \end{align*}
    with 
    \begin{align}\label{theorem.ordinary smooth under H_0 ros}
        &r^{os}_{\infty}(Y,W;\xi,\theta_0) = c_0^{os}(\xi)Y\ee^{\ii W\xi} - \sum\limits_{l=0}^{\alpha} c_l^{os}(\xi) [\ee^{\ii W\xi} g^{(l)}(W;\theta_0)] \notag\\
        & - \left\{ \sum\limits_{l=0}^\alpha c_l^{os}(0) \left[Y\dot{g}^{(l)}(W;\theta_0)\right] - \sum\limits_{l=0}^\alpha c_l^{os}(0) \left[(g\dot{g})^{(l)}(W;\theta_0)\right] \right\}^\top\Delta^{-1}(\theta_0)G(\xi,\theta_0).
    \end{align}
\end{theorem}

Based on Theorem \ref{theorem.ordinary smooth under H_0} and the continuous mapping theorem, see, e.g., \cite{van1996weak}, we can further derive the asymptotic null distributions of $KS_n$ and $CvM_n$ for the ordinary smooth case. 

\begin{corollary}\label{Corollary.known ordinary}
Suppose that Assumptions \ref{ass.D} and \ref{ass.O} hold. Then, under the null hypothesis $H_0$ in \eqref{hyp.null1}, we have 
\begin{align*}
 \sqrt{n}KS_{n} \stackrel{d}\longrightarrow \sup\limits_{\xi\in\Pi}\left\vert S_{\infty}^{os}(\xi,\theta_0)\right\vert  \quad\text{ and }
    \quad nCvM_{n}\stackrel{d}\longrightarrow \int_{\Pi}\left\vert S_{\infty}^{os}(\xi,\theta_0)\right\vert^2\,d\xi.
\end{align*}
\end{corollary}

Theorem \ref{theorem.ordinary smooth under H_0} shows that $S_{n}^{pro}(\cdot,\hat\theta_n)$ converges weakly to a centered Gaussian process with its covariance structure depending on the data-generating process (DGP) and the parametric model $g(x;\theta_0)$ under the null. Subsequently, Corollary \ref{Corollary.known ordinary} shows that the  $KS_n$ and $CvM_n$ statistics converge to the sup norm and the squared norm of the aforementioned Gaussian process, respectively. We note that $S_{n}^{pro}(\cdot,\hat\theta_n)$ exhibits $\sqrt{n}$-rate convergence for the ordinary smooth case and is asymptotically independent of the chosen bandwidth $b$, improving the convergence results of local smoothing tests.

For the supersmooth case where the characteristic function of measurement error decays at an exponential rate, e.g., the distribution of $\epsilon$ is normal, regularity conditions underlying the derivation of the asymptotic distribution of our proposed test statistics under the null are given as follows.

\begin{assumption}\label{ass.S}

\quad

\begin{enumerate}[label=(\roman*)]
        \item The functions $f_X(x)$, $g(x;\theta_0)$, and $h(x;\theta_0)$ are infinitely differentiable with respect to $x$, where the function $h(x;\theta)$ is mentioned in Assumption \ref{ass.O}(i).

        \item The measurement error $\epsilon$ follows a Gaussian distribution with the characteristic function of the following form for all $t\in\mathbb{R}$ and some positive constant $\mu$,
        \begin{align*}
            f^{ft}_\epsilon(t) = \ee^{-\mu t^2}.
        \end{align*}

        \item The kernel function $K$ is infinitely differentiable and satisfies the following equations for all $l\in\mathbb{N}$,
        \begin{align*}
            \int K(u)du = 1, \qquad \int u^lK(u)du = 0.
        \end{align*}
        In addition, $K^{ft}$ has a compact support set that is symmetric around zero, and is bounded.

        \item $b\to0$ as $n\to\infty$.

        \item For $c_l^{ss}(\xi) = (-\ii)^l\sum\limits_{j\geq\frac{l}{2}}^{\infty}\frac{\mu^j}{j!}\binom{2j}{l}\xi^{2j-l}$, we have
        \begin{align*}
            &\mathbb{E}\left[\sup\limits_{\xi\in\Pi}\left\vert r^{ss}_{h, \infty}(Y,W;\xi,\theta_0)\right\vert^2\right]<\infty,
        \end{align*}
        where 
        \begin{align*}
            &r^{ss}_{h, \infty}(Y,W;\xi,\theta_0) = \left\{ \sum_{l=0}^\infty c_l^{ss}(\xi) \left[Yh^{(l)}(W;\theta_0)\right] - \sum_{l=0}^\infty c_l^{ss}(\xi) \left[(gh)^{(l)}(W;\theta_0)\right] \right\}e^{{\rm i}W\xi}.
        \end{align*}
    \end{enumerate}

\end{assumption}

Assumption \ref{ass.S}({\romannumeral1}) requires structural functions and the distribution of latent variables to be infinitely smooth. Although assumptions about conditional mean functions are restrictive, commonly used functions, such as polynomials, circular functions, exponentials, and sums or products of such functions, 
satisfy the requirement. Additionally, the construction of the class of infinitely differentiable density functions is also mentioned in \cite{alexander2009deconvolution}. Assumption \ref{ass.S}({\romannumeral2}) is a supersmooth condition. Assuming the function $f^{ft}_\epsilon$ has an exponential decay rate further enhances Assumption \ref{ass.O}({\romannumeral2}), with the Gaussian distribution being a typical example. The supersmoothness of measurement error and structural functions necessitate the use of an infinite-order kernel in Assumption \ref{ass.S}({\romannumeral3}) and integral properties to derive the asymptotic behavior of the statistic. Assumption \ref{ass.S}({\romannumeral4}) is simply a trivial bandwidth requirement, because the infinite-order smoothness condition removes the bias of the proposed empirical process under the null hypothesis, so that the undersmoothing condition required in Assumption \ref{ass.O}({\romannumeral4}) is no longer needed. Assumption \ref{ass.S}({\romannumeral5}) pertains to the boundedness of the asymptotic variance of our statistics, analogous to Assumption \ref{ass.O}({\romannumeral5}). 

Using these assumptions, we derive the asymptotic behavior of the empirical process $S_{n}^{pro}(\cdot,\hat\theta_n)$ for the supersmooth case in the following theorem.

\begin{theorem}\label{theorem.supersmooth under H_0}

Suppose that Assumptions \ref{ass.D} and \ref{ass.S} hold. Then, under the null hypothesis $H_0$ in \eqref{hyp.null1}, we have \eqref{result of param effect ass O H0} and 
\begin{align}\label{result of main term ass S H0}
\sqrt n S_{n}^{pro}(\cdot,\hat\theta_n)\Longrightarrow S_{\infty}^{ss}(\cdot,\theta_0),
\end{align}
where $S_{\infty}^{ss}(\cdot,\theta_0)$ is a Gaussian process with mean zero and covariance structure
\begin{align*}
    &Cov\left[S_{\infty}^{ss}(\xi,\theta_0),S_{\infty}^{ss}(\xi^\prime,\theta_0)\right] = \mathbb{E}\left[r^{ss}_{\infty}(Y,W;\xi,\theta_0)\overline{r^{ss}_{\infty}}(Y,W;\xi^\prime,\theta_0)\right],
\end{align*}
with 
\begin{align}\label{theorem.supersmooth under H_0 rss}
    &r^{ss}_{\infty}(Y,W;\xi,\theta_0) = c_0^{ss}(\xi)Y\ee^{\ii W\xi} - \sum_{l=0}^{\infty}c_l^{ss}(\xi)g^{(l)}(W;\theta_0)\ee^{\ii W\xi}\notag\\
    - & \left\{Y\sum_{l=0}^{\infty}c_l^{ss}(0)g^{(l)}(W;\theta_0)\ee^{\ii W\xi} - \sum_{l=0}^{\infty}c_l^{ss}(0)\left[g\dot{g}\right]^{(l)}(W;\theta_0)\ee^{\ii W\xi}\right\}^\top\Delta^{-1}(\theta_0)G(\xi,\theta_0).
\end{align}
\end{theorem}

Based on the theorem above, we use the continuous mapping theorem to derive the asymptotic null distributions of $KS_n$ and $CvM_n$ for the supersmooth case.

\begin{corollary}\label{Corollary.known super}

Suppose that Assumptions \ref{ass.D} and \ref{ass.S} hold. Then, under the null hypothesis $H_0$ in \eqref{hyp.null1}, we have
\begin{align*}
\sqrt{n}KS_{n} \stackrel{d}\longrightarrow \sup\limits_{\xi\in\Pi}\left\vert S_{\infty}^{ss}(\xi,\theta_0)\right\vert \quad \text{ and } \quad nCvM_{n} \stackrel{d}\longrightarrow \int_{\Pi}\left\vert S_{\infty}^{ss}(\xi,\theta_0)\right\vert^2\,d\xi.
\end{align*}
\end{corollary}

The projection plays a crucial role in Corollaries \ref{Corollary.known ordinary} and \ref{Corollary.known super}. By eliminating the effect of the estimator $\hat\theta_n$, the proposed test statistics achieve a parametric rate of convergence, even if the convergence rate of $\hat\theta_n$ is typically slower than $\sqrt{n}$ in the presence of measurement errors.

\subsection{Asymptotic power}
In this section, we proceed to derive the power properties of the proposed tests against a sequence of local alternatives. In contrast to \cite{otsu2021specification}, where the power of the local smoothing specification test is analyzed under local alternatives converging to the null hypothesis at a nonparametric rate, our tests have nontrivial power against local alternatives converging at a parametric rate under the following form:
\begin{align}\label{hyp.local alt}
    H_{1n}:m(X) = g(X;\theta_0)+\frac{\Delta(X)}{\sqrt{n}} \quad a.s.
\end{align}
where $\Delta:\mathbb{R}\to\mathbb{R}$ is a bounded nonzero function satisfying the regularity conditions of the function $h$ mentioned in Assumptions \ref{ass.O} and \ref{ass.S}. The asymptotic local power properties of $S_{n}^{pro}(\cdot,\hat\theta_n)$ under $H_{1n}$ are established in the following theorem.

\begin{theorem}\label{theorem.known under H_1n}
Under the sequence of local alternatives $H_{1n}$ in \eqref{hyp.local alt}, we suppose that Assumption \ref{ass.D} holds. If Assumption \ref{ass.O} holds for the ordinary smooth case,
\begin{equation*}
\sqrt n S_{n}^{pro}(\cdot,\hat\theta_n)\Longrightarrow S_{\infty}^{os}(\cdot,\theta_0)+\mu_\Delta(\cdot,\theta_0),
\end{equation*}
and if Assumption \ref{ass.S} holds for the supersmooth case,
\begin{equation*}
    \sqrt n S_{n}^{pro}(\cdot,\hat\theta_n)\Longrightarrow S_{\infty}^{ss}(\cdot,\theta_0)+\mu_\Delta(\cdot,\theta_0),
\end{equation*}
where $S_{\infty}^{os}(\cdot,\theta_0)$ and $S_{\infty}^{ss}(\cdot,\theta_0)$ are the centered Gaussian processes defined in Theorems \ref{theorem.ordinary smooth under H_0} and \ref{theorem.supersmooth under H_0}, respectively, and $\mu_\Delta(\cdot,\theta_0)$ is a deterministic shift process given by
\begin{align*}
    \mu_\Delta(\xi,\theta_0) = \mathbb{E}\left\{\Delta(X)\left[\ee^{\ii X\xi}-\dot{g}^\top(X;\theta_0)\Delta^{-1}(\theta_0)G(\xi,\theta_0)\right]\right\}.
\end{align*}
\end{theorem}

Theorem \ref{theorem.known under H_1n} shows that under $H_{1n}$, both for the ordinary smooth and the supersmooth measurement errors, the asymptotic behavior of the $S_{n}^{pro}(\cdot,\hat\theta_n)$ consists of the limiting Gaussian process under $H_0$ and a deterministic shift term $\mu_\Delta(\cdot,\theta_0)$. By similar arguments to those of the null hypothesis, we then apply the continuous mapping theorem to obtain the convergence in distribution of $KS_n$ and $CvM_n$, which will have nontrivial power against local alternatives as described above, provided that $\mu_\Delta(\xi,\theta_0)\neq 0$ for some $\xi$. 


By denoting $\theta^\ast = \operatorname*{plim}\hat{\theta}_n$ as the pseudo-true value and noting that $\theta^\ast = \theta_0$ under the null or the local alternatives, we can derive the asymptotic global power properties for our test statistics under the fixed alternative.

\begin{theorem}\label{theorem.known alternative}

Under the alternative hypothesis $H_{1}$ in \eqref{hyp.alternative1}, suppose that Assumption \ref{ass.D} holds. Under either Assumption \ref{ass.O} or Assumption \ref{ass.S},
\begin{align*}
    &\sup_{\xi\in\Pi}\left\vert S_{n}^{pro}(\xi,\hat{\theta}_n) - C(\xi,\theta^\ast)\right\vert =  o_p(1),
\end{align*}
where
\begin{align}
    &C(\xi,\theta^\ast) = \mathbb{E}\left\{\left[m(X)-g(X;\theta^\ast)\right]\left[\ee^{\ii X\xi}-\dot{g}^\top(X;\theta^\ast)\Delta^{-1}(\theta^\ast)G(\xi,\theta^\ast)\right]\right\}. \label{Drift}
\end{align}
\end{theorem}

Note that $C(\cdot,\theta^\ast)$ can be understood as the projection of $m(X)-g(X;\theta^\ast)$ onto the orthogonal space of $\dot{g}(X;\theta^\ast)$, as explained in \cite{dominguez2015simple}. Under this interpretation, the consistency of our tests is guaranteed as long as for any vector $\gamma$,
\begin{align}\label{thm.known alt consis cond}
    \mathbb{P}\left[m(X)-g(X;\theta^\ast) = \gamma^\top\dot{g}(X;\theta^\ast)\right]<1.
\end{align}
This ensures $C(\xi,\theta^\ast)\neq 0$ for at least some $\xi$ with a positive Lebesgue measure, thereby leading to the fact that $\sqrt nS_{n}^{pro}(\cdot,\hat{\theta}_n)$ diverges asymptotically and thus $\sqrt nKS_n$ and $nCvM_n$ will diverge to positive infinity in probability. This indicates that any fixed alternative satisfying \eqref{thm.known alt consis cond} can be detected by our proposed test with probability tending to one. Although the tests are not consistent against all possible alternatives, we do not regard those alternatives that violate this condition as a primary empirical concern, since they are nearly observationally equivalent to the null.

\section{Case of unknown measurement error distribution}\label{sec.unknown ME}
Since it is often challenging for researchers to obtain the density function of the measurement error in practical applications, we develop specification tests for settings with an unknown measurement error. As stated in \cite{delaigle2008deconvolution} and \cite{alexander2009deconvolution}, the approach typically used to address this problem is based on additional data, more specifically, repeated measurements on $X$ in the form of
\begin{align}\label{var.repeated ME}
    &W^r = X + \epsilon^r,
\end{align}
where $\epsilon^r$ and $\epsilon$ are identically distributed
and $(X,\epsilon,\epsilon^r)$ are mutually independent. Repeated measurements are used to construct the following consistent estimator for $f^{ft}_\epsilon$, 
\begin{align}\label{est.ME density}
    &\hat{f}^{ft}_\epsilon(t) = \left\vert\frac{1}{n}\sum\limits_{i=1}^n\cos\left[t(W_i-W_i^r)\right]\right\vert^{1/2}
\end{align}
as proposed by \cite{delaigle2008deconvolution} and recommended in \cite{otsu2021specification}. 

Our main contribution is to establish detailed theoretical results for ICM-type specification tests based on the repeated measurements approach described above. The theoretical justification and the practical implementation of the multiplier bootstrap for obtaining critical values are then developed and discussed in Section \ref{sec.boot}. Specifically, we construct the following empirical process based on the aforementioned estimator,
\begin{align}\label{Test stat unknown}
\hat{S}_{n}^{pro}(\xi,\hat\theta_n)=\frac{1}{n}\sum_{i=1}^n\int\left(Y_i-g(x;\hat{\theta}_n)\right)\hat{\mathcal{K}}_b\left(\frac{x-W_i}{b}\right)\hat{\mathcal{P}}_n(x;\xi,\hat{\theta}_n)\,dx,
\end{align}
where $\hat{\mathcal{K}}_b(\cdot)$ and $\hat{\mathcal{P}}_n(\cdot)$ denote the respective counterparts of $\mathcal{K}_b(\cdot)$ and $\mathcal{P}_n(\cdot)$ by replacing $f^{ft}_\epsilon(\cdot)$ with $\hat{f}^{ft}_\epsilon(\cdot)$. Using the empirical process $\hat{S}_{n}^{pro}(\cdot,\hat\theta_n)$ constructed above, the test statistics $CvM_n$ and $KS_n$ introduced earlier can be modified as
\begin{align*}
  \widehat{CvM}_{n}=\int_\Pi\left\vert \hat{S}^{pro}_{n}(\xi,\hat{\theta}_n)\right\vert^2\,d\xi \quad\text{and}\quad \widehat{KS}_{n}=\sup_{\xi\in\Pi}\left\vert \hat{S}^{pro}_{n}(\xi,\hat{\theta}_n)\right\vert.
\end{align*}

Furthermore, to derive the asymptotic properties of the modified statistics, we need to strengthen the assumptions to address the uncertainty introduced by the estimation of the error characteristic functions, by imposing the following additional assumptions in addition to Assumption \ref{ass.D}.
\begin{assumption}\label{ass.D'}

    \quad

    \begin{enumerate}[label=(\roman*)]
	\item $\{W_i^r\}_{i=1}^n$ is an i.i.d sample of $W^r$ satisfying \eqref{var.repeated ME}.

	\item $\epsilon^r$ is i.i.d. as $\epsilon$, $\mathbb{E}\vert\epsilon\vert^{(p+1)(2+\zeta)}<\infty$ for some positive constant $\zeta$, and $f_\epsilon$ is symmetric around zero.
    \end{enumerate}

\end{assumption}
Assumption \ref{ass.D'}({\romannumeral1}) is common in literature, see, e.g., \cite{delaigle2008deconvolution}, requiring independent and identically distributed repeated measurements. 
Assumption \ref{ass.D'}({\romannumeral2}) is used to evaluate the asymptotic properties of the estimator $\hat{f}^{ft}_\epsilon(t)$ in \eqref{est.ME density}, as mentioned in \cite{kurisu2022uniform}. To derive the asymptotic distribution of our statistics under the ordinary smooth case, we need to impose the following assumptions.

\begin{assumption}\label{ass.O'}

    \quad 

    \begin{enumerate}[label=(\roman*)]
        \item $nb^{10\alpha+6}\log(\frac{1}{b})^{-4}\to \infty$ as $n\to\infty$.

        \item Following the notation $r^{os}_{\infty}$ and the function $h(x;\theta)$ defined in Theorem \ref{theorem.ordinary smooth under H_0} and Assumption \ref{ass.O}, respectively,
    \begin{align*}
        &\mathbb{E}\left\{\int_{\Pi} \left[r^{os}_{\infty}(Y,W,\xi,\theta_0)+ r^{\epsilon}_{h,\infty}(Y,W,\xi,\theta_0)\right]^2\,d\xi\right\}<\infty,
    \end{align*}
    where
    \begin{align*}
        &r^{\epsilon}_{h,\infty}(Y,W,\xi,\theta_0) = \int \left[(gf_X)^{ft}(t)h^{ft}(\xi-t) - f_X^{ft}(t)(gh)^{ft}(\xi-t)\right]\Pi_{\epsilon}(t)dt,\\
        &\Pi_{\epsilon}(t) = \frac{1}{2}-\frac{\cos(t(W-W^r))}{2\left\vert f^{ft}_\epsilon(t)\right\vert^2}.
    \end{align*}
    \end{enumerate}
\end{assumption}

In Assumption \ref{ass.O'}({\romannumeral1}), we enhance the assumption about bandwidth compared to Assumption \ref{ass.O}({\romannumeral4}). 
This enhancement is necessary because in the absence of information about the measurement error distribution, the statistic is influenced by the uncertainty in estimating $f^{ft}_\epsilon$. Assumption \ref{ass.O'}({\romannumeral1}) ensures the asymptotic negligibility of such uncertainty brought by the estimator $\hat{f}^{ft}_\epsilon$. Next, compared to our Assumptions in \ref{ass.O}({\romannumeral5}), we add a moment restriction on the term brought by the estimator of the unknown distribution in Assumption \ref{ass.O'}({\romannumeral2}).

With the assumptions above, the following theorem, which characterizes the asymptotic behavior of the proposed empirical process with the unknown characteristic function replaced by its estimator, is established under the ordinary smooth case.
\begin{theorem}\label{theorem.unknown ordinary smooth under H_0}
    
    Suppose that Assumptions \ref{ass.D}, \ref{ass.O}, \ref{ass.D'}, and \ref{ass.O'} hold. Then, under the null hypothesis $H_0$ in \eqref{hyp.null1}, we have 
    \begin{align}\label{result of param effect unknown ass O H0}
        &\sup_{\xi\in\Pi}\left\vert \hat{S}_{n}^{pro}(\xi,\hat\theta_n)- \hat{S}_{n}^{pro}(\xi,\theta_0)\right\vert=o_p\left(n^{-\frac{1}{2}}\right).
    \end{align}
    Furthermore, 
    \begin{align}\label{result of main term unknown ass O H0}
        \sqrt n \hat{S}_{n}^{pro}(\cdot,\hat\theta_n)\Longrightarrow \hat{S}_{\infty}^{os}(\cdot,\theta_0),
     \end{align}
    where $\hat{S}_{\infty}^{os}(\cdot,\theta_0)$ is a Gaussian process with mean zero and covariance structure
    \begin{align*}
        &Cov\left[\hat{S}_{\infty}^{os}(\xi,\theta_0),\hat{S}_{\infty}^{os}(\xi^\prime,\theta_0)\right] = \mathbb{E}\left[\hat{r}^{os}_{\infty}(Y,W;\xi,\theta_0)\overline{\hat{r}^{os}_{\infty}}(Y,W;\xi^\prime,\theta_0)\right],
    \end{align*}
    where
    \begin{align}\label{theorem.unknown ordinary smooth under H_0 roshat}
        &\hat{r}^{os}_{\infty}(Y,W;\xi,\theta_0) = r^{os}_{\infty}(Y,W;\xi,\theta_0) + r^{\epsilon}_{\infty}(Y,W;\xi,\theta_0),
    \end{align}
    with
    \begin{align}\label{theorem.unknown ordinary smooth under H_0 rplus}
        &2\pi\cdot r^{\epsilon}_{\infty}(Y,W;\xi,\theta_0) = (gf_X)^{ft}(\xi)\Pi_{\epsilon}(\xi) - \int f_X^{ft}(t)g^{ft}(\xi-t)\Pi_{\epsilon}(t)dt\notag \\
        & - \left\{\int \left[(gf_X)^{ft}(t)\dot{g}^{ft}(-t) - f_X^{ft}(t)(g\dot{g})^{ft}(-t)\right]\Pi_{\epsilon}(t)dt\right\}^\top\Delta^{-1}(\theta_0)G(\xi,\theta_0).
    \end{align}
\end{theorem}

Using the continuous mapping theorem, we can readily derive the asymptotic null distributions of $\widehat{KS}_n$ and $\widehat{CvM}_n$ for the ordinary smooth case.

\begin{corollary}\label{Corollary.unknown ordinary}
Suppose that Assumptions \ref{ass.D},\ref{ass.D'},\ref{ass.O}, and \ref{ass.O'} hold. Then, under the null hypothesis $H_0$ in \eqref{hyp.null1}, we have
\begin{align*}
\sqrt{n}\widehat{KS}_{n} \stackrel{d}\longrightarrow \sup\limits_{\xi\in\Pi}\left\vert\hat{S}_{\infty}^{os}(\xi,\theta_0)\right\vert   \quad\text{ and }   \quad  n\widehat{CvM}_{n} \stackrel{d}\longrightarrow \int_{\Pi}\left\vert\hat{S}_{\infty}^{os}(\xi,\theta_0)\right\vert^2d\xi.
\end{align*}
\end{corollary}

Theorem \ref{theorem.unknown ordinary smooth under H_0} shows that for the ordinary smooth case when the distribution of measurement error is unknown, $\hat{S}^{pro}_n(\cdot,\hat{\theta}_n)$ still converges at the $\sqrt{n}$-rate to a centered Gaussian process. But it is worth noting that the Gaussian process here is different from the one mentioned in Theorem \ref{theorem.ordinary smooth under H_0}, because the uncertainty brought by the $\hat{f}_{\epsilon}^{ft}$ (represented by the term $ r^{\epsilon}_{\infty}(Y,W;\xi,\theta_0)$) changes the covariance structure of the limiting process.

For the supersmooth case, the following assumptions are necessary in addition to Assumption \ref{ass.S} to derive asymptotic properties of our statistics.

\begin{assumption}\label{ass.S'}

    \quad

    \begin{enumerate}[label=(\roman*)]
        \item $n\ee^{-6\mu (1+b^{-1})^{2}}\log(\frac{1}{b})^{-2}\to \infty$ as $n\to\infty$.

        \item Assume that
    \begin{align*}
        &\mathbb{E}\left\{\int_{\Pi} \left[r^{ss}_{\infty}(Y,W,\xi,\theta_0)+ r^{\epsilon}_{h,\infty}(Y,W,\xi,\theta_0)\right]^2\,d\xi\right\}<\infty,
    \end{align*}
    where $r^{ss}_{\infty}$ and $r^{\epsilon}_{h,\infty}$ are defined in Theorem \ref{theorem.supersmooth under H_0} and Assumption \ref{ass.O'}, respectively.
    \end{enumerate}
\end{assumption}

Similar to Assumption \ref{ass.O'}, Assumption \ref{ass.S'}({\romannumeral1}) requires a stronger assumption about the bandwidth $b$ for ensuring the asymptotic negligibility of uncertainty brought by the estimation of the error characteristic functions, while Assumption \ref{ass.S'}({\romannumeral2}) ensures that the asymptotic variance of $\hat{S}_{n}^{pro}(\cdot,\hat\theta_n)$ is bounded through an additional moment restriction based on Assumption \ref{ass.S}({\romannumeral5}).

The following theorem characterizes the asymptotic behavior of $\hat{S}_{n}^{pro}(\cdot,\hat\theta_n)$ for the supersmooth case under the null hypothesis.

\begin{theorem}\label{theorem.unknown supersmooth under H_0}

Suppose that Assumptions \ref{ass.D},\ref{ass.S},\ref{ass.D'}, and \ref{ass.S'} hold. Then, under the null hypothesis $H_0$ in \eqref{hyp.null1}, we have 
\begin{align}\label{result of main term unknown ass S H0}
\sqrt n \hat{S}_{n}^{pro}(\cdot,\hat\theta_n)\Longrightarrow \hat{S}_{\infty}^{ss}(\cdot,\theta_0),
\end{align}
where $\hat{S}_{\infty}^{ss}(\cdot,\theta_0)$ is a Gaussian process with mean zero and covariance structure 
\begin{align*}
    &Cov\left[\hat{S}_{\infty}^{ss}(\xi,\theta_0),\hat{S}_{\infty}^{ss}(\xi^\prime,\theta_0)\right] = \mathbb{E}\left[\hat{r}^{ss}_{\infty}(Y,W;\xi,\theta_0)\overline{\hat{r}^{ss}_{\infty}}(Y,W;\xi^\prime,\theta_0)\right],
\end{align*}
where
\begin{align}\label{theorem.unknown supersmooth under H_0 rss}
    &\hat{r}^{ss}_{\infty}(Y,W;\xi,\theta_0) = r^{ss}_{\infty}(Y,W;\xi,\theta_0) + r^\epsilon_{\infty}(Y,W;\xi,\theta_0)
\end{align}
and $r^\epsilon_{\infty}$ is defined in Theorem \ref{theorem.unknown ordinary smooth under H_0}.
\end{theorem}

\begin{corollary}\label{Corollary.unknown super}

Suppose that Assumptions \ref{ass.D}, \ref{ass.D'}, \ref{ass.S}, and \ref{ass.S'} hold. Then, under the null hypothesis $H_0$ in \eqref{hyp.null1}, we have
\begin{align*}
\sqrt{n}\widehat{KS}_{n} \stackrel{d}\longrightarrow \sup\limits_{\xi\in\Pi}\left\vert \hat{S}_{\infty}^{ss}(\xi,\theta_0)\right\vert    \quad\text{ and }\quad n\widehat{CvM}_{n} \stackrel{d}\longrightarrow \int_{\Pi}\left\vert \hat{S}_{\infty}^{ss}(\xi,\theta_0)\right\vert^2\,d\xi.
\end{align*}
\end{corollary}

Theorem \ref{theorem.unknown supersmooth under H_0} shows that for the supersmooth case, our empirical process still maintains the $\sqrt{n}$-convergence. Furthermore, Corollary \ref{Corollary.unknown super} shows convergence in distribution of $\sqrt n\widehat{KS}_n$ and $n\widehat{CvM}_n$ under the null hypothesis. The subsequent analysis considers the results under the sequence of local alternatives and the fixed alternative, thereby establishing the asymptotic power properties of the tests without requiring knowledge of the measurement error distribution.

\begin{theorem}\label{theorem.unknown under H_1n}

Under the sequence of local alternatives $H_{1n}$ in \eqref{hyp.local alt}, suppose that Assumptions \ref{ass.D} and \ref{ass.D'} hold. For the ordinary smooth case, suppose that Assumptions \ref{ass.O} and \ref{ass.O'} hold, 
\begin{align*}
\sqrt n \hat{S}_{n}^{pro}(\cdot,\hat\theta_n)\Longrightarrow \hat{S}_{\infty}^{os}(\cdot,\theta_0)+\mu_\Delta(\cdot,\theta_0),
\end{align*}
and for the supersmooth case, suppose that Assumptions \ref{ass.S} and \ref{ass.S'} hold,
\begin{align*}
    \sqrt n \hat{S}_{n}^{pro}(\cdot,\hat\theta_n)\Longrightarrow \hat{S}_{\infty}^{ss}(\cdot,\theta_0)+\mu_\Delta(\cdot,\theta_0),
\end{align*}
where $\hat{S}_{\infty}^{os}(\cdot,\theta_0)$ and $\hat{S}_{\infty}^{ss}(\cdot,\theta_0)$ are the centered Gaussian processes defined in Theorems \ref{theorem.unknown ordinary smooth under H_0} and \ref{theorem.unknown supersmooth under H_0}, respectively, and $\mu_\Delta(\cdot,\theta_0)$ is the deterministic shift process defined in Theorem \ref{theorem.known under H_1n}.
\end{theorem}

\begin{theorem}\label{theorem.unknown alternative}

Under the alternative hypothesis $H_{1}$ in \eqref{hyp.alternative1}, suppose that Assumptions \ref{ass.D} and \ref{ass.D'} hold. Whether Assumptions \ref{ass.O} and \ref{ass.O'} hold for the ordinary smooth case or Assumptions \ref{ass.S} and \ref{ass.S'} hold for the supersmooth case, we have
\begin{equation*}
\sup_{\xi\in\Pi}\left\vert \hat{S}_{n}^{pro}(\xi,\hat{\theta}_n) - C(\xi,\theta^\ast) \right\vert = o_p(1),
\end{equation*}
where $C(\cdot,\theta^\ast)$ is defined in \eqref{Drift}.
\end{theorem}

Under the sequence of local alternative hypotheses defined in \eqref{hyp.local alt}, we still have a deterministic shift term similar to Theorem \ref{theorem.known under H_1n}, signifying the nontrivial local power of our tests. We also find that under the alternative hypothesis $H_1$, our proposed specification tests share the same consistency condition as in the case where the measurement error distribution is known, see condition \eqref{thm.known alt consis cond}. 

\section{Multiplier bootstrap}\label{sec.boot}
As we show in Sections \ref{sec.Asy} and \ref{sec.unknown ME}, the null limiting distributions of test statistics $KS_{n}$ and $CvM_{n}$ usually depend on the underlying DGP in a rather complicated manner. As such, the associated critical values are case-dependent, forcing people to resort to bootstrap procedures. Unfortunately, although the wild bootstrap procedure is typically employed to implement tests in the error-free case, no computationally easy and valid bootstrap procedures are currently available for tests in the presence of measurement errors, as indicated above. In particular, any residual-based bootstrap methods (e.g., the widely used wild bootstrap in the literature of specification tests) are apparently infeasible in the presence of measurement errors, given that we need the true regressors to construct the residuals, but we cannot observe the regressors directly. 

For the influential local smoothing specification test proposed in \cite{otsu2021specification}, a multi-step estimation and resampling method is employed to operationalize the bootstrap procedure, a strategy also adopted in classical studies such as \cite{hall2007testing} for global smoothing tests. Specifically, they first estimate the density function of the unobservable true regressor $X$, the fitted nonparametric regression function, and the parametric fit using standard deconvolution techniques and a parametric estimator. Then, the residuals are resampled by estimating their second moments and drawing from distributions such as the two-point distribution of \cite{mammen1993bootstrap}. Based on these resampled residuals, bootstrap samples of $Y$ are generated as the sum of the resampled residuals and the estimated parametric fit evaluated at the resampled regressors drawn from the estimated density of $X$. The measurement error–contaminated bootstrap samples of variable $W$ are subsequently obtained by adding the true regressor and an independent draw from the known (or estimated) measurement error distribution. Finally, the bootstrap counterparts of the test statistics are computed using each resampled pair $\{(Y_i^\ast,W_i^\ast)^\top\}_{i=1}^n$. Furthermore, bandwidth selection is also a critical issue for local smoothing tests due to the sensitivity of local smoothers to the bandwidth choice, which is addressed by adopting the \textquotedblleft two-stage selection plug-in\textquotedblright{} method of \cite{delaigle2004practical} in the test of \cite{otsu2021specification}.

While the above approach is feasible, it still has several limitations that deserve further attention. First, the complex estimation and resampling steps complicate both theoretical justification and empirical implementation. In particular, integrals involving deconvolution kernels often lack closed-form solutions for many higher-order kernel functions, and numerical approximation may reduce accuracy. Second, within each simulation, bootstrap iteration requires repeating the estimation of functions and the computation of the bootstrap counterparts of the test statistics, making the overall procedure computationally inefficient. Third, the need to select suitable bandwidths adds further computational complexity and theoretical difficulties. In our proposed tests, we eliminate the parameter estimation effect via the projection approach, thereby facilitating a much simpler multiplier bootstrap and avoiding the aforementioned problems. This serves as one of the main contributions of the study.


\subsection{Known measurement error}\label{subsec.kme}
For the cases with known measurement error distribution, we approximate the asymptotic null distribution of a continuous functional $\Gamma(S_n^{pro})$ by that of $\Gamma(S_n^{pro,\ast})$, where $\Gamma(\cdot)$ represents the sup norm and the squared norm employed in constructing the statistics $KS_{n}$ and $CvM_{n}$, respectively, and
\begin{align}\label{bootstrap version of stat known ME}
S_{n}^{pro,\ast}(\xi,\hat\theta_n)=\frac{1}{n}\sum_{i=1}^n\int V_i\left(Y_i-g(x;\hat{\theta}_n)\right)\mathcal{K}_b\left(\frac{x-W_i}{b}\right)\mathcal{P}_n (x;\xi,\hat{\theta}_n)\,dx.
\end{align}
Here, $\{V_i\}_{i=1}^n$ is a sequence of i.i.d. random variables with zero mean, unit variance, and independent of the original sample, known as the multipliers. A popular example is i.i.d. Bernoulli variables constructed by \cite{mammen1993bootstrap}. Both for the ordinary smooth case and the supersmooth case, $S_{n}^{pro,\ast}(\cdot,\hat\theta_n)$ is expected to have the same limiting process as $S_{n}^{pro}(\cdot,\hat\theta_n)$ under the null hypothesis, due to the zero mean and unit variance properties of the multipliers. Furthermore, we expect that the multipliers can eliminate the deterministic shift term mentioned in Theorem \ref{theorem.known under H_1n} under the local alternatives, thus establishing that $S_{n}^{pro,\ast}(\cdot,\hat\theta_n)$ converges to the same limiting process as under the null hypothesis. 

Defining \textquotedblleft $\overset{\ast}{\Longrightarrow}$\textquotedblright{} as weak convergence and $\mathbb{P}_n^\ast$ as the bootstrap probability under the bootstrap law, i.e., conditional on the original sample, see, e.g., Section 2.9 of \cite{van1996weak}, the validity of the proposed multiplier bootstrap is formally established by the following theorem.

\begin{theorem}\label{theorem.boot known}
Under both the null hypothesis $H_0$ in \eqref{hyp.null1} and the sequence of local alternatives $H_{1n}$ in \eqref{hyp.local alt}, suppose that Assumptions \ref{ass.D} and \ref{ass.O} hold for the ordinary smooth case, and Assumptions \ref{ass.D} and \ref{ass.S} hold for the supersmooth case. Then, we have
\begin{align}\label{result of param effect boot}
\sup_{\xi\in\Pi}\left\vert S_n^{pro,\ast}(\xi,\hat\theta_n)- S_n^{pro,\ast}(\xi,\theta_0)\right\vert=o_p\left(n^{-\frac{1}{2}}\right).
\end{align}
Furthermore, for the ordinary smooth case,
\begin{align}\label{result of main term ass O boot}
    \sqrt n S_{n}^{pro,\ast}(\cdot,\hat\theta_n)\overset{\ast}{\Longrightarrow}S_{\infty}^{os}(\cdot,\theta_0),
\end{align}
and for the supersmooth case, 
\begin{align}\label{result of main term ass S boot}
    \sqrt n S_{n}^{pro,\ast}(\cdot,\hat\theta_n)\overset{\ast}{\Longrightarrow}S_{\infty}^{ss}(\cdot,\theta_0),
\end{align}
where $S_{\infty}^{os}(\cdot,\theta_0)$ and $S_{\infty}^{ss}(\cdot,\theta_0)$ are the centered Gaussian processes defined in Theorems \ref{theorem.ordinary smooth under H_0} and \ref{theorem.supersmooth under H_0}, respectively.
\end{theorem}

For the ordinary smooth and the supersmooth cases, respectively, the bootstrap empirical processes share the same limiting behavior under $H_0$ and $H_{1n}$ as their original sample counterparts under the null, while the latter exhibit a deterministic shift under $H_{1n}$. As a consequence of the above analysis, the asymptotic critical value at the significance level $\alpha$ is $c^\ast_{\alpha}=\inf\{c_\alpha\in[0,\infty):\lim_{n\to\infty}\mathbb{P}_n^\ast\{\Gamma(\sqrt n S_{n}^{pro,\ast})>c_\alpha\}=\alpha\}$. In practice, $c^\ast_{\alpha}$ can be approximated as $c^\ast_{n,\alpha} = \{\Gamma(\sqrt n S_{n}^{pro,\ast})\}_{B(1-\alpha)}$, the $B(1-\alpha)$-th order statistic for $B$ replicates $\{\Gamma(\sqrt n S_{n,b}^{pro,\ast})\}_{b=1}^B$ of $\{\Gamma(\sqrt n S_{n}^{pro,\ast})\}$ and we reject $H_0$ if $\Gamma(\sqrt n S_{n}^{pro})>c^\ast_{n,\alpha}$. Specifically, building upon Theorem \ref{theorem.boot known} and taking $\Gamma(\cdot)$ to be the sup norm and the squared norm, respectively, we establish the asymptotic validity of the multiplier bootstrap procedure for the $KS_n$ and $CvM_n$ statistics. 

\subsection{Unknown measurement error}
As we mentioned in Section \ref{sec.unknown ME}, in the absence of information about the distribution of measurement errors, we argue that our multiplier bootstrap procedure remains feasible. We still use repeated measurements to estimate the characteristic function of the measurement error, thereby estimating the deconvolution kernel as shown in \eqref{var.repeated ME}. Then we can construct the multiplier bootstrap version of $\hat{S}_n^{pro}(\cdot,\hat{\theta}_n)$ as given by
\begin{align}\label{bootstrap version of stat unknown ME}
\hat S_{n}^{pro,\ast}(\xi,\hat\theta_n)=\frac{1}{n}\sum_{i=1}^n\int V_i\left(Y_i-g(x;\hat{\theta}_n)\right)\hat{\mathcal{K}}_b\left(\frac{x-W_i}{b}\right)\hat{\mathcal{P}}_n(x;\xi,\hat{\theta}_n)\,dx.
\end{align}
Similar to what we have shown in Theorem \ref{theorem.boot known}, we establish the validity of our bootstrap procedure by the following theorem for the unknown measurement error case.

\begin{theorem}\label{theorem.boot unknown}
Suppose that Assumptions \ref{ass.D} and \ref{ass.D'} hold. Under both the null hypothesis $H_0$ in \eqref{hyp.null1} and the sequence of local alternatives $H_{1n}$ in \eqref{hyp.local alt}, suppose that Assumptions \ref{ass.O} and \ref{ass.O'} hold for the ordinary smooth case, and Assumptions \ref{ass.S} and \ref{ass.S'} hold for the supersmooth case. Then, we have 
\begin{align}\label{result of param effect unknown boot}
\sup_{\xi\in\Pi}\left\vert \hat{S}_n^{pro,\ast}(\xi,\hat\theta_n)- \hat{S}_n^{pro,\ast}(\xi,\theta_0)\right\vert=o_p\left(n^{-\frac{1}{2}}\right).
\end{align}
Furthermore, for the ordinary smooth case, 
\begin{align}\label{result of main term ass O unknown boot}
    \sqrt n \hat{S}_{n}^{pro,\ast}(\cdot,\hat\theta_n)\overset{\ast}{\Longrightarrow}\hat{S}_{\infty}^{os}(\cdot,\theta_0),
\end{align}
and for the supersmooth case, 
\begin{align}\label{result of main term ass S unknown boot}
    \sqrt n \hat{S}_{n}^{pro,\ast}(\cdot,\hat\theta_n)\overset{\ast}{\Longrightarrow}\hat{S}_{\infty}^{ss}(\cdot,\theta_0),
\end{align}
where $\hat{S}_{\infty}^{os}(\cdot,\theta_0)$ and $\hat{S}_{\infty}^{ss}(\cdot,\theta_0)$ are the centered Gaussian processes defined in Theorems \ref{theorem.unknown ordinary smooth under H_0} and \ref{theorem.unknown supersmooth under H_0}, respectively.
\end{theorem}

Theorem \ref{theorem.boot unknown} shows that the bootstrap version of the empirical process constructed by repeated measurements converges to the same limiting process under $H_0$ and $H_{1n}$, which is the same limiting process as the one under the null hypothesis mentioned in Section \ref{sec.unknown ME}, for the ordinary smooth and super smooth cases, respectively. Consequently, we can validate the multiplier bootstrap for $\widehat{KS}_n$ and $\widehat{CvM}_n$ using the continuous mapping theorem and obtain the critical values as discussed in Subsection \ref{subsec.kme}. Therefore, using the multiplier bootstrap to obtain the critical values for our tests remains feasible even when the measurement error is unknown.

\section{Numerical evidence}\label{sec.Simulation}
In this section, we investigate the finite-sample performance of the proposed test by Monte Carlo experiments. We set the unobservable regressors $\{X_i\}_{i=1}^n$ distributed as $N(0,1)$, and we use a trigonometric model and a polynomial model in our simulation, although we note that any nonlinear model can be used in our test procedure. We name $Y_i=1+X_i+\delta X_i^2+U_i$ as model one and $Y_i=1+X_i+\delta \cos(\pi X_i)+U_i$ as model two, respectively, for $i$ ranging from $1$ to $n$, where $\delta$ is a variable constant representing the nonlinearity of the model and is set to $0.5$. Additionally, $U_i\sim N(0,1/4)$. The contaminated regressor is given by $W_i=X_i+\epsilon_i$, where we use the Laplace distribution with variance of $1/12$ for the ordinary smooth case and $N(0,1/12)$ for the supersmooth case to remain consistent with the experiment in \cite{otsu2021specification}. We use the infinite-order flat-top kernel proposed by \cite{mcmurry2004nonparametric}, which has also been employed in \cite{dong2022nonparametric}, for both cases and all of our simulations, which is defined by
\begin{equation*}
K^{ft}(t) = \begin{cases} 
1 & \text{if } |t| \leq 0.05, \\
\exp \left\{ \frac{-\exp(-(|t|-0.05)^{-2})}{(|t|-1)^2} \right\} & \text{if } 0.05 < |t| < 1, \\
0 & \text{if } |t| \geq 1.
\end{cases}
\end{equation*}
We choose the bandwidth according to the rules of thumb, see, e.g., \cite{delaigle2008deconvolution}. Specifically, we use $b=c(5\sigma^4/n)^{1/27}$ for the ordinary smooth case and $b = c(4\sigma^2/\log(n))^{1/2}$ for the supersmooth case, where $\sigma$ is the standard deviation of the measurement error and sensitivity indicator $c$ varies in the grid presented in the following table. For the estimator in our test, we use a polynomial estimator of degree $2$, as mentioned in \cite{cheng1998polynomial}, which is also consistent with the results in \cite{otsu2021specification}.

We report the simulation results based on $199$ bootstrap iterations and $1000$ Monte Carlo replications. Tables include the size (the proportion of rejections when the DGP is under the null hypothesis, named DGP$(0)$) and powers (the proportion of rejections when the DGP is under the alternative hypothesis model one, named DGP$(1)$, and model two, named DGP$(2)$). We report the results for two statistics $KS_n$ and $CvM_n$, two sample sizes $n=\{500,1000\}$, three nominal level $\alpha = \{1\%,5\%,10\%\}$ and varied sensitivity indicators. Constrained by the length of papers, only part of the results $c\in\{1,5,10\}$ and $\alpha\in\{5\%,10\%\}$ are shown in the main text, and the complete results are reported in Tables \ref{tab:knwon laplace 500}--\ref{tab:unknwon normal 1000} of Appendix \ref{sec.Appendix sim res}.

\begin{table}[ht!]
\centering

	\caption{Results under known ordinary smooth case}
		
		\scalebox{0.95}{
			\begin{tabular}{c|c|c|cccccc}
			\hline\hline
                \multirow{2}{*}{$n$} & \multirow{2}{*}{$c$} & \multirow{2}{*}{level} & \multicolumn{2}{c}{DGP(0)} & \multicolumn{2}{c}{DGP(1)} & \multicolumn{2}{c}{DGP(2)}\\ 
                \cline{4-9} 
                \quad & \quad & \quad & KS & CvM & KS & CvM & KS & CvM\\
                \hline
                \multirow{6}{*}{500} & \multirow{2}{*}{1} & 0.05 & 0.042 & 0.041 & 0.815 & 0.806 & 0.775 & 0.762\\
                \quad & \quad & 0.1 & 0.115 & 0.114 & 0.898 & 0.888 & 0.864 & 0.849\\
                \cline{2-9}
                \quad & \multirow{2}{*}{5} & 0.05 & 0.058 & 0.054 & 0.835 & 0.818 & 0.779 & 0.769\\
                \quad & \quad & 0.1 & 0.099 & 0.101 & 0.890 & 0.882 & 0.870 & 0.856\\
                \cline{2-9}
                \quad & \multirow{2}{*}{10} & 0.05 & 0.051 & 0.052 & 0.818 & 0.797 & 0.745 & 0.735\\
                \quad & \quad & 0.1 & 0.093 & 0.093 & 0.881 & 0.870 & 0.855 & 0.843\\
                \hline
                \multirow{6}{*}{1000} & \multirow{2}{*}{1} & 0.05 & 0.054 & 0.052 & 0.985 & 0.979 & 0.965 & 0.959\\
                \quad & \quad & 0.1 & 0.123 & 0.117 & 0.990 & 0.987 & 0.979 & 0.974\\
                \cline{2-9}
                \quad & \multirow{2}{*}{5} & 0.05 & 0.052 & 0.050 & 0.977 & 0.968 & 0.966 & 0.955\\
                \quad & \quad & 0.1 & 0.101 & 0.099 & 0.993 & 0.988 & 0.978 & 0.973\\
                \cline{2-9}
                \quad & \multirow{2}{*}{10} & 0.05 & 0.049 & 0.050 & 0.969 & 0.962 & 0.969 & 0.964\\
                \quad & \quad & 0.1 & 0.087 & 0.090 & 0.985 & 0.979 & 0.987 & 0.978\\
			\hline\hline
			\end{tabular}
		}
		\label{tab:knwon laplace}
\end{table}

\begin{table}[ht!]
\centering
\caption{Results under known supersmooth case}
\scalebox{0.95}{
			\begin{tabular}{c|c|c|cccccc}
			\hline\hline
                \multirow{2}{*}{$n$} & \multirow{2}{*}{$c$} & \multirow{2}{*}{level} & \multicolumn{2}{c}{DGP(0)} & \multicolumn{2}{c}{DGP(1)} & \multicolumn{2}{c}{DGP(2)}\\ 
                \cline{4-9} 
                \quad & \quad & \quad & KS & CvM & KS & CvM & KS & CvM\\
                \hline
                \multirow{6}{*}{500} & \multirow{2}{*}{1} & 0.05 & 0.059 & 0.061 & 0.921 & 0.920 & 0.893 & 0.885\\
                \quad & \quad & 0.1 & 0.109 & 0.108 & 0.961 & 0.957 & 0.935 & 0.933\\
                \cline{2-9}
                \quad & \multirow{2}{*}{5} & 0.05 & 0.062 & 0.061 & 0.933 & 0.932 & 0.901 & 0.896\\
                \quad & \quad & 0.1 & 0.107 & 0.109 & 0.953 & 0.948 & 0.931 & 0.929\\
                \cline{2-9}
                \quad & \multirow{2}{*}{10} & 0.05 & 0.042 & 0.045 & 0.953 & 0.951 & 0.901 & 0.898\\
                \quad & \quad & 0.1 & 0.110 & 0.113 & 0.971 & 0.969 & 0.937 & 0.936\\
                \hline
                \multirow{6}{*}{1000} & \multirow{2}{*}{1} & 0.05 & 0.056 & 0.048 & 0.997 & 0.997 & 0.998 & 0.997\\
                \quad & \quad & 0.1 & 0.110 & 0.111 & 0.999 & 0.999 & 0.997 & 0.997\\
                \cline{2-9}
                \quad & \multirow{2}{*}{5} & 0.05 & 0.050 & 0.049 & 1.000 & 1.000 & 0.993 & 0.992\\
                \quad & \quad & 0.1 & 0.114 & 0.114 & 1.000 & 1.000 & 0.998 & 0.998\\
                \cline{2-9}
                \quad & \multirow{2}{*}{10} & 0.05 & 0.056 & 0.057 & 0.996 & 0.995 & 0.994 & 0.994\\
                \quad & \quad & 0.1 & 0.110 & 0.108 & 0.999 & 0.999 & 0.998 & 0.998\\
			\hline\hline
			\end{tabular}
		}
\label{tab:known normal}
\end{table}

\begin{table}[ht!]
\centering
\caption{Results under unknown ordinary smooth case}
\scalebox{0.95}{
			\begin{tabular}{c|c|c|cccccc}
			\hline\hline
                \multirow{2}{*}{$n$} & \multirow{2}{*}{$c$} & \multirow{2}{*}{level} & \multicolumn{2}{c}{DGP(0)} & \multicolumn{2}{c}{DGP(1)} & \multicolumn{2}{c}{DGP(2)}\\ 
                \cline{4-9} 
                \quad & \quad & \quad & KS & CvM & KS & CvM & KS & CvM\\
                \hline
                \multirow{6}{*}{500} & \multirow{2}{*}{1} & 0.05 & 0.060 & 0.059 & 0.814 & 0.796 & 0.777 & 0.768\\
                \quad & \quad & 0.1 & 0.133 & 0.128 & 0.892 & 0.874 & 0.868 & 0.851\\
                \cline{2-9}
                \quad & \multirow{2}{*}{5} & 0.05 & 0.055 & 0.053 & 0.822 & 0.801 & 0.786 & 0.770\\
                \quad & \quad & 0.1 & 0.116 & 0.108 & 0.891 & 0.885 & 0.844 & 0.837\\
                \cline{2-9}
                \quad & \multirow{2}{*}{10} & 0.05 & 0.042 & 0.041 & 0.830 & 0.812 & 0.777 & 0.766\\
                \quad & \quad & 0.1 & 0.104 & 0.105 & 0.909 & 0.900 & 0.844 & 0.831\\
                \hline
                \multirow{6}{*}{1000} & \multirow{2}{*}{1} & 0.05 & 0.061 & 0.058 & 0.970 & 0.963 & 0.949 & 0.946\\
                \quad & \quad & 0.1 & 0.118 & 0.115 & 0.989 & 0.984 & 0.969 & 0.964\\
                \cline{2-9}
                \quad & \multirow{2}{*}{5} & 0.05 & 0.062 & 0.060 & 0.980 & 0.970 & 0.964 & 0.956\\
                \quad & \quad & 0.1 & 0.116 & 0.111 & 0.989 & 0.987 & 0.980 & 0.976\\
                \cline{2-9}
                \quad & \multirow{2}{*}{10} & 0.05 & 0.057 & 0.052 & 0.980 & 0.978 & 0.955 & 0.946\\
                \quad & \quad & 0.1 & 0.105 & 0.104 & 0.992 & 0.988 & 0.982 & 0.973\\
			\hline\hline
			\end{tabular}
		}
\label{tab:unknown laplace}
\end{table}

\begin{table}[ht!]
\centering
\caption{Results under unknown supersmooth case}
\scalebox{0.95}{
			\begin{tabular}{c|c|c|cccccc}
			\hline\hline
                \multirow{2}{*}{$n$} & \multirow{2}{*}{$c$} & \multirow{2}{*}{level} & \multicolumn{2}{c}{DGP(0)} & \multicolumn{2}{c}{DGP(1)} & \multicolumn{2}{c}{DGP(2)}\\ 
                \cline{4-9} 
                \quad & \quad & \quad & KS & CvM & KS & CvM & KS & CvM\\
                \hline
                \multirow{6}{*}{500} & \multirow{2}{*}{1} & 0.05 & 0.056 & 0.057 & 0.921 & 0.920 & 0.895 & 0.886\\
                \quad & \quad & 0.1 & 0.114 & 0.115 & 0.974 & 0.973 & 0.936 & 0.933\\
                \cline{2-9}
                \quad & \multirow{2}{*}{5} & 0.05 & 0.057 & 0.057 & 0.927 & 0.924 & 0.901 & 0.901\\
                \quad & \quad & 0.05 & 0.057 & 0.057 & 0.927 & 0.924 & 0.901 & 0.901\\
                \cline{2-9}
                \quad & \multirow{2}{*}{10} & 0.05 & 0.055 & 0.055 & 0.927 & 0.927 & 0.896 & 0.894\\
                \quad & \quad & 0.1 & 0.092 & 0.092 & 0.973 & 0.971 & 0.932 & 0.932\\
                \hline
                \multirow{6}{*}{1000} & \multirow{2}{*}{1} & 0.05 & 0.046 & 0.045 & 0.998 & 0.998 & 0.995 & 0.993\\
                \quad & \quad & 0.1 & 0.116 & 0.115 & 1.000 & 1.000 & 0.998 & 0.998\\
                \cline{2-9}
                \quad & \multirow{2}{*}{5} & 0.05 & 0.052 & 0.051 & 0.996 & 0.996 & 0.993 & 0.993\\
                \quad & \quad & 0.1 & 0.097 & 0.102 & 0.999 & 0.999 & 0.998 & 0.997\\
                \cline{2-9}
                \quad & \multirow{2}{*}{10} & 0.05 & 0.058 & 0.057 & 1.000 & 0.999 & 0.993 & 0.992\\
                \quad & \quad & 0.1 & 0.108 & 0.106 & 0.999 & 0.999 & 1.000 & 1.000\\
			\hline\hline
			\end{tabular}
		}
\label{tab:unknown normal}
\end{table}

Tables \ref{tab:knwon laplace} and \ref{tab:known normal} show our testing results when we know the distribution of the measurement error for the ordinary smooth case and the supersmooth case, respectively. In our results, the $KS_n$ and $CvM_n$ statistics perform well 
in terms of both size and power. 
We observe significant improvements in size accuracy and power gain as the sample size increases. 
We also observe that the size and power of the tests do not fluctuate much with changes in bandwidth values on the grid. Comparing the results between the ordinary smooth case and the supersmooth case, we find that test sizes for the supersmooth cases perform comparably to those for the ordinary smooth cases, which differs from the findings reported in \cite{otsu2021specification}. It stems from the fact that the statistics for both the ordinary smooth and supersmooth cases exhibit $\sqrt{n}$-rate convergence, another advantage of the proposed tests. Other results in the tables reveal the power properties of our tests, where DGP$(1)$ represents the polynomial model and DGP$(2)$ represents the trigonometric model. It is well known that global smoothing tests perform better in polynomial models, while local smoothing tests perform better in trigonometric models. For our results in Tables \ref{tab:knwon laplace} and \ref{tab:known normal}, the test performance under DGP$(1)$ is better than that under DGP$(2)$. We can expect better performance with higher nonlinearity and more samples.

Tables \ref{tab:unknown laplace} and \ref{tab:unknown normal} display size and power properties of the $\widehat{KS}_n$ and $\widehat{CvM}_n$ statistics in the absence of information about the distribution of the measurement error. However, because our asymptotic variance is affected by the estimation error of the characteristic functions, the results for size and power in finite samples are slightly worse but remain within the acceptable range. 
We can also observe that our tests perform better under the polynomial model than under the trigonometric model. Meanwhile, our tests still retain robustness to bandwidth selection when the measurement error distribution is unknown. 

Finally, we summarize the advantages of our tests, as demonstrated in the simulations, including a good level of accuracy and power properties due to the $\sqrt{n}$-convergence of the statistics, better efficiency for polynomials and low-frequency alternative hypotheses, and robustness in bandwidth selection. Additionally, our testing framework naturally accommodates the case without measurement error.

\section{Conclusion}\label{sec.Conclusion}
In this paper, we propose new ICM-type specification tests for regression models with measurement errors based on a deconvoluted residual-marked empirical process. The innovation of our method is the use of a projection to eliminate the parameter estimation effect, a particularly attractive tool in the presence of measurement errors. 
We formally establish the asymptotic properties of the projection-based 
test statistics and suggest a straightforward multiplier bootstrap procedure to simulate the critical values. Tests and the multiplier bootstrap procedure are also examined when the measurement error distribution is unknown. We emphasize three main advantages of the tests proposed in this paper. From a theoretical perspective, the statistics of $\sqrt{n}$-convergence for both the ordinary smooth and the supersmooth measurement errors ensure excellent size accuracy and satisfactory power properties, especially for the supersmooth cases, where slower convergence rates can typically be achieved, as mentioned in the previous literature on measurement error. From the application perspective, our proposed method enables a general convergence rate of the parametric estimator, thereby resolving the difficulty in obtaining a $\sqrt{n}$-convergence estimator in the presence of measurement errors. Additionally, our tests are more robust to bandwidth selection and, therefore, more reliable. From a computational perspective, our projection approach enables us to 
simulate the critical values as accurately as desired via the multiplier bootstrap, which is easier to understand and compute. 

For future research, we anticipate that the valuable combination of projection and multiplier bootstrap can be extended to other interesting testing problems in the presence of measurement errors, such as tests for heteroskedasticity and significance tests.


\bibliographystyle{apalike_revised}
\bibliography{Ref}

\newpage

\begin{center}
\Large{Specification tests for regression models with
measurement errors\\
-- Online supplementary appendix}
\end{center}

\begin{center}
\Large{Xiaojun Song and Jichao Yuan\\
Peking University}
\end{center}

\normalsize

In this appendix, we provide additional simulation results and proofs of the main theoretical results.

\appendix
\section{Additional simulation results}\label{sec.Appendix sim res}
\begin{table}[ht!]
\centering

	\caption{Results for the known ordinary smooth case, $n=500$}
		
		\scalebox{0.75}{
			\begin{tabular}{cccccccc}
			\hline\hline
                \multirow{2}{*}{$c$} & \multirow{2}{*}{level} & \multicolumn{2}{c}{DGP(0)} & \multicolumn{2}{c}{DGP(1)} & \multicolumn{2}{c}{DGP(2)}\\ 
                \cline{3-8} 
                \quad & \quad & KS & CvM & KS & CvM & KS & CvM\\
                \hline
                \multirow{3}{*}{1} & 0.01 & 0.021 & 0.022 & 0.662 & 0.644 & 0.609 & 0.591\\
                \quad & 0.05 & 0.042 & 0.041 & 0.815 & 0.806 & 0.775 & 0.762\\
                \quad & 0.1 & 0.115 & 0.114 & 0.898 & 0.888 & 0.864 & 0.849\\
                \hline
                \multirow{3}{*}{2} & 0.01 & 0.009 & 0.008 & 0.664 & 0.648 & 0.603 & 0.580\\
                \quad & 0.05 & 0.070 & 0.067 & 0.849 & 0.840 & 0.783 & 0.767\\
                \quad & 0.1 & 0.109 & 0.109 & 0.899 & 0.886 & 0.832 & 0.824\\
                \hline
                \multirow{3}{*}{3} & 0.01 & 0.011 & 0.011 & 0.656 & 0.641 & 0.596 & 0.583\\
                \quad & 0.05 & 0.045 & 0.045 & 0.820 & 0.804 & 0.773 & 0.764\\
                \quad & 0.1 & 0.112 & 0.114 & 0.876 & 0.867 & 0.842 & 0.832\\
                \hline
                \multirow{3}{*}{5} & 0.01 & 0.012 & 0.013 & 0.639 & 0.617 & 0.616 & 0.601\\
                \quad & 0.05 & 0.058 & 0.054 & 0.835 & 0.818 & 0.779 & 0.769\\
                \quad & 0.1 & 0.099 & 0.101 & 0.890 & 0.882 & 0.870 & 0.856\\
                \hline
                \multirow{3}{*}{10} & 0.01 & 0.01 & 0.01 & 0.677 & 0.662 & 0.584 & 0.570\\
                \quad & 0.05 & 0.051 & 0.052 & 0.818 & 0.797 & 0.745 & 0.735\\
                \quad & 0.1 & 0.093 & 0.093 & 0.881 & 0.870 & 0.855 & 0.843\\
                \hline
                \multirow{3}{*}{15} & 0.01 & 0.011 & 0.011 & 0.675 & 0.656 & 0.595 & 0.582\\
                \quad & 0.05 & 0.047 & 0.045 & 0.837 & 0.829 & 0.775 & 0.761\\
                \quad & 0.1 & 0.047 & 0.045 & 0.837 & 0.829 & 0.775 & 0.761\\
			\hline\hline
			\end{tabular}
		}
		\label{tab:knwon laplace 500}
\end{table}

\begin{table}[ht!]
\centering

	\caption{Results for the known ordinary smooth case, $n=1000$}
		
		\scalebox{0.75}{
			\begin{tabular}{cccccccc}
			\hline\hline
                \multirow{2}{*}{$c$} & \multirow{2}{*}{level} & \multicolumn{2}{c}{DGP(0)} & \multicolumn{2}{c}{DGP(1)} & \multicolumn{2}{c}{DGP(2)}\\ 
                \cline{3-8} 
                \quad & \quad & KS & CvM & KS & CvM & KS & CvM\\
                \hline
                \multirow{3}{*}{1} & 0.01 & 0.014 & 0.014 & 0.911 & 0.898 & 0.867 & 0.850\\
                \quad & 0.05 & 0.054 & 0.052 & 0.985 & 0.979 & 0.965 & 0.959\\
                \quad & 0.1 & 0.123 & 0.117 & 0.990 & 0.987 & 0.979 & 0.974\\
                \hline
                \multirow{3}{*}{2} & 0.01 & 0.013 & 0.012 & 0.931 & 0.917 & 0.881 & 0.873\\
                \quad & 0.05 & 0.054 & 0.056 & 0.974 & 0.963 & 0.961 & 0.953\\
                \quad & 0.1 & 0.105 & 0.102 & 0.988 & 0.986 & 0.976 & 0.965\\
                \hline
                \multirow{3}{*}{3} & 0.01 & 0.019 & 0.017 & 0.930 & 0.921 & 0.875 & 0.863\\
                \quad & 0.05 & 0.052 & 0.049 & 0.973 & 0.968 & 0.947 & 0.935\\
                \quad & 0.1 & 0.101 & 0.107 & 0.987 & 0.983 & 0.980 & 0.970\\
                \hline
                \multirow{3}{*}{5} & 0.01 & 0.010 & 0.011 & 0.930 & 0.918 & 0.879 & 0.859\\
                \quad & 0.05 & 0.052 & 0.050 & 0.977 & 0.968 & 0.966 & 0.955\\
                \quad & 0.1 & 0.101 & 0.099 & 0.993 & 0.988 & 0.978 & 0.973\\
                \hline
                \multirow{3}{*}{10} & 0.01 & 0.012 & 0.015 & 0.924 & 0.915 & 0.882 & 0.853\\
                \quad & 0.05 & 0.049 & 0.050 & 0.969 & 0.962 & 0.969 & 0.964\\
                \quad & 0.1 & 0.087 & 0.090 & 0.985 & 0.979 & 0.987 & 0.978\\
                \hline
                \multirow{3}{*}{15} & 0.01 & 0.009 & 0.009 & 0.938 & 0.921 & 0.910 & 0.898\\
                \quad & 0.05 & 0.045 & 0.045 & 0.975 & 0.973 & 0.962 & 0.949\\
                \quad & 0.1 & 0.091 & 0.094 & 0.993 & 0.987 & 0.984 & 0.980\\
			\hline\hline
			\end{tabular}
		}
		\label{tab:knwon laplace 1000}
\end{table}

\begin{table}[ht!]
\centering

	\caption{Results for the known supersmooth case, $n=500$}
		
		\scalebox{0.9}{
			\begin{tabular}{cccccccc}
			\hline\hline
                \multirow{2}{*}{$c$} & \multirow{2}{*}{level} & \multicolumn{2}{c}{DGP(0)} & \multicolumn{2}{c}{DGP(1)} & \multicolumn{2}{c}{DGP(2)}\\ 
                \cline{3-8} 
                \quad & \quad & KS & CvM & KS & CvM & KS & CvM\\
                \hline
                \multirow{3}{*}{1} & 0.01 & 0.014 & 0.014 & 0.805 & 0.805 & 0.755 & 0.754\\
                \quad & 0.05 & 0.059 & 0.061 & 0.921 & 0.920 & 0.893 & 0.885\\
                \quad & 0.1 & 0.109 & 0.108 & 0.961 & 0.957 & 0.935 & 0.933\\
                \hline
                \multirow{3}{*}{2} & 0.01 & 0.010 & 0.010 & 0.819 & 0.814 & 0.750 & 0.741\\
                \quad & 0.05 & 0.059 & 0.059 & 0.927 & 0.928 & 0.905 & 0.902\\
                \quad & 0.1 & 0.090 & 0.092 & 0.968 & 0.965 & 0.963 & 0.962\\
                \hline
                \multirow{3}{*}{3} & 0.01 & 0.023 & 0.022 & 0.812 & 0.811 & 0.747 & 0.742\\
                \quad & 0.05 & 0.049 & 0.049 & 0.937 & 0.936 & 0.882 & 0.876\\
                \quad & 0.1 & 0.123 & 0.122 & 0.973 & 0.971 & 0.950 & 0.950\\
                \hline
                \multirow{3}{*}{5} & 0.01 & 0.017 & 0.017 & 0.807 & 0.803 & 0.757 & 0.752\\
                \quad & 0.05 & 0.062 & 0.061 & 0.933 & 0.932 & 0.901 & 0.896\\
                \quad & 0.1 & 0.107 & 0.109 & 0.953 & 0.948 & 0.931 & 0.929\\
                \hline
                \multirow{3}{*}{10} & 0.01 & 0.011 & 0.011 & 0.834 & 0.829 & 0.761 & 0.759\\
                \quad & 0.05 & 0.042 & 0.045 & 0.953 & 0.951 & 0.901 & 0.898\\
                \quad & 0.1 & 0.110 & 0.113 & 0.971 & 0.969 & 0.937 & 0.936\\
                \hline
                \multirow{3}{*}{15} & 0.01 & 0.020 & 0.021 & 0.821 & 0.815 & 0.773 & 0.766\\
                \quad & 0.05 & 0.049 & 0.049 & 0.934 & 0.936 & 0.896 & 0.896\\
                \quad & 0.1 & 0.106 & 0.107 & 0.977 & 0.974 & 0.949 & 0.947\\
			\hline\hline
			\end{tabular}
		}
		\label{tab:knwon normal 500}
\end{table}

\begin{table}[ht!]
\centering

	\caption{Results for the known supersmooth case, $n=1000$}
		
		\scalebox{0.9}{
			\begin{tabular}{cccccccc}
			\hline\hline
                \multirow{2}{*}{$c$} & \multirow{2}{*}{level} & \multicolumn{2}{c}{DGP(0)} & \multicolumn{2}{c}{DGP(1)} & \multicolumn{2}{c}{DGP(2)}\\ 
                \cline{3-8} 
                \quad & \quad & KS & CvM & KS & CvM & KS & CvM\\
                \hline
                \multirow{3}{*}{1} & 0.01 & 0.014 & 0.015 & 0.989 & 0.989 & 0.975 & 0.974\\
                \quad & 0.05 & 0.056 & 0.048 & 0.997 & 0.997 & 0.998 & 0.997\\
                \quad & 0.1 & 0.110 & 0.111 & 0.999 & 0.999 & 0.997 & 0.997\\
                \hline
                \multirow{3}{*}{2} & 0.01 & 0.008 & 0.008 & 0.984 & 0.983 & 0.977 & 0.977\\
                \quad & 0.05 & 0.056 & 0.056 & 0.996 & 0.996 & 0.997 & 0.997\\
                \quad & 0.1 & 0.096 & 0.098 & 0.999 & 0.999 & 0.998 & 0.999\\
                \hline
                \multirow{3}{*}{3} & 0.01 & 0.010 & 0.011 & 0.984 & 0.981 & 0.969 & 0.967\\
                \quad & 0.05 & 0.054 & 0.053 & 1.000 & 1.000 & 0.998 & 0.998\\
                \quad & 0.1 & 0.103 & 0.100 & 0.999 & 0.999 & 0.997 & 0.997\\
                \hline
                \multirow{3}{*}{5} & 0.01 & 0.011 & 0.010 & 0.983 & 0.983 & 0.975 & 0.974\\
                \quad & 0.05 & 0.050 & 0.049 & 1.000 & 1.000 & 0.993 & 0.992\\
                \quad & 0.1 & 0.114 & 0.114 & 1.000 & 1.000 & 0.998 & 0.998\\
                \hline
                \multirow{3}{*}{10} & 0.01 & 0.016 & 0.018 & 0.993 & 0.993 & 0.980 & 0.977\\
                \quad & 0.05 & 0.056 & 0.057 & 0.996 & 0.995 & 0.994 & 0.994\\
                \quad & 0.1 & 0.110 & 0.108 & 0.999 & 0.999 & 0.998 & 0.998\\
                \hline
                \multirow{3}{*}{15} & 0.01 & 0.012 & 0.013 & 0.992 & 0.992 & 0.981 & 0.979\\
                \quad & 0.05 & 0.047 & 0.048 & 0.998 & 0.997 & 0.998 & 0.998\\
                \quad & 0.1 & 0.109 & 0.111 & 0.999 & 0.999 & 1.000 & 1.000\\
			\hline\hline
			\end{tabular}
		}
		\label{tab:knwon normal 1000}
\end{table}

\begin{table}[ht!]
\centering

	\caption{Results for the unknown ordinary smooth case, $n=500$}
		
		\scalebox{0.9}{
			\begin{tabular}{cccccccc}
			\hline\hline
                \multirow{2}{*}{$c$} & \multirow{2}{*}{level} & \multicolumn{2}{c}{DGP(0)} & \multicolumn{2}{c}{DGP(1)} & \multicolumn{2}{c}{DGP(2)}\\ 
                \cline{3-8} 
                \quad & \quad & KS & CvM & KS & CvM & KS & CvM\\
                \hline
                \multirow{3}{*}{1} & 0.01 & 0.015 & 0.012 & 0.624 & 0.607 & 0.591 & 0.578\\
                \quad & 0.05 & 0.060 & 0.059 & 0.814 & 0.796 & 0.777 & 0.768\\
                \quad & 0.1 & 0.133 & 0.128 & 0.892 & 0.874 & 0.868 & 0.851\\
                \hline
                \multirow{3}{*}{2} & 0.01 & 0.015 & 0.014 & 0.666 & 0.649 & 0.567 & 0.549\\
                \quad & 0.05 & 0.061 & 0.056 & 0.823 & 0.813 & 0.776 & 0.759\\
                \quad & 0.1 & 0.124 & 0.121 & 0.892 & 0.880 & 0.851 & 0.833\\
                \hline
                \multirow{3}{*}{3} & 0.01 & 0.015 & 0.016 & 0.682 & 0.669 & 0.576 & 0.566\\
                \quad & 0.05 & 0.065 & 0.058 & 0.830 & 0.819 & 0.782 & 0.767\\
                \quad & 0.1 & 0.115 & 0.112 & 0.905 & 0.890 & 0.853 & 0.846\\
                \hline
                \multirow{3}{*}{5} & 0.01 & 0.016 & 0.018 & 0.651 & 0.633 & 0.597 & 0.590\\
                \quad & 0.05 & 0.055 & 0.053 & 0.822 & 0.801 & 0.786 & 0.770\\
                \quad & 0.1 & 0.116 & 0.108 & 0.891 & 0.885 & 0.844 & 0.837\\
                \hline
                \multirow{3}{*}{10} & 0.01 & 0.019 & 0.019 & 0.664 & 0.640 & 0.590 & 0.571\\
                \quad & 0.05 & 0.042 & 0.041 & 0.830 & 0.812 & 0.777 & 0.766\\
                \quad & 0.1 & 0.104 & 0.105 & 0.909 & 0.900 & 0.844 & 0.831\\
                \hline
                \multirow{3}{*}{15} & 0.01 & 0.006 & 0.006 & 0.672 & 0.657 & 0.604 & 0.590\\
                \quad & 0.05 & 0.056 & 0.046 & 0.833 & 0.822 & 0.796 & 0.775\\
                \quad & 0.1 & 0.056 & 0.046 & 0.833 & 0.822 & 0.796 & 0.775\\
			\hline\hline
			\end{tabular}
		}
		\label{tab:unknwon laplace 500}
\end{table}

\begin{table}[ht!]
\centering

	\caption{Results for the unknown ordinary smooth case, $n=1000$}
		
		\scalebox{0.9}{
			\begin{tabular}{cccccccc}
			\hline\hline
                \multirow{2}{*}{$c$} & \multirow{2}{*}{level} & \multicolumn{2}{c}{DGP(0)} & \multicolumn{2}{c}{DGP(1)} & \multicolumn{2}{c}{DGP(2)}\\ 
                \cline{3-8} 
                \quad & \quad & KS & CvM & KS & CvM & KS & CvM\\
                \hline
                \multirow{3}{*}{1} & 0.01 & 0.012 & 0.011 & 0.931 & 0.914 & 0.890 & 0.870\\
                \quad & 0.05 & 0.061 & 0.058 & 0.970 & 0.963 & 0.949 & 0.946\\
                \quad & 0.1 & 0.118 & 0.115 & 0.989 & 0.984 & 0.969 & 0.964\\
                \hline
                \multirow{3}{*}{2} & 0.01 & 0.012 & 0.010 & 0.936 & 0.924 & 0.893 & 0.878\\
                \quad & 0.05 & 0.061 & 0.061 & 0.983 & 0.977 & 0.961 & 0.953\\
                \quad & 0.1 & 0.106 & 0.106 & 0.993 & 0.986 & 0.983 & 0.976\\
                \hline
                \multirow{3}{*}{3} & 0.01 & 0.013 & 0.012 & 0.928 & 0.913 & 0.904 & 0.894\\
                \quad & 0.05 & 0.062 & 0.053 & 0.976 & 0.971 & 0.972 & 0.960\\
                \quad & 0.1 & 0.104 & 0.100 & 0.989 & 0.985 & 0.981 & 0.978\\
                \hline
                \multirow{3}{*}{5} & 0.01 & 0.020 & 0.018 & 0.925 & 0.909 & 0.905 & 0.886\\
                \quad & 0.05 & 0.062 & 0.060 & 0.98 & 0.97 & 0.964 & 0.956\\
                \quad & 0.1 & 0.116 & 0.111 & 0.989 & 0.987 & 0.980 & 0.976\\
                \hline
                \multirow{3}{*}{10} & 0.01 & 0.011 & 0.010 & 0.935 & 0.920 & 0.895 & 0.878\\
                \quad & 0.05 & 0.057 & 0.052 & 0.980 & 0.978 & 0.955 & 0.946\\
                \quad & 0.1 & 0.105 & 0.104 & 0.992 & 0.988 & 0.982 & 0.973\\
                \hline
                \multirow{3}{*}{15} & 0.01 & 0.010 & 0.009 & 0.920 & 0.906 & 0.894 & 0.872\\
                \quad & 0.05 & 0.042 & 0.040 & 0.979 & 0.971 & 0.956 & 0.949\\
                \quad & 0.1 & 0.109 & 0.101 & 0.991 & 0.984 & 0.992 & 0.984\\
			\hline\hline
			\end{tabular}
		}
		\label{tab:unknwon laplace 1000}
\end{table}

\begin{table}[ht!]
\centering

	\caption{Results for the unknown supersmooth case, $n=500$}
		
		\scalebox{0.9}{
			\begin{tabular}{cccccccc}
			\hline\hline
                \multirow{2}{*}{$c$} & \multirow{2}{*}{level} & \multicolumn{2}{c}{DGP(0)} & \multicolumn{2}{c}{DGP(1)} & \multicolumn{2}{c}{DGP(2)}\\ 
                \cline{3-8} 
                \quad & \quad & KS & CvM & KS & CvM & KS & CvM\\
                \hline
                \multirow{3}{*}{1} & 0.01 & 0.019 & 0.019 & 0.816 & 0.812 & 0.750 & 0.739\\
                \quad & 0.05 & 0.056 & 0.057 & 0.921 & 0.920 & 0.895 & 0.886\\
                \quad & 0.1 & 0.114 & 0.115 & 0.974 & 0.973 & 0.936 & 0.933\\
                \hline
                \multirow{3}{*}{2} & 0.01 & 0.011 & 0.012 & 0.822 & 0.819 & 0.781 & 0.778\\
                \quad & 0.05 & 0.059 & 0.060 & 0.934 & 0.932 & 0.902 & 0.899\\
                \quad & 0.1 & 0.120 & 0.118 & 0.969 & 0.968 & 0.934 & 0.931\\
                \hline
                \multirow{3}{*}{3} & 0.01 & 0.012 & 0.012 & 0.820 & 0.814 & 0.757 & 0.755\\
                \quad & 0.05 & 0.076 & 0.077 & 0.925 & 0.923 & 0.930 & 0.929\\
                \quad & 0.1 & 0.076 & 0.077 & 0.925 & 0.923 & 0.930 & 0.929\\
                \hline
                \multirow{3}{*}{5} & 0.01 & 0.013 & 0.013 & 0.830 & 0.830 & 0.751 & 0.744\\
                \quad & 0.05 & 0.057 & 0.057 & 0.927 & 0.924 & 0.901 & 0.901\\
                \quad & 0.1 & 0.127 & 0.124 & 0.959 & 0.960 & 0.952 & 0.951\\
                \hline
                \multirow{3}{*}{10} & 0.01 & 0.014 & 0.014 & 0.804 & 0.800 & 0.757 & 0.749\\
                \quad & 0.05 & 0.055 & 0.055 & 0.927 & 0.927 & 0.896 & 0.894\\
                \quad & 0.1 & 0.092 & 0.092 & 0.973 & 0.971 & 0.932 & 0.932\\
                \hline
                \multirow{3}{*}{15} & 0.01 & 0.015 & 0.014 & 0.808 & 0.804 & 0.777 & 0.775\\
                \quad & 0.05 & 0.057 & 0.057 & 0.926 & 0.920 & 0.904 & 0.898\\
                \quad & 0.1 & 0.098 & 0.100 & 0.976 & 0.974 & 0.951 & 0.950\\
			\hline\hline
			\end{tabular}
		}
		\label{tab:unknwon normal 500}
\end{table}

\begin{table}[ht!]
\centering

	\caption{Results for the unknown supersmooth case, $n=1000$}
		
		\scalebox{0.9}{
			\begin{tabular}{cccccccc}
			\hline\hline
                \multirow{2}{*}{$c$} & \multirow{2}{*}{level} & \multicolumn{2}{c}{DGP(0)} & \multicolumn{2}{c}{DGP(1)} & \multicolumn{2}{c}{DGP(2)}\\ 
                \cline{3-8} 
                \quad & \quad & KS & CvM & KS & CvM & KS & CvM\\
                \hline
                \multirow{3}{*}{1} & 0.01 & 0.018 & 0.019 & 0.990 & 0.988 & 0.983 & 0.984\\
                \quad & 0.05 & 0.046 & 0.045 & 0.998 & 0.998 & 0.995 & 0.993\\
                \quad & 0.1 & 0.116 & 0.115 & 1.000 & 1.000 & 0.998 & 0.998\\
                \hline
                \multirow{3}{*}{2} & 0.01 & 0.017 & 0.018 & 0.987 & 0.986 & 0.975 & 0.974\\
                \quad & 0.05 & 0.046 & 0.046 & 0.999 & 0.999 & 0.996 & 0.993\\
                \quad & 0.1 & 0.109 & 0.105 & 0.998 & 0.998 & 0.998 & 0.998\\
                \hline
                \multirow{3}{*}{3} & 0.01 & 0.020 & 0.020 & 0.990 & 0.990 & 0.983 & 0.982\\
                \quad & 0.05 & 0.051 & 0.050 & 0.999 & 0.999 & 0.989 & 0.989\\
                \quad & 0.1 & 0.108 & 0.106 & 0.999 & 0.999 & 0.999 & 0.999\\
                \hline
                \multirow{3}{*}{5} & 0.01 & 0.019 & 0.020 & 0.986 & 0.985 & 0.981 & 0.979\\
                \quad & 0.05 & 0.052 & 0.051 & 0.996 & 0.996 & 0.993 & 0.993\\
                \quad & 0.1 & 0.097 & 0.102 & 0.999 & 0.999 & 0.998 & 0.997\\
                \hline
                \multirow{3}{*}{10} & 0.01 & 0.018 & 0.018 & 0.985 & 0.985 & 0.968 & 0.969\\
                \quad & 0.05 & 0.058 & 0.057 & 1.000 & 0.999 & 0.993 & 0.992\\
                \quad & 0.1 & 0.108 & 0.106 & 0.999 & 0.999 & 1.000 & 1.000\\
                \hline
                \multirow{3}{*}{15} & 0.01 & 0.017 & 0.017 & 0.979 & 0.978 & 0.981 & 0.980\\
                \quad & 0.05 & 0.056 & 0.055 & 0.997 & 0.997 & 0.998 & 0.997\\
                \quad & 0.1 & 0.099 & 0.100 & 1.000 & 1.000 & 0.999 & 0.999\\
			\hline\hline
			\end{tabular}
		}
		\label{tab:unknwon normal 1000}
\end{table}

\newpage

\section{Definitions in this article}\label{sec.AppendixA}
In this subsection, we first give the notations to be used, and then begin our proofs of theorems and lemmas. We recall the univariate deconvolution kernel from Section \ref{sec.Test},
\begin{align*}
    &\mathcal{K}_b(a)=\frac{1}{2\pi b}\int \ee^{-\ii ta}\frac{K^{\text{ft}}(t)}{f_\epsilon^{\text{ft}}(t/b)}\,dt.
\end{align*}
As explained in \cite{fan1992deconvolution}, the idea behind the construction of the kernel function is to use Fourier transform and its inverse, possessing the following properties, $\mathcal{F}(f\ast g) = \mathcal{F}(f)\mathcal{F}(g)$ and $\mathcal{F}^{-1}(fg) = \mathcal{F}^{-1}(f)\ast\mathcal{F}^{-1}(g)$, where
\begin{align*}
    &\mathcal{F}(f) = \int\ee^{\ii xt}f(x)\,dx, \,\,\, \mathcal{F}^{-1}(g) = \frac{1}{2\pi}\int\ee^{-\ii xt}g(t)\,dt.
\end{align*}
And for convenience, we give the following notations about the kernel function,
\begin{align*}
    &\mathcal{K}_\epsilon(x) = b\mathcal{K}_b(x),\quad
    \hat{\mathcal{K}}_\epsilon(x) = b\hat{\mathcal{K}}_b(x).
\end{align*}
Then, for the unknown measurement error case mentioned in section \ref{sec.unknown ME}, we denote,
\begin{align} \label{Def of est of ME}
    &\hat{f}_\epsilon^{\text{ft}}(t) = \left\vert\frac{1}{n}\sum_{j=1}^n\zeta_j(t)\right\vert^{1/2},\,\,\,\zeta_j(t) = \cos\left[t(W_j-W_j^r)\right].
\end{align}
Next, we decompose the deconvolution kernel with the Fr\'{e}chet derivative as mentioned in \cite{dong2022nonparametric},
\begin{align}
    &\mathcal{K}_{\epsilon,1}(a)=\frac{1}{2\pi}\int \ee^{-\ii t^\top a}\frac{K^{\text{ft}}(t)}{f_\epsilon^{\text{ft}}(t/b)}\psi_1(t/b)\,dt,\quad \mathcal{K}_{\epsilon,1}(x) = b\mathcal{K}_{b,1}(x),\notag \\ 
    &\psi_1(t) = \frac{\left[f_\epsilon^{\text{ft}}(t)\right]^2 - \left[\hat{f}_\epsilon^{\text{ft}}(t)\right]^2}{2\left[f_\epsilon^{\text{ft}}(t)\right]^2},\label{Decompose of est of decon part1}\\
    &\mathcal{K}_{\epsilon,2}(a)=\frac{1}{2\pi}\int \ee^{-\ii t^\top a}\frac{K^{\text{ft}}(t)}{f_\epsilon^{\text{ft}}(t/b)}\psi_2(t/b)\,dt,\quad \mathcal{K}_{\epsilon,2}(x) = b\mathcal{K}_{b,2}(x),\notag \\ 
    &\psi_2(t) = \frac{\left[\hat{f}_\epsilon^{\text{ft}}(t)+2f_\epsilon^{\text{ft}}(t)\right]\left\{\left[\hat{f}_\epsilon^{\text{ft}}(t)\right]^2 - \left[f_\epsilon^{\text{ft}}(t)\right]^2\right\}^2}{2\hat{f}_\epsilon^{\text{ft}}(t)\left[f_\epsilon^{\text{ft}}(t)\right]^2\left[\hat{f}_\epsilon^{\text{ft}}(t)+f_\epsilon^{\text{ft}}(t)\right]^2}.\label{Decompose of est of decon part2}
\end{align}
For convenience, we define the following transformation of kernels,
\begin{align} \label{tsf of kernel}
    &K_{\epsilon,j}(x) = \frac{1}{2\pi}\int \ee^{-\ii tx}K^{ft}(t)\psi_j(t/b)\,dt, \,\,\, j=1,2.
\end{align}
Thus, we can decompose the estimation of the deconvolution kernel by the following equation,
\begin{align} \label{Decompose of est of decon}
    &\hat{\mathcal{K}}_\epsilon(x) = \mathcal{K}_\epsilon(x) + \mathcal{K}_{\epsilon,1}(x) + \mathcal{K}_{\epsilon,2}(x) .
\end{align}
We then list the notations related to the statistics in our testing framework. First, empirical processes without projection mentioned in Section \ref{sec.Test} and their bootstrap versions mentioned in Section \ref{sec.boot} are defined as follows:
\begin{align} \label{Stat without proj}
    &S_{n}(\xi,\theta)  = \frac{1}{n}\sum_{i=1}^n\int\left(Y_i-g(x;\theta)\right)\mathcal{K}_b\left(\frac{x-W_i}{b}\right)e^{ix\xi}\,dx ,\\
    &\hat{S}_{n}(\xi,\theta)  = \frac{1}{n}\sum_{i=1}^n\int\left(Y_i-g(x;\theta)\right)\hat{\mathcal{K}}_b\left(\frac{x-W_i}{b}\right)e^{ix\xi}\,dx ,\\
    &S_{n}^\ast(\xi,\theta)  = \frac{1}{n}\sum_{i=1}^n V_i\int\left(Y_i-g(x;\theta)\right)\mathcal{K}_b\left(\frac{x-W_i}{b}\right)e^{ix\xi}\,dx ,\\
    &\hat{S}_{n}^\ast(\xi,\theta)  = \frac{1}{n}\sum_{i=1}^n V_i\int\left(Y_i-g(x;\theta)\right)\hat{\mathcal{K}}_b\left(\frac{x-W_i}{b}\right)e^{ix\xi}\,dx .
\end{align}
Then we define the non-negligible main items brought by the projection and their bootstrap version,
\begin{align} \label{Stat with proj main}
    &M_{n}(\theta)  = \frac{1}{n}\sum_{i=1}^n\int\left(Y_i-g(x;\theta)\right)\frac{\partial g(x;\theta)}{\partial\theta}\mathcal{K}_b\left(\frac{x-W_i}{b}\right)\,dx ,\\
    &\hat{M}_{n}(\theta)  = \frac{1}{n}\sum_{i=1}^n\int\left(Y_i-g(x;\theta)\right)\frac{\partial g(x;\theta)}{\partial\theta}\hat{\mathcal{K}}_b\left(\frac{x-W_i}{b}\right)\,dx \\
    &M_{n}^\ast(\theta)  = \frac{1}{n}\sum_{i=1}^n V_i\int\left(Y_i-g(x;\theta)\right)\frac{\partial g(x;\theta)}{\partial\theta}\mathcal{K}_b\left(\frac{x-W_i}{b}\right)\,dx ,\\
    &\hat{M}_{n}^\ast(\theta)  = \frac{1}{n}\sum_{i=1}^n V_i\int\left(Y_i-g(x;\theta)\right)\frac{\partial g(x;\theta)}{\partial\theta}\hat{\mathcal{K}}_b\left(\frac{x-W_i}{b}\right)\,dx .
\end{align}
Next, we give the notation for the main components of the projection and bootstrap versions,
\begin{align*}
    &\Delta(\theta) =\mathbb{E}[\frac{\partial g(x;\theta)}{\partial\theta}\frac{\partial g(x;\theta)}{\partial\theta^\top}]=\int\frac{\partial g(x;\theta)}{\partial\theta}\frac{\partial g(x;\theta)}{\partial\theta^\top}f_X(x)\,dx,\\
    &\Delta_n(\theta)  = \frac{1}{n}\sum_{i=1}^n\int\frac{\partial g(x;\theta)}{\partial\theta}\frac{\partial g(x;\theta)}{\partial\theta^\top}\mathcal{K}_b\left(\frac{x-W_i}{b}\right)\,dx ,\\
    &\hat{\Delta}_n(\theta)  = \frac{1}{n}\sum_{i=1}^n\int \frac{\partial g(x;\theta)}{\partial\theta}\frac{\partial g(x;\theta)}{\partial\theta^\top}\hat{\mathcal{K}}_b\left(\frac{x-W_i}{b}\right)\,dx,\\
    &\Delta_n^\ast(\theta)  = \frac{1}{n}\sum_{i=1}^n V_i\int \frac{\partial g(x;\theta)}{\partial\theta}\frac{\partial g(x;\theta)}{\partial\theta^\top}\mathcal{K}_b\left(\frac{x-W_i}{b}\right)\,dx ,\\
    &\hat{\Delta}_n^\ast(\theta)  = \frac{1}{n}\sum_{i=1}^n V_i\int\frac{\partial g(x;\theta)}{\partial\theta}\frac{\partial g(x;\theta)}{\partial\theta^\top}\hat{\mathcal{K}}_b\left(\frac{x-W_i}{b}\right)\,dx.\\
    &G(\xi,\theta)=\mathbb{E}[\frac{\partial g(x;\theta)}{\partial\theta}\ee^{\ii X\xi}]=\int\frac{\partial g(x;\theta)}{\partial\theta}\ee^{\ii x\xi}f_X(x)\,dx,\\
    &G_n(\xi,\theta)  = \frac{1}{n}\sum_{i=1}^n\int \frac{\partial g(x;\theta)}{\partial\theta}\mathcal{K}_b\left(\frac{x-W_i}{b}\right)\ee^{\ii x\xi}\,dx ,\\
    &\hat{G}_n(\xi,\theta)  = \frac{1}{n}\sum_{i=1}^n\int \frac{\partial g(x;\theta)}{\partial\theta}\hat{\mathcal{K}}_b\left(\frac{x-W_i}{b}\right)\ee^{\ii x\xi}\,dx,\\
    &G_n^\ast(\xi,\theta)  = \frac{1}{n}\sum_{i=1}^n V_i\int \frac{\partial g(x;\theta)}{\partial\theta}\mathcal{K}_b\left(\frac{x-W_i}{b}\right)\ee^{\ii x\xi}\,dx ,\\
    &\hat{G}_n^\ast(\xi,\theta)  = \frac{1}{n}\sum_{i=1}^n V_i\int \frac{\partial g(x;\theta)}{\partial\theta}\hat{\mathcal{K}}_b\left(\frac{x-W_i}{b}\right)\ee^{\ii x\xi}\,dx.
\end{align*}
At the same time, we define our projection structure using the following notations,
\begin{align*}
    &\mathcal{P}_n(x;\xi,\theta)=\ee^{\ii x\xi}
    -\frac{\partial g(x;\theta)}{\partial\theta^\top}\Delta_n^{-1}(\theta)G_n(\xi,\theta),\\
    &\hat{\mathcal{P}}_n(x;\xi,\theta)=\ee^{\ii x\xi}
    -\frac{\partial g(x;\theta)}{\partial\theta^\top}\hat{\Delta}_n^{-1}(\theta)\hat{G}_n(\xi,\theta).
\end{align*}
Finally, we give the notation for the statistics constructed from projection-based empirical processes, both for known and unknown measurement errors.
\begin{align*}
    &S_{n}^{pro}(\xi,\hat\theta_n)=\frac{1}{n}\sum_{i=1}^n\int\left(Y_i-g(x;\hat{\theta}_n)\right)\mathcal{K}_b\left(\frac{x-W_i}{b}\right)\mathcal{P}_n(x;\xi,\theta)\,dx,\\
    &\hat{S}_{n}^{pro}(\xi,\hat\theta_n)=\frac{1}{n}\sum_{i=1}^n\int\left(Y_i-g(x;\hat{\theta}_n)\right)\hat{\mathcal{K}}_b\left(\frac{x-W_i}{b}\right)\hat{\mathcal{P}}_n(x;\xi,\theta)\,dx .
\end{align*}
Additionally, our statistics can be rewritten as,
\begin{align*}
    &S_{n}^{pro}(\xi,\hat{\theta}_n)  = S_{n}(\xi,\hat{\theta}_n) - M_{n}^\top(\hat{\theta}_n)\Delta^{-1}_n(\hat{\theta}_n)G_n(\xi,\hat{\theta}_n),\\
    &\hat{S}_{n}^{pro}(\xi,\hat{\theta}_n)  = \hat{S}_{n}(\xi,\hat{\theta}_n) - \hat{M}_{n}^\top(\hat{\theta}_n)\hat{\Delta}^{-1}_n(\hat{\theta}_n)\hat{G}_n(\xi,\hat{\theta}_n) .
\end{align*}

\section{Proofs of theorems}\label{sec.AppendixB}

\begin{proof}[Proof of Theorem \ref{theorem.ordinary smooth under H_0}]
    We first prove for the ordinary smooth case under the null hypothesis \eqref{hyp.null1}. First, we can rewrite our statistics as follows, using the notations from Appendix \ref{sec.AppendixA}, 
    \begin{align}\label{Decompose of statistics}
        &S_{n}^{pro}(\xi,\hat{\theta}_n)  = S_{n}(\xi,\hat{\theta}_n) - M_{n}^\top(\hat{\theta}_n)\Delta^{-1}_n(\hat{\theta}_n)G_n(\xi,\hat{\theta}_n).
    \end{align}
    Subsequently, \eqref{Decompose of statistics} along with Lemmas \ref{lemma.projG}--\ref{lemma.proj main term unknown} implies the further decomposition
    \begin{align*}
        \sup_{\xi\in\Pi}\left\vert S_{n}^{pro}(\xi,\hat{\theta}_n)-S_{n}(\xi,\theta_0)\right\vert=o_p\left( \left\vert M_n(\theta_0)\right\vert+\left\vert\kappa_n(\theta_0)\right\vert\right)+O_p\left(\left\vert \hat{\theta}_n-\theta_0\right\vert^2\right).
    \end{align*}
    Henceforth, $|c|$ is used to denote the Euclidean norm of a vector $c$ and $|A|$ is used to denote the Frobenius norm of a matrix $A$. Under Assumption \ref{ass.O}, designating $\partial^2 g(x;\theta)/\partial\theta_j\partial\theta_k$ and $\partial g(x;\theta)/\partial\theta_j$ as $h(x;\theta)$ in Lemma \ref{lemma.convergence of main term known} separately, where $j,k$ value from $1$ to $d$, we obtain $|M_n(\theta_0)| = O_p(n^{-1/2})$ and $|\kappa(\theta_0)| = O_p(n^{-1/2})$ and thus,
    \begin{align}\label{theoremproof.ordinary smooth under H_0 stat decomp}
        &\sup_{\xi\in\Pi}\left\vert S_{n}^{pro}(\xi,\hat{\theta}_n) - S_{n}^{pro}(\xi,\theta_0)\right\vert = o_p\left(n^{-\frac{1}{2}}\right) + O_p\left(\left\vert\hat{\theta}-\theta_0\right\vert^2\right).
    \end{align}
    In a similar manner to the proof of Lemma \ref{lemma.convergence of main term known} and with additional consideration of the null hypothesis,
    \begin{align}\label{theoremproof.ordinary smooth under H_0 main term decomp}
        \sup_{\xi\in\Pi}\left\vert S_{n}^{pro}(\xi,\theta_0) - \frac{1}{n}\sum_{i=1}^n \left\{r^{os}_{\infty}(d_i;\xi,\theta_0) - \mathbb{E}\left[r^{os}_{\infty}(d_i;\xi,\theta_0)\right]\right\}\right\vert = o_p\left(n^{-\frac{1}{2}}\right),\
    \end{align}
    for the ordinary smooth case, where $r^{os}_{\infty}(d_j;\xi,\theta)$ is defined in \eqref{theorem.ordinary smooth under H_0 ros}. Note that Assumption \ref{ass.O} implies $\mathbb{E}[\int_{\Pi}r^{os}_{\infty}(Y,W;\xi,\theta_0)^2d\xi]<\infty$, we obtain 
   \begin{align}\label{theoremproof.ordinary smooth under H_0 conv}
        \frac{1}{\sqrt{n}}\sum_{i=1}^n \left\{r^{os}_{\infty}(d_i;\cdot,\theta_0) - \mathbb{E}\left[r^{os}_{\infty}(d_i;\cdot,\theta_0)\right]\right\} \Longrightarrow S_{\infty}^{os}(\cdot,\theta_0)
    \end{align}
    by Theorem $3.9$ of \cite{chen1998central}. Combining \eqref{theoremproof.ordinary smooth under H_0 stat decomp}, \eqref{theoremproof.ordinary smooth under H_0 main term decomp}, and \eqref{theoremproof.ordinary smooth under H_0 conv}, we have finished our proof.
\end{proof}

\begin{proof}[Proof of Theorem \ref{theorem.supersmooth under H_0}]
    For the supersmooth case, we follow along the same line as Theorem \ref{theorem.ordinary smooth under H_0} except for using Assumption \ref{ass.S} instead of Assumption \ref{ass.O}. Therefore, under the null hypothesis, 
    \begin{align}\label{theoremproof.supersmooth under H_0 main term decomp}
        &\sup_{\xi\in\Pi}\left\vert S_{n}^{pro}(\xi,\theta_0) - \frac{1}{n}\sum_{i=1}^n \left\{r^{ss}_{\infty}(d_i;\xi,\theta_0) - \mathbb{E}\left[r^{ss}_{\infty}(d_i;\xi,\theta_0)\right]\right\}\right\vert= o_p\left(n^{-\frac{1}{2}}\right),
    \end{align}
    where $r^{ss}_{\infty}(d_j;\xi,\theta)$ is mentioned in \eqref{theorem.supersmooth under H_0 rss}. By similar arguments to the proof of Lemma \ref{lemma.convergence of main term known}, 
    \begin{align}\label{theoremproof.supersmooth under H_0 conv}
        \frac{1}{\sqrt{n}}\sum_{i=1}^n \left\{r^{ss}_{\infty}(d_i;\cdot,\theta_0) - \mathbb{E}\left[r^{ss}_{\infty}(d_i;\cdot,\theta_0)\right]\right\} \Longrightarrow S_{\infty}^{ss}(\cdot,\theta_0)
    \end{align}
    follows by $\mathbb{E}[\int_{\Pi}r^{ss}_{\infty}(Y,W;\xi,\theta_0)^2d\xi]<\infty$ which is guaranteed by Assumption \ref{ass.S}. Combining \eqref{theoremproof.ordinary smooth under H_0 stat decomp}, \eqref{theoremproof.supersmooth under H_0 main term decomp}, and \eqref{theoremproof.supersmooth under H_0 conv}, the conclusion holds.
\end{proof}

\begin{proof}[Proof of Theorem \ref{theorem.known under H_1n}]
    Under the local alternative hypothesis, we notice that our variable can be decomposed as $Y_i = Y_{i0} + n^{-1/2}\Delta(X_i)$, and our conclusion under the null hypothesis \eqref{hyp.null1} can be fully applied to $Y_{i0}$. Therefore, we can decompose our statistics as follows,
    \begin{align*}
        &S_{n}^{pro}(\xi,\hat{\theta}_n) = S_{n0}^{pro}(\xi,\hat{\theta}_n) + S_{\Delta,n}^{pro}(\xi,\hat{\theta}_n),
    \end{align*}
    where $S_{n0}^{pro}(\xi,\hat{\theta}_n)$ represents the statistics replacing $Y_{i}$ with $Y_{i0}$, and 
    \begin{align*}
        &\sqrt{n} S_{\Delta,n}^{pro}(\xi,\hat{\theta}_n) = \frac{1}{n}\sum_{i=1}^n\int\Delta(X_i)\mathcal{K}_b\left(\frac{x-W_i}{b}\right)\mathcal{P}_n(x;\xi,\hat{\theta}_n)\,dx\\
        = & S_{\Delta,n}(\xi,\hat{\theta}_n) - M_{\Delta,n}^\top(\hat{\theta}_n)\Delta^{-1}_n(\hat{\theta}_n)G_n(\xi,\hat{\theta}_n).
    \end{align*}
    For the main term, we still obtain our conclusions in Theorems \ref{theorem.ordinary smooth under H_0} and \ref{theorem.supersmooth under H_0}, that is to say,
    \begin{align*}
        &\sqrt n S_{n0}^{pro}(\cdot,\hat\theta_n)\Longrightarrow S_{\infty}^{os}(\cdot,\theta_0),\quad \sqrt n S_{n0}^{pro}(\cdot,\hat\theta_n)\Longrightarrow S_{\infty}^{ss}(\cdot,\theta_0).
    \end{align*}
    For the deterministic shift term, 
    \begin{align*}
        &\sup_{\xi\in\Pi}\left\vert\sqrt n S_{\Delta,n}^{pro}(\xi,\theta_0) - \mu_\Delta(\xi,\theta_0)\right\vert = o_p\left(1\right)
    \end{align*}
    follows by similar arguments to the proof of Lemma \ref{lemma.ULLN of main term} where $\mu_\Delta(\xi,\theta_0)$ is defined as,
    \begin{align*}
        \mu_\Delta(\xi,\theta_0) = \mathbb{E}\left\{\Delta(X)\left[e^{\ii X\xi}-\dot{g}^T(X;\theta_0)\Delta^{-1}(\theta_0)G(\xi,\theta_0)\right]\right\}.
    \end{align*}
    Notably, we use $h(x,\theta) = \Delta(x)$ in $S_{\Delta,n}(\xi,\theta_0)$, while $h(x,\theta) = \Delta(x)[\partial g(x;\theta)/\partial\theta]$ and $\xi=0$ in $M_{\Delta,n}(\theta)$. Thus, we now complete our proof of Theorem \ref{theorem.known under H_1n}.
\end{proof}

\begin{proof}[Proof of Theorem \ref{theorem.known alternative}]
    We follow the same line as proof of Theorem \ref{theorem.known under H_1n} to get consistency. Notice that $Y_i$ can be decomposed as $Y_i = Y_{i0} + m(X_i) - g(X_i;\theta^\ast)$, and our conclusion under the null hypothesis \eqref{hyp.null1} can be applied to $Y_{i0}$, we can decompose our statistics as follows,
    \begin{align*}
        &S_{n}^{pro}(\xi,\hat{\theta}_n) = S_{n0}^{pro}(\xi,\hat{\theta}_n) + S_{n1}^{pro}(\xi,\hat{\theta}_n),
    \end{align*}
    where $S_{n0}^{pro}(\xi,\hat{\theta}_n)$ is defined in Theorem \ref{theorem.known under H_1n}, and
    \begin{align*}
        &S_{n1}^{pro}(\xi,\hat{\theta}_n) = \frac{1}{n}\sum_{i=1}^n\int \left(m(X_i) - g(X_i;\theta^\ast)\right)\mathcal{K}_b\left(\frac{x-W_i}{b}\right)\mathcal{P}_n(x;\hat{\theta}_n,\xi)\,dx.
    \end{align*}
    By similar arguments to the proof of Theorem \ref{theorem.known under H_1n},
    \begin{align*}
        &\sup_{\xi\in\Pi}\left\vert S_{n1}^{pro}(\xi,\theta^\ast) - C(\xi,\theta^\ast)\right\vert = o_p(1),
    \end{align*}
    where
    \begin{align*}
        &C(\xi,\theta^\ast) = \mathbb{E}\left\{\left[m(X)-g(X;\theta^\ast)\right]\left[e^{\ii X\xi}-\frac{\partial g(X;\theta^\ast)}{\partial\theta^\top}\Delta^{-1}(\theta^\ast)G(\xi,\theta^\ast)\right]\right\}.
    \end{align*}
\end{proof}

\begin{proof}[Proof of Theorem \ref{theorem.unknown ordinary smooth under H_0}]
    For unknown measurement error under the null hypothesis, we can still rewrite,
    \begin{align}\label{Decompose of statistics unknown}
        &\hat{S}_{n}^{pro}(\xi,\hat{\theta}_n)  = \hat{S}_{n}(\xi,\hat{\theta}_n) - \hat{M}_{n}^\top(\hat{\theta}_n)\hat{\Delta}^{-1}_n(\hat{\theta}_n)\hat{G}_n(\xi,\hat{\theta}_n).
    \end{align}
    By similar arguments to the proof of Theorem \ref{theorem.ordinary smooth under H_0}, 
    \begin{align*}
        \sup_{\xi\in\Pi}\left\vert S_{n}^{pro}(\xi,\hat{\theta}_n)-S_{n}(\xi,\theta_0)\right\vert=o_p\left( \left\vert\hat{M}_n(\theta_0)\right\vert+\left\vert\hat{\kappa}_n(\theta_0)\right\vert\right)+O_p\left(\left\vert \hat{\theta}_n-\theta_0\right\vert^2\right)
    \end{align*}
    follows by Lemma \ref{lemma.projG}--\ref{lemma.proj main term unknown} and therefore,
    \begin{align}\label{theorem.unknown ordinary smooth under H_0 main term}
        &\sup_{\xi\in\Pi}\left\vert \hat{S}_{n}^{pro}(\xi,\hat{\theta}_n) - \hat{S}_{n}^{pro}(\xi,\theta_0)\right\vert = o_p\left(n^{-\frac{1}{2}}\right) + O_p\left(\left\vert \hat{\theta}_n-\theta_0\right\vert^2\right).
    \end{align}
    Subsequently, under the null hypothesis, following the proof of Lemma \ref{lemma.convergence of main term unknown}, 
    \begin{align}\label{theorem.unknown ordinary smooth under H_0 main term decomp}
        &\sup_{\xi\in\Pi}\left\vert \hat{S}_{n}^{pro}(\xi,\theta_0)-\frac{1}{n}\sum_{i=1}^n \left\{\hat{r}^{os}_{\infty}(d_i;\xi,\theta_0) - \mathbb{E}\left[\hat{r}^{os}_{\infty}(d_i;\xi,\theta_0)\right]\right\}\right\vert=o_p\left(n^{-\frac{1}{2}}\right).
    \end{align}
    where $\hat{r}^{os}_{\infty}(d_i;\xi,\theta)$ is defined by \eqref{theorem.unknown ordinary smooth under H_0 roshat} and \eqref{theorem.unknown ordinary smooth under H_0 rplus}. With Assumption \ref{ass.O'} which implies $\mathbb{E}[\int_{\Pi}\hat{r}^{os}_{\infty}(Y,W;\xi,\theta_0)^2d\xi]<\infty$, 
    \begin{align}\label{theoremproof.unknown ordinary smooth under H_0 conv}
        \frac{1}{\sqrt{n}}\sum_{i=1}^n \left\{\hat{r}^{os}_{\infty}(d_i;\cdot,\theta_0) - \mathbb{E}\left[\hat{r}^{os}_{\infty}(d_i;\cdot,\theta_0)\right]\right\} \Longrightarrow \hat{S}_{\infty}^{os}(\cdot,\theta_0)
    \end{align}
    holds by Theorem $3.9$ of \cite{chen1998central}. Combining \eqref{theorem.unknown ordinary smooth under H_0 main term}, \eqref{theorem.unknown ordinary smooth under H_0 main term decomp} and \eqref{theoremproof.unknown ordinary smooth under H_0 conv}, we finish our proof of \eqref{result of main term unknown ass O H0}.
\end{proof}

\begin{proof}[Proof of Theorem \ref{theorem.unknown supersmooth under H_0}]
    Now we consider the supersmooth case under the null hypothesis \eqref{hyp.null1}. By similar arguments to the proof of Theorem \ref{theorem.supersmooth under H_0}, we follow the same line as Theorem \ref{theorem.unknown ordinary smooth under H_0} except for using Assumption \ref{ass.S'} instead of Assumption \ref{ass.O'}. Consequently,
    \begin{align}\label{theorem.unknown supersmooth under H_0 main term decomp}
        &\sup_{\xi\in\Pi}\left\vert \hat{S}_{n}^{pro}(\xi,\theta_0)-\frac{1}{n}\sum_{i=1}^n \left\{\hat{r}^{ss}_{\infty}(d_i;\xi,\theta_0) - \mathbb{E}\left[\hat{r}^{ss}_{\infty}(d_i;\xi,\theta_0)\right]\right\}\right\vert=o_p\left(n^{-\frac{1}{2}}\right),
    \end{align}
    where $\hat{r}^{ss}_{\infty}(d_i;\xi,\theta)$ is defined by \eqref{theorem.unknown supersmooth under H_0 rss}. Meanwhile,
    \begin{align}\label{theoremproof.unknown supersmooth under H_0 conv}
        \frac{1}{\sqrt{n}}\sum_{i=1}^n \left\{\hat{r}^{ss}_{\infty}(d_i;\cdot,\theta_0) - \mathbb{E}\left[\hat{r}^{ss}_{\infty}(d_i;\cdot,\theta_0)\right]\right\} \Longrightarrow \hat{S}_{\infty}^{ss}(\cdot,\theta_0)
    \end{align}
    holds by $\mathbb{E}[\int_{\Pi}\hat{r}^{ss}_{\infty}(Y,W;\xi,\theta_0)^2d\xi]<\infty$ which is guaranteed by Assumption \ref{ass.S'}. Thus, \eqref{result of main term unknown ass S H0} follows by combining \eqref{theorem.unknown ordinary smooth under H_0 main term}, \eqref{theorem.unknown supersmooth under H_0 main term decomp} and \eqref{theoremproof.unknown supersmooth under H_0 conv}.
\end{proof}

\begin{proof}[Proof of Theorem \ref{theorem.unknown under H_1n}]
    Similar to proof of \eqref{theorem.known under H_1n}, we decompose $Y_i$ as $Y_{i0} + n^{-1/2}\Delta(X_i)$, and our conclusion under the null hypothesis can be applied to $Y_{i0}$. Therefore, we can decompose,
    \begin{align*}
        &\hat{S}_{n}^{pro}(\xi,\hat{\theta}_n) = \hat{S}_{n0}^{pro}(\xi,\hat{\theta}_n) + \hat{S}_{\Delta,n}^{pro}(\xi,\hat{\theta}_n),
    \end{align*}
    where $\hat{S}_{n0}^{pro}(\xi,\hat{\theta}_n)$ represents the statistics replacing $Y_{i}$ with $Y_{i0}$ and,
    \begin{align*}
        &\sqrt{n} \hat{S}_{\Delta,n}^{pro}(\xi,\hat{\theta}_n) = \frac{1}{n}\sum_{i=1}^n\int\Delta(X_i)\hat{\mathcal{K}}_b\left(\frac{x-W_i}{b}\right)\hat{\mathcal{P}}_n(x;\xi,\hat{\theta}_n)\,dx\\
        = & \hat{S}_{\Delta,n}(\xi,\hat{\theta}_n) - \hat{M}_{\Delta,n}^\top(\hat{\theta}_n)\hat{\Delta}^{-1}_n(\hat{\theta}_n)\hat{G}_n(\xi,\hat{\theta}_n).
    \end{align*}
    Thus, 
    \begin{align*}
        &\sup_{\xi\in\Pi}\left\vert \sqrt n \hat{S}_{\Delta,n}^{pro}(\xi,\theta_0) - \mu_\Delta(\xi,\theta_0)\right\vert = o_p\left(1\right)
    \end{align*}
    follows by using $h(x,\theta) = \Delta(x)$ for $\hat{S}_{\Delta,n}(\cdot,\theta)$, while $h(x,\theta) = \Delta(x)[\partial g(x;\theta)/\partial\theta]$ and $\xi=0$ for $M_{\Delta,n}(\theta)$ in Lemma \ref{lemma.ULLN of main term unknown}, where $\theta$ values between $\theta_0$ and $\hat{\theta}_n$. Together with our conclusions in Theorems \ref{theorem.ordinary smooth under H_0} and \ref{theorem.supersmooth under H_0}, conclusions in Theorem \ref{theorem.unknown under H_1n} hold.
\end{proof}

\begin{proof}[Proof of Theorem \ref{theorem.unknown alternative}]
    By similar arguments to the proof of Theorem \ref{theorem.known alternative}, we decompose
    \begin{align*}
        &\hat{S}_{n}^{pro}(\xi,\hat{\theta}_n) = \hat{S}_{n0}^{pro}(\xi,\hat{\theta}_n) + \hat{S}_{n1}^{pro}(\xi,\hat{\theta}_n),
    \end{align*}
    where $\hat{S}_{n0}^{pro}(\xi,\hat{\theta}_n)$ is defined in Theorem \ref{theorem.unknown under H_1n}, and
    \begin{align*}
        &\hat{S}_{n1}^{pro}(\xi,\hat{\theta}_n) = \frac{1}{n}\sum_{i=1}^n\int \left(m(X_i) - g(X_i;\theta^\ast)\right)\hat{\mathcal{K}}_b\left(\frac{x-W_i}{b}\right)\hat{\mathcal{P}}_n(x;\xi,\hat{\theta}_n)\,dx.
    \end{align*}
    With Lemma \ref{lemma.ULLN of main term unknown}, following the same line as the proof of Theorem \ref{theorem.unknown under H_1n}, our conclusions hold.
\end{proof}

\begin{proof}[Proof of Theorem \ref{theorem.boot known}]
    Under the null hypothesis, we start by decomposing
    \begin{align}\label{Decompose of boot}
        &S_{n}^{pro,\ast}(\xi,\hat{\theta}_n)  = S_{n}^\ast(\xi,\hat{\theta}_n) - M_{n}^{\ast\top}(\hat{\theta}_n)\Delta^{-1}_n(\hat{\theta}_n)G_n(\xi,\hat{\theta}_n).
    \end{align}
    With Lemma \ref{lemma.projG}--\ref{lemma.proj main term unknown}, we obtain
    \begin{align*}
        &\sup_{\xi\in\Pi}\left\vert S_{n}^{pro,\ast}(\xi,\hat{\theta}_n)-S^{pro,\ast}_{n}(\xi,\theta_0)\right\vert\\
        =&o_p\left(\left\vert M_n^\ast(\theta_0)\right\vert+\left\vert\kappa^\ast(\theta_0)\right\vert+\left\vert\Delta_n^\ast(\theta_0)\right\vert\right)+o_p\left(\sup_{\xi\in\Pi}\left\vert G_n(\xi,\theta_0)\right\vert\right)+ O_p\left(\left\vert \hat{\theta}_n-\theta_0\right\vert^2\right)\\
        =&O_p\left(\left\vert \hat{\theta}_n-\theta_0\right\vert^2\right)+o_p\left(n^{-\frac{1}{2}}\right),
    \end{align*}
    where the last equation follows from Lemma \ref{lemma.convergence of main term known}. Consequently,
    \begin{align*}
        &\sup_{\xi\in\Pi}\left\vert S_{n}^{pro,\ast}(\xi,\hat{\theta}_n) - S_{n}^{pro,\ast}(\xi,\theta_0)\right\vert = o_p\left(n^{-\frac{1}{2}}\right) + O_p\left(\left\vert \hat{\theta}_n-\theta_0\right\vert^2\right).
    \end{align*}
    It is worth noting that results for the statistics mentioned in Theorem \ref{theorem.ordinary smooth under H_0} and \ref{theorem.supersmooth under H_0} still hold for the bootstrap versions because of the unit variance of the multiplier $V_i$. Consequently, $\sqrt n S_{n}^{pro,\ast}(\cdot,\hat\theta_n)\overset{\ast}{\Longrightarrow} S_{\infty}^{os}(\cdot,\theta_0)$ holds for the ordinary smooth case and $\sqrt n S_{n}^{pro,\ast}(\cdot,\hat\theta_n)\overset{\ast}{\Longrightarrow} S_{\infty}^{ss}(\cdot,\theta_0)$ for the supersmooth case and thus \eqref{result of main term ass O boot} and \eqref{result of main term ass S boot} hold. 

    Under the alternative hypothesis, by similar arguments to our proof of Theorem \ref{theorem.known under H_1n}, we obtain the following decomposition,
    \begin{align*}
        &S_{n}^{pro,\ast}(\xi,\hat{\theta}_n) = S_{n0}^{pro,\ast}(\xi,\hat{\theta}_n) + S_{\Delta,n}^{pro,\ast}(\xi,\hat{\theta}_n),
    \end{align*}
    where $Y_i=Y_{i0} + n^{-1/2}\Delta(X_i)$, $S_{n0}^{pro,\ast}(\xi,\hat{\theta}_n)$ represents $S_{n}^{pro,\ast}(\xi,\hat{\theta}_n)$ replacing $Y_{i}$ with $Y_{i0}$ and,
    \begin{align*}
        &\sqrt{n} S_{\Delta,n}^{pro,\ast}(\xi,\hat{\theta}_n) = \frac{1}{n}\sum_{i=1}^nV_i\int\Delta(X_i)\mathcal{K}_b\left(\frac{x-W_i}{b}\right)\mathcal{P}_n(x;\xi,\hat{\theta}_n)\,dx.
    \end{align*}
    In contrast to the earlier conclusion, deterministic shift term $S_{\Delta,n}^{pro,\ast}(\cdot,\hat{\theta}_n)$ is asymptotically negligible with Lemma \ref{lemma.ULLN of main term}. Consequently, the bootstrap version of the test statistics converges to the same limiting process under both $H_0$ and $H_{1n}$, which naturally leads to the conclusion.
\end{proof}

\begin{proof}[Proof of Theorem \ref{theorem.boot unknown}]
    By similar arguments to the proof of Theorem \ref{theorem.boot known},
    \begin{align*}
        &\hat{S}_{n}^{pro,\ast}(\xi,\hat{\theta}_n)  = \hat{S}_{n}^\ast(\xi,\hat{\theta}_n) - \hat{M}_{n}^{\ast\top}(\hat{\theta}_n)\hat{\Delta}^{-1}_n(\hat{\theta}_n)\hat{G}_n(\xi,\hat{\theta}_n).
    \end{align*}
    Subsequently, we analyze $\hat{S}_{n}^\ast(\cdot,\hat{\theta}_n)$, $\hat{M}_{n}^{\ast}(\hat{\theta}_n)$, $\hat{\Delta}^{-1}_n(\hat{\theta}_n)$ and $\hat{G}_n(\cdot,\hat{\theta}_n)$ using Lemmas \ref{lemma.projG}--\ref{lemma.proj main term unknown}. Together with $|\hat{M}_n^\ast(\theta_0)| = O_p(n^{-1/2})$, $|\hat{\kappa}^\ast(\theta_0)| = O_p(n^{-1/2})$, $\sup_{\xi\in\Pi}|\hat{G}_n^\ast(\xi,\theta_0)| = O_p(n^{-1/2})$, and $|\hat{\Delta}_n^\ast(\theta_0)| = O_p(n^{-1/2})$ implied by Lemma \ref{lemma.convergence of main term unknown}, we obtain
    \begin{align*}
        &\sup_{\xi\in\Pi}\left\vert \hat{S}_{n}^{pro,\ast}(\xi,\hat{\theta}_n) - \hat{S}_{n}^{pro,\ast}(\xi,\theta_0)\right\vert = o_p\left(n^{-\frac{1}{2}}\right) + O_p\left(\left\vert \hat{\theta}_n-\theta_0\right\vert^2\right).
    \end{align*}
    Given the unit variance property of $V_i$, $\sqrt n \hat{S}_{n}^{pro,\ast}(\cdot,\hat\theta_n)\overset{\ast}{\Longrightarrow}\hat{S}_{\infty}^{os}(\cdot,\theta_0)$ holds for the ordinary smooth case and $\sqrt n \hat{S}_{n}^{pro,\ast}(\cdot,\hat\theta_n)\overset{\ast}{\Longrightarrow} \hat{S}_{\infty}^{ss}(\cdot,\theta_0)$ for the supersmooth case. Under the local alternative hypothesis, we decompose the statistics as done in the proof of Theorem \ref{theorem.boot known} and follow the same lines. Thus, the validity of the bootstrap procedure can still be established.
\end{proof}

\section{Lemmas and proofs}\label{sec.AppendixC}

\begin{lemma}\label{lemma.convergence of main term known}
    Suppose that Assumption \ref{ass.D} holds, together with either Assumption \ref{ass.O} for the ordinary smooth case or Assumption \ref{ass.S} for the supersmooth case,
    \begin{align}\label{lemma.convergence of main term known eq1}
        &\sup_{\xi\in\Pi}\left\vert
        \begin{aligned}
            \frac{1}{n}\sum_{i=1}^n \int \left(Y_i-g(x;\theta_0)\right)h(x;\theta_0)\mathcal{K}_b\left(\frac{x-W_i}{b}\right)\ee^{\ii x\xi}\,dx\\
            -\int \mathbb{E}\left[\left(Y-g(X;\theta_0)\right)K_b\left(\frac{x-X}{b}\right)\right]h(x;\theta_0)\ee^{\ii x\xi}\,dx
        \end{aligned}
        \right\vert=O_p\left(n^{-\frac{1}{2}}\right).
    \end{align}
    For the bootstrap version,
    \begin{align}\label{lemma.convergence of main term known eq2}
        &\sup_{\xi\in\Pi}\left\vert\frac{1}{n}\sum\limits_{i=1}^n V_i\int \left(Y_i-g(x;\theta_0)\right)h(x;\theta_0)\mathcal{K}_b\left(\frac{x-W_i}{b}\right)\ee^{\ii x\xi}\,dx\right\vert = O_p\left(n^{-\frac{1}{2}}\right).
    \end{align}
\end{lemma}

\begin{proof}[Proof of Lemma \ref{lemma.convergence of main term known}]
    Define $d_i = (Y_i,W_i)$ for $i=1,2,\cdots,n$. We first consider the ordinary smooth case and show the properties of the bias term as follows,
    \begin{align*}
        &\mathbb{E}\left[\frac{1}{n}\sum_{i=1}^n \int \left(Y_i-g(x;\theta_0)\right)h(x;\theta_0)\mathcal{K}_b\left(\frac{x-W_i}{b}\right)\ee^{\ii x\xi}\,dx\right]\\
        = &\mathbb{E}\left[\int \left(Y-g(X;\theta_0)\right)h(x;\theta_0)\mathcal{K}_b\left(\frac{x-W}{b}\right)\ee^{\ii x\xi}\,dx\right]\\
        & + \int h(x;\theta_0)\left[J_{1,n}(x)-g(x;\theta_0)J_{2,n}(x)\right]\ee^{\ii x\xi}\,dx,
    \end{align*}
    where,
    \begin{align}\label{Def of exp}
    &J_{1,n}(x) = \mathbb{E}\left[g(X;\theta_0)\mathcal{K}_b\left(\frac{x-W}{b}\right)\right],\,\,\,  J_{2,n}(x) = \mathbb{E}\left[\mathcal{K}_b\left(\frac{x-W}{b}\right)\right].
    \end{align}
    Under Assumption \ref{ass.O}, decomposing $J_{1,n}$ as
    \begin{align*}
        J_{1,n}(x) =& \frac{1}{b}\mathbb{E}\left[g(X;\theta_0)K\left(\frac{x-X}{b}\right)\right]=\int \frac{1}{b}g(y;\theta_0)K\left(\frac{x-y}{b}\right)f_X(y)\,dy\\
        =& \int K(y)g(x-by;\theta_0)f_X(x-by)\,dy\\
        =& g(x;\theta_0)f_X(x) + \frac{b^p}{p!}\int \left[gf_X\right]^{(p)}(\tilde{y})K(y)y^p\,dy = g(x;\theta_0)f_X(x)+b^p\Delta_1(x),
    \end{align*}
    where $\tilde{y}\in [x-by,x]$ and $[gf_X]^{(p)}$ represents the $p$-times derivative of $g(x,\theta_0)f_X(x)$ with respect to $x$. $\Delta_1(x)$ is a new notation we define here to represent the remainder. By similar arguments,
    \begin{align}\label{lemmaproof.convergence of main term known J2n}
        &J_{2,n}(x) = f_X(x)+b^p\Delta_2(x).
    \end{align}
    Define $\mu_K^p = \int k(u)\left\vert u\right\vert^p du$ and with the Lipschitz continuity mentioned in Assumption \ref{ass.O}. We have,
    \begin{align*}
        &\vert \Delta_1(x)\vert \leq \frac{1}{p!}\left\{ \left[gf_X\right]^{(p)}(x)\mu_K^p+ bL_{[gf_X]^{(p)}}(x)\mu_K^{p+1}\right\},\\
        &\vert \Delta_2(x)\vert \leq \frac{1}{p!}\left\{ f_X^{(p)}(x)\mu_K^p+ bL_{f_X^{(p)}}(x)\mu_K^{p+1}\right\}.
    \end{align*}
    Thus, with the integral properties and the undersmooth assumption in Assumption \ref{ass.O},
    \begin{align}\label{lemmaproof.convergence of main term known bias eq1}
        &\sup_{\xi\in\Pi}\left\vert\mathbb{E}\left[\int \left(g(X;\theta_0)-g(x;\theta_0)\right)h(x;\theta_0)\mathcal{K}_b\left(\frac{x-W}{b}\right)\ee^{\ii x\xi}\,dx\right]\right\vert = O\left(b^p\right) = o\left(n^{-\frac{1}{2}}\right).
    \end{align}
    Hence,
    \begin{align}\label{Cov of Mean under Ass O}
        &\sup_{\xi\in\Pi}\left\vert
        \begin{aligned}
            \mathbb{E}\left[\int \left(Y-g(x;\theta_0)\right)h(x;\theta_0)\mathcal{K}_b\left(\frac{x-W}{b}\right)\ee^{\ii x\xi}\,dx\right]\\
            -\int \mathbb{E}\left[\left(Y-g(X;\theta_0)\right)K_b\left(\frac{x-X}{b}\right)\right]h(x;\theta_0)\ee^{\ii x\xi}\,dx
        \end{aligned}
        \right\vert=o_p\left(n^{-\frac{1}{2}}\right).
    \end{align}
    For the supersmooth case, we notice
    \begin{align*}
        &J_{2,n}(x) = \mathbb{E}\left[\mathcal{K}_b\left(\frac{x-W}{b}\right)\right] =  \frac{1}{b}\mathbb{E}\left[K\left(\frac{x-X}{b}\right)\right]\\
        =& \int \frac{1}{b}K\left(\frac{x-y}{b}\right)f_X(y)\,dy = \int K(y)f_X(x-by)\,dy = f_X(x),
    \end{align*}
    and by the same line, $J_{1,n}(x) = g(x;\theta_0)f_X(x)$. Consequently, \eqref{Cov of Mean under Ass O} still follows by
    \begin{align*}
        &\mathbb{E}\left[\int \left(g(X;\theta_0)-g(x;\theta_0)\right)h(x;\theta_0)\mathcal{K}_b\left(\frac{x-W}{b}\right)e^{ix\xi}\,dx\right] = 0.
    \end{align*}

    For the main term of our statistics, under Assumption \ref{ass.O} for the ordinary smooth case where we have the $p$-th order differentiable $g,f$ and $h$,
    \begin{align}\label{lemmaproof.convergence of main term known decompose of statistics}
        &\frac{1}{n}\sum_{i=1}^n\int \left(Y_i-g(x;\theta_0)\right)h(x;\theta_0)\mathcal{K}_b\left(\frac{x-W_i}{b}\right)\ee^{\ii x\xi}\,dx \notag\\
        = & \frac{1}{n}\sum_{i=1}^n \ee^{\ii W_i\xi}\int \left(Y_i-g(W_i+bx;\theta_0)\right)h(W_i+bx;\theta_0)\mathcal{K}_\epsilon(x)\ee^{\ii bx\xi}\,dx\notag\\
        = & \frac{1}{n}\sum_{i=1}^n r_{1,n}(d_i; \xi) + \frac{1}{n}\sum_{i=1}^n t_{1,n}(d_i; \xi), 
    \end{align}
    where
    \begin{align*}
        &r_{1,n}(d_i; \xi) = \ee^{\ii W_i\xi}\sum_{l=0}^{p-1}\left[\left(Y_i-g(W_i;\theta_0)\right)h(W_i;\theta_0)\right]^{(l)}\frac{b^l}{l!}\int x^l e^{\ii b\xi x}\mathcal{K}_\epsilon(x)\,dx,\\
        &t_{1,n}(d_i; \xi) =  \ee^{\ii W_i\xi}\frac{b^p}{p!}\int \left[\left(Y_i-g(\tilde{W}_i;\theta_0)\right)h(\tilde{W}_i;\theta_0)\right]^{(p)}x^p \ee^{\ii bx \xi}\mathcal{K}_\epsilon(x)\,dx,
    \end{align*}
    and $\tilde{W}_i$ values between $W_i$ and $W_i+bx$. We show that the variance of residual term $t_{1,n}$ is negligible, i.e.,
    \begin{align}\label{lemmaproof.convergence of main term known negligible residual}
        &\sup_{\xi \in \Pi}\left\vert \frac{1}{n}\sum_{i=1}^n \left\{ t_{1,n}(d_i;\xi) - \mathbb{E}\left[t_{1,n}(d_i;\xi)\right]\right\}\right\vert = o_p\left(n^{-\frac{1}{2}}\right),
    \end{align}
    which follows by
    \begin{align*}
        &Var\left(\sup_{\xi \in \Pi}\left\vert \frac{1}{\sqrt{n}}\sum_{i=1}^n  t_{1,n}(d_i;\xi)\right\vert\right) \leq \mathbb{E}\left[\sup_{\xi \in \Pi}t_{1,n}(d_i;\xi)^2\right] \\
        = & O\left(b^{2p} \mathbb{E}\left\{\int \left\vert\left[\left(Y-g(\tilde{W};\theta_0)\right)h(\tilde{W};\theta_0)\right]^{(p)}\right\vert\times\left\vert x^p\mathcal{K}_\epsilon(x) \right\vert\,dx\right\}^2\right)\\
        \leq & O\left(b^{2p} \mathbb{E}\left\{\int \left\vert\left[\left(Y-g(W;\theta_0)\right)h(W;\theta_0)\right]^{(p)}\right\vert\times\left\vert x^p\mathcal{K}_\epsilon(x) \right\vert\,dx\right\}^2\right)\\
        & + O\left(b^{2p+2} \mathbb{E}\left\{\int m(Y,W) \times\left\vert x^{p+1}\mathcal{K}_\epsilon(x) \right\vert\,dx\right\}^2\right)\\
        & + O\left(b^{2p+1} \mathbb{E}\left\{\int \left\vert\left[\left(Y-g(W;\theta_0)\right)h(W,\theta_0)\right]^{(p)}\right\vert \times m(Y,W) \times\left\vert x^{p+1}\mathcal{K}_\epsilon(x) \right\vert\,dx\right\}\right)\\
        = &O\left(b^{2(p-\alpha)}\right) = o\left(1\right),
    \end{align*}
    where $m(Y,W)$ represents $\vert YL_{h^{(p)}}(W)\vert+\vert L_{[gh]^{(p)}}(W)\vert$ and the existence of the second moment of $m(Y,W)$ as assumed in Assumption \ref{ass.D} and \ref{ass.O}. Equation in the last line follows by $\sup_{0\leq l \leq p+1}\int \vert x^l\mathcal{K}(x)\vert \,dx = O(b^{-\alpha})$ mentioned in Lemma 4 of \cite{dong2022nonparametric}. For the main term $r_{1,n}$, given \eqref{Power of decon without ME Assm O} in Lemma \ref{lemma.power of decon},
    \begin{align*}
        &Var\left(\sup_{\xi \in \Pi}\vert \frac{1}{\sqrt{n}}\sum_{i=1}^n  \left[r_{1,n}(d_i;\xi) - r_{1,\infty}(d_i;\xi)\right]\vert\right) \\
        \leq & \mathbb{E}\left[\sup_{\xi \in \Pi}\left[r_{1,n}(d_i;\xi)-r_{1,\infty}(d_i;\xi)\right]^2\right] = O\left(b^2\right) = o\left(1\right),
    \end{align*}
    where $r_{1,\infty}(d_i;\xi)= \sum_{l=0}^{\alpha}c_l^{os}(\xi)\left[(Y_i-g(W_i;\theta_0))h(W_i;\theta_0)\right]^{(l)}\ee^{\ii W_i\xi}$. Therefore,
    \begin{align}\label{lemmaproof.convergence of main term known negligible main and residual}
        &\sup_{\xi \in \Pi}\left\vert \frac{1}{n}\sum_{i=1}^n \left\{ r_{1,n}(d_i;\xi) - r_{1,\infty}(d_i;\xi) - \mathbb{E}\left[r_{1,n}(d_i;\xi)- r_{1,\infty}(d_i;\xi)\right]\right\}\right\vert = o_p\left(n^{-\frac{1}{2}}\right).
    \end{align}
    Noting that Assumption \ref{ass.O}(v) implies $\mathbb{E}\left[\int_\Pi r_{1,\infty}(d_i;\xi)^2d\xi\right]<\infty$, we obtain,
    \begin{align}\label{lemmaproof.convergence of main term known main}
        &\sup_{\xi\in\Pi}\left\vert\frac{1}{n}\sum_{i=1}^n  \left\{r_{1,\infty}(d_i;\xi) - \mathbb{E}\left[r_{1,\infty}(d_i;\xi)\right]\right\}\right\vert = O_p\left(n^{-\frac{1}{2}}\right).
    \end{align}
    by Theorem 3.9 of \cite{chen1998central}. Combining \eqref{Cov of Mean under Ass O}, \eqref{lemmaproof.convergence of main term known decompose of statistics}, \eqref{lemmaproof.convergence of main term known negligible main and residual}, and \eqref{lemmaproof.convergence of main term known main}, \eqref{lemma.convergence of main term known eq1} follows for the ordinary smooth case. 
    
    For the supersmooth case, we still have the expansion of the main terms,
    \begin{align*}
        &\frac{1}{n}\sum_{i=1}^n\int \left(Y_i-g(x;\theta_0)\right)h(x;\theta_0)\mathcal{K}_b\left(\frac{x-W_i}{b}\right)\ee^{\ii x\xi}\,dx \\
        = & \frac{1}{n}\sum_{i=1}^n \ee^{\ii W_i\xi}\int \left(Y_i-g(W_i+bx;\theta_0)\right)h(W_i+bx;\theta_0)\mathcal{K}_\epsilon(x)\ee^{\ii bx\xi}\,dx\\
        = & \frac{1}{n}\sum_{i=1}^n\ee^{\ii W_i\xi}\sum_{l=0}^{\infty}\left[\left(Y_i-g(W_i;\theta_0)\right)h(W_i;\theta_0)\right]^{(l)}\frac{b^l}{l!}\int x^l e^{ib\xi x}\mathcal{K}_\epsilon(x)\,dx \\
        = & \frac{1}{n}\sum_{i=1}^n r_{2,n}(d_i; \xi) = \frac{1}{n}\sum_{i=1}^n r_{2,\infty}(d_i; \xi), 
    \end{align*}
    where the last equation follows by \eqref{Power of decon without ME Assm S} and definition,
    \begin{align*}
        r_{2,\infty}(d_i;\xi) = \sum_{l=0}^{\infty}c_l^{ss}(\xi)\left[\left(Y_i-g(W_i;\theta_0)\right)h(W_i;\theta_0)\right]^{(l)}\ee^{\ii W_i\xi}.
    \end{align*}
    Noting that Assumption \ref{ass.S}(v) implies $\mathbb{E}\left[\int_\Pi r_{2,\infty}(d_i;\xi)^2d\xi\right]<\infty$, by similar arguments to ordinary smooth case,
    \begin{align}\label{lemmaproof.convergence of main term known main super}
        &\sup_{\xi\in\Pi}\left\vert\frac{1}{n}\sum_{i=1}^n  \left\{r_{2,\infty}(d_i;\xi) - \mathbb{E}\left[r_{2,\infty}(d_i;\xi)\right]\right\}\right\vert = O_p\left(n^{-\frac{1}{2}}\right).
    \end{align}
    Combining \eqref{Cov of Mean under Ass O} and \eqref{lemmaproof.convergence of main term known main super}, we have finished our proof of \eqref{lemma.convergence of main term known eq1} for the supersmooth case. For bootstrap version, notice the independence between $(Y,X,W)$ and $V$, together with unit variance of $V$, we claim our proofs above still hold except
    \begin{align*}
        \int \mathbb{E}\left[V\left(Y-g(X;\theta_0)\right)K_b\left(\frac{x-X}{b}\right)\right]h(x;\theta_0)\ee^{\ii x\xi}\,dx = 0,
    \end{align*}
    which implies
    \begin{align}\label{Cov of Mean under Ass O boot}
        &\sup_{\xi\in\Pi}\left\vert
        \mathbb{E}\left[\int V\left(Y-g(x;\theta_0)\right)h(x;\theta_0)\mathcal{K}_b\left(\frac{x-W}{b}\right)\ee^{\ii x\xi}\,dx\right]
        \right\vert=o_p\left(n^{-\frac{1}{2}}\right).
    \end{align}
    Thus, \eqref{lemma.convergence of main term known eq2} holds.
\end{proof}

\begin{lemma}\label{lemma.convergence of main term unknown}
    Suppose that Assumption \ref{ass.D} and \ref{ass.D'} hold, along with either Assumption \ref{ass.O} and \ref{ass.O'} for the ordinary smooth case or Assumption \ref{ass.S} and \ref{ass.S'} for the supersmooth case,
    \begin{align}\label{lemma.convergence of main term unknown eq1}
        &\sup_{\xi\in\Pi}\left\vert
        \begin{aligned}
            &\frac{1}{n}\sum_{i=1}^n \int \left(Y_i-g(x;\theta_0)\right)h(x;\theta_0)\hat{\mathcal{K}}_b\left(\frac{x-W_i}{b}\right)\ee^{\ii x\xi}\,dx\\
            &-\int \mathbb{E}\left[\left(Y-g(X;\theta_0)\right)K_b\left(\frac{x-X}{b}\right)\right]h(x;\theta_0)\ee^{\ii x\xi}\,dx\\
            &-\int \mathbb{E}\left[\left(Y-g(X;\theta_0)\right)K_{b,2}\left(\frac{x-X}{b}\right)\right]h(x;\theta_0)\ee^{\ii x\xi}\,dx
        \end{aligned}
        \right\vert=O_p\left(n^{-\frac{1}{2}}\right),
    \end{align}
    and 
    \begin{align}\label{lemma.convergence of main term unknown eq2}
        &\sup_{\xi\in\Pi}\left\vert\frac{1}{n}\sum_{i=1}^n V_i\int \left(Y_i-g(x;\theta_0)\right)h(x;\theta_0)\hat{\mathcal{K}}_b\left(\frac{x-W_i}{b}\right)\ee^{\ii x\xi}\,dx\right\vert=O_p\left(n^{-\frac{1}{2}}\right).
    \end{align}
\end{lemma}

\begin{proof}[Proof of Lemma \ref{lemma.convergence of main term unknown}]
    In the absence of information about distribution of measurement error, we start be decomposing $\hat{S}_n(\xi,\theta_0) - S_n(\xi,\theta_0)= S_{n,1}(\xi,\theta_0)+S_{n,2}(\xi,\theta_0)$, where
    \begin{align}\label{Decompose of main term under unknown ME}
        &S_{n,1}(\xi,\theta_0) = \frac{1}{n}\sum_{i=1}^n \int \left(Y_i-g(x;\theta_0)\right)h(x;\theta_0)\mathcal{K}_{b,1}\left(\frac{x-W_i}{b}\right)\ee^{\ii x\xi}\,dx,\notag\\
        &S_{n,2}(\xi,\theta_0) = \frac{1}{n}\sum_{i=1}^n \int \left(Y_i-g(x;\theta_0)\right)h(x;\theta_0)\mathcal{K}_{b,2}\left(\frac{x-W_i}{b}\right)\ee^{\ii x\xi}\,dx.
    \end{align}
    
    For the first term, we define $p(D_i,D_j;\xi)$ as
    \begin{align*}
        p(D_i,D_j;\xi) = &\frac{1}{2\pi b}\iint \left(Y_i-g(x;\theta_0)\right)h(x;\theta_0)\ee^{-\ii \frac{x-W_i}{b}t}\frac{K^{ft}(t)}{f^{ft}_\epsilon(\frac{t}{b})}\Pi_{\epsilon,j}(\frac{t}{b})\ee^{\ii x\xi}\,dx\,dt\\
        =&\frac{1}{2\pi}\iint \left(Y_i-g(x;\theta_0)\right)h(x;\theta_0)\ee^{-\ii (x-W_i)t}\frac{K^{ft}(bt)}{f^{ft}_\epsilon(t)}\Pi_{\epsilon,j}(t)\ee^{\ii x\xi}\,dx\,dt,
    \end{align*}
    where $D_i=(Y_i,W_i,W_i^r)$ for $i=1,2,\cdots,n$ and $\Pi_{\epsilon,j}$ is denoted as
    \begin{align}\label{part of psi}
        &\Pi_{\epsilon,j}\left(t\right) = \frac{\zeta_j(t)-\left[f_\epsilon^{ft}(t)\right]^2}{2\left[f_\epsilon^{ft}(t)\right]^2}.
    \end{align}
    Then, we can write $S_{n,1}(\xi,\theta_0) = S_{n,11}(\xi,\theta_0)+nS_{n,12}(\xi,\theta_0)/(n-1)$,
    \begin{align*}
        &S_{n,11}(\xi,\theta_0) = \frac{1}{n^2}\sum_{i=1}^np(D_i,D_i;\xi), \quad S_{n,12}(\xi,\theta_0) = \frac{1}{n(n-1)}\sum_{i\neq j}^n p(D_i,D_j;\xi).
    \end{align*}
    By similar arguments to the proof of Lemma \ref{lemma.power of tsf kernel}, we obtain
    \begin{align*}
        \mathbb{E}\left[\frac{1}{2\pi}\int\left\vert \int x^{l}\ee^{-\ii xt}\frac{K^{ft}(bt)}{f^{ft}_\epsilon(t)}\Pi_{\epsilon,j}(t)\,dt\right\vert dx\right]=o_p\left(n^{\frac{1}{2}}\right)\quad \text{for }l\leq p+1
    \end{align*}
    when $b^{-3\alpha-2}=o(n^{1/2})$ which holds under Assumption \ref{ass.O'}. For function $h(x;\theta_0)$ that satisfies the requirements of Assumption \ref{ass.O} and \ref{ass.O'}, we obtain
    \begin{align*}
        &\mathbb{E}\left[\sup_{\xi\in\Pi}\left\vert\frac{1}{2\pi}\iint h(x+W_i;\theta_0)\ee^{-\ii xt}\frac{K^{ft}(bt)}{f^{ft}_\epsilon(t)}\Pi_{\epsilon,j}(t)e^{\ii x\xi}\,dx\,dt\right\vert\right]\\
        \leq&\mathbb{E}\sup_{\xi\in\Pi}\left\{
        \begin{aligned}
            &\frac{1}{2\pi}\sum_{l=0}^{p-1}\frac{\left\vert h^{(l)}(W_i;\theta_0)\right\vert}{l!}\left\vert\iint x^l\ee^{-\ii xt}\frac{K^{ft}(bt)}{f^{ft}_\epsilon(t)}\Pi_{\epsilon,j}(t)e^{\ii x\xi}\,dx\,dt\right\vert\\
            &+\frac{1}{2\pi p!}\left\vert\iint h^{(p)}(\bar{W};\theta_0)x^l\ee^{-\ii xt}\frac{K^{ft}(bt)}{f^{ft}_\epsilon(t)}\Pi_{\epsilon,j}(t)e^{\ii x\xi}\,dx\,dt\right\vert\\
        \end{aligned}
        \right\}\\
        \leq&\mathbb{E}\left\{
        \begin{aligned}
            &\frac{1}{2\pi}\sum_{l=0}^p\frac{\left\vert h^{(l)}(W_i;\theta_0)\right\vert}{l!}\int\left\vert \int x^l\ee^{-\ii xt}\frac{K^{ft}(bt)}{f^{ft}_\epsilon(t)}\Pi_{\epsilon,j}(t)\,dt\right\vert\,dx \\
            &+ \frac{1}{2\pi p!} \left\vert m_h(W_i)\right\vert\int\left\vert \int x^{p+1}\ee^{-\ii xt}\frac{K^{ft}(bt)}{f^{ft}_\epsilon(t)}\Pi_{\epsilon,j}(t)\,dt\right\vert\,dx
        \end{aligned}
        \right\}=o_p\left(n^{\frac{1}{2}}\right),
    \end{align*}
    where $\bar{W}$ lies between $W_i$ and $W_i+x$. By similar arguments, the function $h$ can be replaced by $gh$. Consequently, $\mathbb{E}[\sup_{\xi\in\Pi}\vert p(D_i,D_j;\xi)\vert]=o_p(n^{1/2})$ follows by
    \begin{align*}
        &\mathbb{E}\left[\sup_{\xi\in\Pi}\left\vert p(D_i,D_j;\xi)\right\vert\right]\\
        &\leq \left\{\mathbb{E}\left[\sup_{\xi\in\Pi}\left\vert \frac{1}{2\pi}\iint h(x;\theta_0)\ee^{-\ii (x-W_i)t}\frac{K^{ft}(bt)}{f^{ft}_\epsilon(t)}\Pi_{\epsilon,j}(t)e^{\ii x\xi}\,dx\,dt\right\vert\right]^2\right\}^{\frac{1}{2}}\sqrt{\mathbb{E}Y_i^2}\\
        & + \mathbb{E}\left[\sup_{\xi\in\Pi}\left\vert \frac{1}{2\pi}\iint g(x;\theta)h(x;\theta_0)\ee^{-\ii (x-W_i)t}\frac{K^{ft}(bt)}{f^{ft}_\epsilon(t)}\Pi_{\epsilon,j}(t)e^{\ii x\xi}\,dx\,dt\right\vert \right].
    \end{align*}
    Thus, $\sup_{\xi\in\Pi}\vert S_{n,11}(\xi,\theta_0)\vert=o_p(n^{-1/2})$ holds. By similar arguments, $\mathbb{E}[\sup_{\xi\in\Pi} p^2(D_i,D_i;\xi)]=o_p(n)$ follows by $b^{-6\alpha-4}=o(n)$ which holds under Assumption \ref{ass.O'}.  Notice that $S_{n,12}(\cdot,\theta_0)$ is a second-order U-process with symmetric kernel $2q(D_i,D_j,\xi) = p(D_i,D_j,\xi)+ p(D_j,D_i,\xi)$ and H\'{a}jek projection
    \begin{align}\label{lemmaproof.convergence of main term unknown hajek proj}
        \hat{S}_{n,12}(\xi,\theta_0) = \mathbb{E}\left[p_1(D_i,\xi)\right]+\frac{2}{n}\sum_{i=1}^n\left\{p_1(D_i,\xi)-\mathbb{E}\left[p_1(D_i,\xi)\right]\right\}.
    \end{align}
    We note that
    \begin{align}\label{lemmaproof.convergence of main term unknown first term}
        &2p_1(D_i,\xi) = \mathbb{E}\left[2q(D_i,D_j,\xi)\vert D_i\right]\notag\\
        =&\frac{1}{2\pi}\iint \mathbb{E}\left[\left(Y-g(X;\theta_0)\right)\ee^{\ii Wt}\right]h(x;\theta_0)\frac{K^{ft}(bt)}{f^{ft}_\epsilon(t)}\Pi_{\epsilon,i}(t)\ee^{\ii x(\xi-t)}\,dx\,dt\notag\\
        &+\frac{1}{2\pi}\iint \mathbb{E}\left[\left(g(X;\theta_0)-g(x;\theta_0)\right)\ee^{\ii Xt}\right]h(x;\theta_0)K^{ft}(bt)\Pi_{\epsilon,i}(t)\ee^{\ii x(\xi-t)}\,dx\,dt,\notag\\
        =&\frac{1}{2\pi}\int\mathbb{E}\left[\left(Y-g(X;\theta_0)\right)\ee^{\ii Xt}\right]h^{ft}(\xi-t)K^{ft}(bt)\Pi_{\epsilon,i}(t)\,dt\notag\\
        &+\frac{1}{2\pi}\int \left[(gf_X)^{ft}(t)h^{ft}(\xi-t) - f_X^{ft}(t)(gh)^{ft}(\xi-t)\right]K^{ft}(bt)\Pi_{\epsilon,i}(t)\,dt,
    \end{align}
    where $(gf_X)^{ft}(t)$ represents $\int g(x;\theta_0)f_X(x)\ee^{\ii xt}\,dx$ and $f_X(x)$ is density of $X$. Consequently,
    \begin{align}\label{lemmaproof.convergence of main term unknown hajek proj and U}
        \sup_{\xi\in\Pi}\left\vert S_{n,12}(\xi,\theta_0) - \hat{S}_{n,12}(\xi,\theta_0)\right\vert = o_p\left(n^{-\frac{1}{2}}\right).
    \end{align}
    Furthermore, Assumption \ref{ass.O} provides that $K^{ft}(bt)$ goes to $1$, together with Assumption \ref{ass.O'}, we obtain
    \begin{align}\label{lemmaproof.convergence of main term unknown hajek and limit}
        \sup_{\xi\in\Pi}\left\vert p_1(D_i,\xi)- \frac{1}{4\pi}\int \left[
        \begin{aligned}
            &(gf_X)^{ft}(t)h^{ft}(\xi-t)- f_X^{ft}(t)(gh)^{ft}(\xi-t)\\
            &+\mathbb{E}\left[\left(Y-g(X;\theta_0)\right)\ee^{\ii Xt}\right]h^{ft}(\xi-t)
        \end{aligned}
         \right]\Pi_{\epsilon,i}(t)\,dt\right\vert=o_p\left(1\right).
    \end{align}
    by using the dominant convergence theorem. Combining \eqref{lemmaproof.convergence of main term unknown hajek and limit} and Assumption \ref{ass.O'}, we obtain
    \begin{align}\label{lemmaproof.convergence of main term unknown hajek convergence}
        \sup_{\xi\in\Pi}\left\vert\frac{1}{n}\sum_{i=1}^n\left\{p_1(D_i,\xi)-\mathbb{E}\left[p_1(D_i,\xi)\right]\right\}\right\vert=O_p\left(n^{-\frac{1}{2}}\right).
    \end{align}
    Hence, we claim that $\sup_{\xi\in\Pi}\vert S_{n,12}(\xi,\theta_0)\vert=O_p(n^{-1/2})$ by combining \eqref{lemmaproof.convergence of main term unknown hajek proj}, \eqref{lemmaproof.convergence of main term unknown hajek proj and U}, \eqref{lemmaproof.convergence of main term unknown hajek convergence} and $\mathbb{E}[p_1(D,\xi)]=0$. Along with $\sup_{\xi\in\Pi}\vert S_{n,11}(\xi,\theta_0)\vert=o_p(n^{-1/2})$, we obtain
    \begin{align}\label{lemmaproof.convergence of main term unknown Sn1}
        &\sup_{\xi\in\Pi}\left\vert S_{n,1}(\xi,\theta_0)-\frac{1}{2\pi n}\sum_{i=1}^n\int \left[
        \begin{aligned}
            &(gf_X)^{ft}(t)h^{ft}(\xi-t)\\
            &- f_X^{ft}(t)(gh)^{ft}(\xi-t)\\
            &+\mathbb{E}\left[\left(Y-g(X;\theta_0)\right)\ee^{\ii Xt}\right]h^{ft}(\xi-t)
        \end{aligned}
         \right]\Pi_{\epsilon,i}(t)\,dt\right\vert=o_p\left(n^{-\frac{1}{2}}\right).
    \end{align}

    For the second term, we define $H(Y,x;\theta_0) = (Y-g(x;\theta_0))h(x;\theta_0)$ and rewrite $S_{n,2}(\xi,\theta_0)$ as
    \begin{align}\label{lemmaproof.convergence of main term unknown Rn123}
        S_{n,2}(\xi,\theta_0) =& \frac{1}{n}\sum_{i=1}^n \ee^{\ii W_i\xi} \int H(Y_i, bx+W_i;\theta_0)\mathcal{K}_{\epsilon,2}(x)\ee^{\ii bx\xi}\,dx\notag\\
        =&R_{n,1}(\xi,\theta_0) + R_{n,2}(\xi,\theta_0) + R_{n,3}(\xi,\theta_0),
    \end{align}
    by Taylor expansion, where
    \begin{align*}
        &R_{n,1}(\xi,\theta_0) = \sum_{l=0}^{p-1}\frac{b^l}{l!}\left\{
        \begin{aligned}
             &\frac{1}{n}\sum_{i=1}^n\ee^{\ii W_i\xi} H^{(l)}(Y_i,W_i;\theta_0)\\
             &-\mathbb{E}\left[\ee^{\ii W\xi} H^{(l)}(Y,W;\theta_0)\right]
        \end{aligned}
        \right\}\int x^l\mathcal{K}_{\epsilon,2}(x)\ee^{\ii bx\xi}\,dx,\\
        &R_{n,2}(\xi,\theta_0) = \frac{b^p}{p!}\frac{1}{n}\sum_{i=1}^n\int\left\{
        \begin{aligned}
            &\ee^{\ii W_i\xi}H^{(p)}(Y_i,\bar{W}_i;\theta_0)\\
            -&\mathbb{E}\left[H^{(p)}(Y,\bar{W};\theta_0)\ee^{\ii W\xi}\right]
        \end{aligned}
        \right\}x^p\mathcal{K}_{\epsilon,2}(x)\ee^{\ii bx\xi}\,dx,\\
        &R_{n,3}(\xi,\theta_0) = \int \mathbb{E}\left[H(Y,x;\theta_0)\mathcal{K}_{b,2}\left(\frac{x-W}{b}\right)\right] \ee^{\ii x\xi}\,dx.
    \end{align*}
    Notice that Assumption \ref{ass.O} implies $\mathbb{E}[\int_{\Pi}\vert \ee^{\ii W\xi}H^{(l)}(Y,W;\theta_0)\vert^2\,d\xi]<\infty$,
    \begin{align}\label{lemmaproof.convergence of main term unknown Rn1 main term}
        &\sup_{\xi\in\Pi}\left\vert\frac{1}{n}\sum_{i=1}^n\ee^{\ii W_i\xi} H^{(l)}(Y_i,W_i;\theta_0) - \mathbb{E}\left[\ee^{\ii W\xi} H^{(l)}(Y,W;\theta_0)\right]\right\vert = O_p\left(n^{-\frac{1}{2}}\right)\quad \text{for }l<p
    \end{align}
    follows by Theorem 3.9 of \cite{chen1998central}. Together with $\int \vert \mathcal{K}_{\epsilon,2}(x)(bx)^l\vert\,dx=o_p(1)$ mentioned in Lemma \ref{lemma.power of decon}, we claim that
    \begin{align}\label{lemmaproof.convergence of main term unknown Rn1}
        \sup_{\xi\in\Pi}\left\vert R_{n,1}(\xi,\theta_0)\right\vert = o_p\left(n^{-\frac{1}{2}}\right).
    \end{align}
    For $R_{n,2}(\xi,\theta_0)$ in decomposition, using the Lipschitz continuity mentioned in Assumption \ref{ass.O},
    \begin{align*}
        \sup_{\xi\in\Pi}\vert R_{n,1}(\xi,\theta_0) \vert \leq&  \frac{1}{p!}\left[\frac{1}{n}\sum_{i=1}^n \left\vert H^{(p)}(Y_i,W_i;\theta_0)\right\vert+\mathbb{E}\left\vert H^{(p)}(Y,W;\theta_0)\right\vert\right]  \int \left\vert (bx)^p\mathcal{K}_{\epsilon,2}(x)\right\vert\,dx \\
        &+ \frac{1}{p!}\left[\frac{1}{n}\sum_{i=1}^n \left\vert L_{h^{(p)}}(Y_i,W_i) \right\vert+\mathbb{E}\left\vert L_{h^{(p)}}(Y,W) \right\vert\right] \int \left\vert (bx)^{p+1}\mathcal{K}_{\epsilon,2}(x)\right\vert\,dx,
    \end{align*}
    where $m_{H.p}(Y,W) = \vert Ym_{h}(W)\vert+\vert m_{gh}(W)\vert$. Notice that Assumption \ref{ass.O'} implies $\mathbb{E}\vert H^{(p)}(Y,W;\theta_0)\vert<\infty$ and $\mathbb{E}\vert L_{h^{(p)}}(Y_i,W_i) \vert<\infty$, along with $\int \vert (bx)^{p+1}\mathcal{K}_{\epsilon,2}(x)\vert\,dx = o_p(n^{-1/2})$ mentioned in Lemma \ref{lemma.power of decon}, we obtain 
    \begin{align}\label{lemmaproof.convergence of main term unknown Rn2}
        \sup_{\xi\in\Pi}\left\vert R_{n,2}(\xi,\theta_0)\right\vert = o_p\left(n^{-\frac{1}{2}}\right).
    \end{align}
    Subsequently, rewriting $R_{n,3}(\xi,\theta_0)$ as
    \begin{align*}
        &R_{n,3}(\xi,\theta_0) =  \int \mathbb{E}\left[H(Y,X+bx;\theta_0)\ee^{\ii (X+bx)\xi}\right] \mathbb{E}\left[K_{\epsilon,2}(x)\right]\,dx\\
        = &\int\mathbb{E}\left[\left(Y-g(X;\theta_0)\right)h(X+bx;\theta_0)\ee^{\ii (X+bx)\xi}\right] \mathbb{E}\left[K_{\epsilon,2}(x)\right]\,dx\\
        &+\int \mathbb{E}\left[\left(g(X;\theta_0)-g(X+bx;\theta_0)\right)h(X+bx;\theta_0)\ee^{\ii (X+bx)\xi}\right] \mathbb{E}\left[K_{\epsilon,2}(x)\right]\,dx\\
        = &\int \mathbb{E}\left[\left(Y-g(X;\theta_0)\right)K_{b,2}\left(\frac{x-X}{b}\right)\right]h(x;\theta_0)\ee^{\ii x\xi}\,dx\\
        &+\iint \left[\left(g(y;\theta_0)-g(x;\theta_0)\right)h(x;\theta_0)\right]\ee^{\ii x\xi} \mathbb{E}\left[K_{b,2}\left(\frac{x-y}{b}\right)\right]f_X(y)\,dx\,dy.
    \end{align*}
    Let
    \begin{align*}
        &R_{n,31}(\xi,\theta_0) = \frac{1}{2\pi b}\int h(x;\theta_0)\ee^{\ii x\xi} \left[\int g(y;\theta_0)K_{\epsilon,2}\left(\frac{x-y}{b}\right)f_X(y)\,dy\right] \,dx,\\
        &R_{n,32}(\xi,\theta_0) = \frac{1}{2\pi b}\int h(x;\theta_0)g(x;\theta_0)\ee^{\ii x\xi} \left[\int K_{\epsilon,2}\left(\frac{x-y}{b}\right)f_X(y)\,dy\right] \,dx.
    \end{align*}
    By similar arguments to the proof of Lemma \ref{lemma.convergence of main term known}, 
    \begin{align*}
        &\int g(y;\theta_0)K_{\epsilon,2}\left(\frac{x-y}{b}\right)f_X(y)\,dy\\
        = & \sum_{l=0}^{p-1}\frac{b^l}{l!}(gf_X)^{(l)}(x;\theta_0) \int y^lK_{\epsilon,2}(y)\,dy + \frac{b^p}{p!}\int (gf_X)^{(p)}(\tilde{x};\theta_0) y^pK_{\epsilon,2}(y)\,dy\\
        = &\frac{b^p}{p!}\int (gf_X)^{(p)}(\tilde{x};\theta_0) y^pK_{\epsilon,2}(y)\,dy
    \end{align*}
    holds using Taylor expansion, where $\tilde{x}$ is between $x-by$ and $x$. We note that the second equation follows by $\int y^lK_{\epsilon,2}(y)\,dy = 0$ for $l<p$, which is mentioned in Lemma \ref{lemma.power of tsf kernel}. Under the Lipschitz continuity mentioned in Assumption \ref{ass.O}, 
    \begin{align*}
        &\left\vert \int g(y;\theta_0)K_{\epsilon,2}\left(\frac{x-y}{b}\right)f_X(y)\,dy \right\vert \\
        \leq & \frac{b^p}{p!}\left\vert (gf_X)^{(p)}(x;\theta_0) \right\vert \int \left\vert y^pK_{\epsilon,2}(y) \right\vert \,dy + \frac{b^{p+1}}{p!}\left\vert L_{[gf_X]^{(p)}}(x) \right\vert \int \left\vert y^{p+1}K_{\epsilon,2}(y) \right\vert \,dy.
    \end{align*}
    Subsequently, we claim $\sup_{\xi\in\Pi}\vert R_{n,31}(\xi,\theta_0)\vert = o_p(n^{-1/2})$ which follows by conclusions in Lemma \ref{lemma.power of tsf kernel}.  By a similar arguments, $\sup_{\xi\in\Pi}\vert R_{n,32}(\xi,\theta_0)\vert = o_p(n^{-1/2})$. Consequently, 
    \begin{align}\label{lemmaproof.convergence of main term unknown Rn3}
        \sup_{\xi\in\Pi}\left\vert R_{n,3}(\xi,\theta_0)-\int \mathbb{E}\left[\left(Y-g(X;\theta_0)\right)K_{b,2}\left(\frac{x-X}{b}\right)\right]h(x;\theta_0)\ee^{\ii x\xi}\,dx\right\vert = o_p\left(n^{-\frac{1}{2}}\right).
    \end{align}
    Thus, \eqref{lemma.convergence of main term unknown eq1} holds for the ordinary smooth case by combining \eqref{lemma.convergence of main term known eq1}, \eqref{Decompose of main term under unknown ME}, \eqref{lemmaproof.convergence of main term unknown Sn1}, \eqref{lemmaproof.convergence of main term unknown Rn123}, \eqref{lemmaproof.convergence of main term unknown Rn1}, \eqref{lemmaproof.convergence of main term unknown Rn2} and \eqref{lemmaproof.convergence of main term unknown Rn3}.

    For the supersmooth case, \eqref{lemma.convergence of main term known eq1} and \eqref{Decompose of main term under unknown ME} still hold. For the first term of decomposition, under Assumption \ref{ass.S'}, following the proof of Lemma \ref{lemma.power of tsf kernel},
    \begin{align*}
        &\mathbb{E}\left[\sup_{\xi\in\Pi}\left\vert\frac{1}{2\pi}\iint h(x+W_i;\theta_0)\ee^{-\ii xt}\frac{K^{ft}(bt)}{f^{ft}_\epsilon(t)}\Pi_{\epsilon,j}(t)e^{\ii x\xi}\,dx\,dt\right\vert\right]\\
        \leq&
        \frac{1}{2\pi}\sum_{l=0}^{\infty}\frac{1}{l!}\mathbb{E}\left\vert h^{(l)}(W_i;\theta_0)\right\vert\mathbb{E}\sup_{\xi\in\Pi}\left\vert\iint x^l\ee^{-\ii xt}\frac{K^{ft}(bt)}{f^{ft}_\epsilon(t)}\Pi_{\epsilon,j}(t)e^{\ii x\xi}\,dx\,dt\right\vert\\
        =&O_p\left(\ee^{-3\mu(1+b^{-1})^2}\right)=o_p\left(n^{\frac{1}{2}}\right),
    \end{align*}
     where the last equation follows by Assumption \ref{ass.S'}. Thus, by similar arguments to ordinary smooth case, $\sup_{\xi\in\Pi}\vert S_{n,11}(\xi,\theta_0)\vert=o_p(n^{-1/2})$ and $\mathbb{E}[\sup_{\xi\in\Pi} p^2(D_i,D_i;\xi)]=o_p(n)$ hold. Subsequently, \eqref{lemmaproof.convergence of main term unknown Sn1} follows by \eqref{lemmaproof.convergence of main term unknown hajek proj}--\eqref{lemmaproof.convergence of main term unknown hajek convergence}. For the second term of decomposition, we decompose
     \begin{align}\label{lemmaproof.convergence of main term unknown Rn12 super}
        S_{n,2}(\xi,\theta_0) =&R^{ss}_{n,1}(\xi,\theta_0) + R^{ss}_{n,2}(\xi,\theta_0),
    \end{align}
    where
    \begin{align*}
        &R^{ss}_{n,1}(\xi,\theta_0) = \sum_{l=0}^{\infty}\frac{b^l}{l!}\left\{
        \begin{aligned}
             &\frac{1}{n}\sum_{i=1}^n\ee^{\ii W_i\xi} H^{(l)}(Y_i,W_i;\theta_0)\\
             &-\mathbb{E}\left[\ee^{\ii W\xi} H^{(l)}(Y,W;\theta_0)\right]
        \end{aligned}
        \right\}\int x^l\mathcal{K}_{\epsilon,2}(x)\ee^{\ii bx\xi}\,dx,\\
        &R^{ss}_{n,2}(\xi,\theta_0) = \int \mathbb{E}\left[H(Y,x;\theta_0)\mathcal{K}_{b,2}\left(\frac{x-W}{b}\right)\right] \ee^{\ii x\xi}\,dx.
    \end{align*}
    Under Assumption \ref{ass.S'}, \eqref{lemmaproof.convergence of main term unknown Rn1 main term} still holds and therefore,
    \begin{align}\label{lemmaproof.convergence of main term unknown Rn1 super}
        \sup_{\xi\in\Pi}\left\vert R^{ss}_{n,1}(\xi,\theta_0)\right\vert = o_p\left(n^{-\frac{1}{2}}\right)
    \end{align}
    follows by conclusion \eqref{Power of super kernel} mentioned in Lemma \ref{lemma.power of decon}. We notice that $R^{ss}_{n,2}(\xi,\theta_0) = R_{n,31}(\xi,\theta_0)+R_{n,32}(\xi,\theta_0)$ still holds, together with
    \begin{align*}
        &\int g(y;\theta_0)K_{\epsilon,2}\left(\frac{x-y}{b}\right)f_X(y)\,dy= \sum_{l=0}^{\infty}\frac{b^l}{l!}(gf_X)^{(l)}(x;\theta_0) \int y^lK_{\epsilon,2}(y)\,dy=0
    \end{align*}
    and 
    \begin{align*}
        &\int K_{\epsilon,2}\left(\frac{x-y}{b}\right)f_X(y)\,dy= \sum_{l=0}^{\infty}\frac{b^l}{l!}(f_X)^{(l)}(x;\theta_0) \int y^lK_{\epsilon,2}(y)\,dy=0,
    \end{align*}
    we claim 
    \begin{align}\label{lemmaproof.convergence of main term unknown Sn2}
         \sup_{\xi\in\Pi}\left\vert S_{n,2}(\xi,\theta_0)-\int \mathbb{E}\left[\left(Y-g(X;\theta_0)\right)K_{b,2}\left(\frac{x-X}{b}\right)\right]h(x;\theta_0)\ee^{\ii x\xi}\,dx\right\vert = o_p\left(n^{-\frac{1}{2}}\right).
    \end{align}
    Consequently, \eqref{lemma.convergence of main term unknown eq1} follows by \eqref{lemmaproof.convergence of main term unknown Sn1} and \eqref{lemmaproof.convergence of main term unknown Sn2}. For the bootstrap version, \eqref{lemma.convergence of main term unknown eq2} holds by similar arguments except for
    \begin{align*}
        \mathbb{E}\left[V\left(Y-g(X;\theta_0)\right)K_{b,2}\left(\frac{x-X}{b}\right)\right] = 0 \text{ and }\mathbb{E}\left[V\left(Y-g(X;\theta_0)\right)K_{b}\left(\frac{x-X}{b}\right)\right]=0
    \end{align*}
    which follows from the mean zero and independence properties of $V$. Given the similarity between the bootstrap version \eqref{lemma.convergence of main term unknown eq2} and equation \eqref{lemma.convergence of main term unknown eq1}, the bootstrap proofs are omitted hereafter unless stated otherwise.
\end{proof}

\begin{lemma}\label{lemma.ULLN of main term}
    Suppose that Assumption \ref{ass.D} holds, together with either Assumption \ref{ass.O} for the ordinary smooth case or Assumption \ref{ass.S} for the supersmooth case,
    \begin{align}\label{lemma.ULLN of main term eq1}
        &\sup_{\xi\in\Pi}\left\vert\frac{1}{n}\sum\limits_{i=1}^n \int h(x;\theta)\mathcal{K}_b\left(\frac{x-W_i}{b}\right)\ee^{\ii x\xi}\,dx-\mathbb{E}\left[h(X;\theta)\ee^{\ii X\xi}\right]\right\vert=O_p\left(n^{-\frac{1}{2}}\right)
    \end{align}
    holds when $\theta$ values in a neighborhood of $\theta_0$. For the bootstrap version,
    \begin{align}\label{lemma.ULLN of main term eq2}
        &\sup_{\xi\in\Pi}\left\vert\frac{1}{n}\sum\limits_{i=1}^nV_i \int h(x;\theta)\mathcal{K}_b\left(\frac{x-W_i}{b}\right)\ee^{\ii x\xi}\,dx\right\vert=O_p\left(n^{-\frac{1}{2}}\right).
    \end{align}
\end{lemma}

\begin{proof}[Proof of Lemma \ref{lemma.ULLN of main term}]
    By similar arguments to the proof of Lemma \ref{lemma.convergence of main term known}, under Assumption \ref{ass.O} or \ref{ass.S},
    \begin{align}\label{lemmaproof.ULLN of main term main-E}
        &\sup_{\xi\in\Pi}\left\vert
        \begin{aligned}
             &\frac{1}{n}\sum\limits_{i=1}^n \int h(x;\theta)\mathcal{K}_b\left(\frac{x-W_i}{b}\right)\ee^{\ii x\xi}\,dx\\
             &-\mathbb{E}\left[ \int h(x;\theta)\mathcal{K}_b\left(\frac{x-W}{b}\right)\ee^{\ii x\xi}\,dx\right]
        \end{aligned}
       \right\vert=O_p\left(n^{-\frac{1}{2}}\right).
    \end{align}
    For ordinary smooth case, under Assumption \ref{ass.O},
    \begin{align}\label{lemmaproof.ULLN of main term exp ord}
        \sup_{\xi\in\Pi}\left\vert\mathbb{E}\left[ \int h(x;\theta)\mathcal{K}_b\left(\frac{x-W}{b}\right)\ee^{\ii x\xi}\,dx\right]-\mathbb{E}\left[h(X;\theta)\ee^{\ii X\xi}\right]\right\vert = O\left(n^{-\frac{1}{2}}\right)
    \end{align}
    follows by \eqref{lemmaproof.convergence of main term known J2n}. And for supersmooth case, under Assumption \ref{ass.S},
    \begin{align}\label{lemmaproof.ULLN of main term exp sup}
        \mathbb{E}\left[ \int h(x;\theta)\mathcal{K}_b\left(\frac{x-W}{b}\right)\ee^{\ii x\xi}\,dx\right]=\mathbb{E}\left[h(X;\theta)\ee^{\ii X\xi}\right]. 
    \end{align}
    Thus, \eqref{lemma.ULLN of main term eq1} holds.
\end{proof}

\begin{lemma}\label{lemma.ULLN of main term unknown}
    Suppose that Assumption \ref{ass.D} and \ref{ass.D'} hold, along with either Assumption \ref{ass.O} and \ref{ass.O'} for the ordinary smooth case or Assumption \ref{ass.S} and \ref{ass.S'} for the supersmooth case,
    \begin{align}\label{lemma.ULLN of main term unknown eq1}
        &\sup_{\xi\in\Pi}\left\vert\frac{1}{n}\sum\limits_{i=1}^n \int h(x;\theta)\hat{\mathcal{K}}_b\left(\frac{x-W_i}{b}\right)\ee^{\ii x\xi}\,dx-\mathbb{E}\left[h(X;\theta)\ee^{\ii X\xi}\right]\right\vert=O_p\left(n^{-\frac{1}{2}}\right)
    \end{align}
    holds when $\theta$ values in a neighborhood of $\theta_0$. For the bootstrap version,
    \begin{align}\label{lemma.ULLN of main term unknown eq2}
        &\sup_{\xi\in\Pi}\left\vert\frac{1}{n}\sum\limits_{i=1}^nV_i \int h(x;\theta)\hat{\mathcal{K}}_b\left(\frac{x-W_i}{b}\right)\ee^{\ii x\xi}\,dx\right\vert=O_p\left(n^{-\frac{1}{2}}\right).
    \end{align}
\end{lemma}

\begin{proof}[Proof of Lemma \ref{lemma.ULLN of main term unknown}]
    By similar arguments to the proof of Lemma \ref{lemma.convergence of main term unknown}, under Assumption \ref{ass.O} and \ref{ass.O'} for ordinary smooth case or \ref{ass.S} and \ref{ass.S'} for supersmooth case,
    \begin{align}\label{lemmaproof.ULLN of main term unknown main-E}
        &\sup_{\xi\in\Pi}\left\vert
        \begin{aligned}
             &\frac{1}{n}\sum\limits_{i=1}^n \int h(x;\theta)\hat{\mathcal{K}}_b\left(\frac{x-W_i}{b}\right)\ee^{\ii x\xi}\,dx\\
             &-\int \mathbb{E}\left[K_{b}\left(\frac{x-X}{b}\right)\right]h(x;\theta)\ee^{\ii x\xi}\,dx
        \end{aligned}
       \right\vert=O_p\left(n^{-\frac{1}{2}}\right).
    \end{align}
    Notice that we have proved \eqref{lemmaproof.ULLN of main term exp ord} and \eqref{lemmaproof.ULLN of main term exp sup}, \eqref{lemma.ULLN of main term unknown eq1} follows.
\end{proof}

\begin{lemma}\label{lemma.projG}
    Suppose that Assumption \ref{ass.D} holds, together with either Assumption \ref{ass.O} for the ordinary smooth case or Assumption \ref{ass.S} for the supersmooth case,
    \begin{align}\label{lemma.projG eq1}
        &\sup_{\xi\in\Pi}\left\vert G_n(\xi,\hat{\theta}_n)-G(\xi,\theta_0)\right\vert=O_p\left(\left\vert \hat{\theta}_n-\theta_0\right\vert\right)+O_p\left(n^{-\frac{1}{2}}\right).
    \end{align}
    For the bootstrap version,
    \begin{align}\label{lemma.projG eq2}
        &\sup_{\xi\in\Pi}\left\vert G_n^\ast(\xi,\theta_0)\right\vert = O_p\left(n^{-\frac{1}{2}}\right).
    \end{align}
\end{lemma}

\begin{proof}[Proof of Lemma \ref{lemma.projG}]
    We start by decomposing
    \begin{align*}
        &G_n(\xi,\hat{\theta}_n) - G_n(\xi,\theta_0) = \frac{1}{n}\sum\limits_{i=1}^n \int \left[\frac{\partial g(x;\hat{\theta}_n)}{\partial\theta} - \frac{\partial g(x;\theta_0)}{\partial\theta}\right]\mathcal{K}_b\left(\frac{x-W_i}{b}\right)\ee^{\ii x\xi}\,dx\\
        & = \left[\frac{1}{n}\sum_{i=1}^n \int \frac{\partial^2 g(x;\tilde{\theta}_1)}{\partial \theta\partial\theta^\top}\mathcal{K}_b\left(\frac{x-W_i}{b}\right)\ee^{\ii x\xi}\,dx\right]\left(\hat{\theta}_n-\theta_0\right),
    \end{align*}
    where $\tilde{\theta}_1$ is a value between $\theta_0$ and $\hat{\theta}_n$. Consequently,
    \begin{align}\label{lemmaproof.projG Gtheta-Gtheta0}
        &\sup_{\xi\in\Pi}\left\vert G_n(\xi,\hat{\theta}_n) -G_n(\xi,\theta_0)\right\vert = O_p\left(\left\vert \hat{\theta}_n-\theta_0\right\vert\right)
    \end{align}
    and 
    \begin{align}\label{lemmaproof.projG Gtheta0-G0}
        \sup_{\xi\in\Pi}\left\vert G_n(\xi,\theta_0) - G(\xi,\theta_0)\right\vert = O_p\left(n^{-\frac{1}{2}}\right)
    \end{align}
    follow by designate $\partial g^2(x;\theta)/\partial\theta_j\partial\theta_k$ and $\partial g(x;\theta)/\partial\theta_j$ as $h(x;\theta)$ in Lemma \ref{lemma.ULLN of main term} and Lemma \ref{lemma.convergence of main term known} separately, where $j,k$ value from $1$ to $d$. Thus, \eqref{lemma.projG eq1} follows by combining \eqref{lemmaproof.projG Gtheta-Gtheta0} and \eqref{lemmaproof.projG Gtheta0-G0}.
\end{proof}

\begin{lemma}\label{lemma.projG unknown}
    Suppose that Assumption \ref{ass.D} and \ref{ass.D'} hold, along with either Assumption \ref{ass.O} and \ref{ass.O'} for the ordinary smooth case or Assumption \ref{ass.S} and \ref{ass.S'} for the supersmooth case,
    \begin{align}\label{lemma.projG unknown eq1}
        &\sup_{\xi\in\Pi}\left\vert \hat{G}_n(\xi,\hat{\theta}_n)-G(\xi,\theta_0)\right\vert=O_p\left(\left\vert \hat{\theta}_n-\theta_0\right\vert\right)+O_p\left(n^{-\frac{1}{2}}\right)
    \end{align}
    and 
    \begin{align}\label{lemma.projG unknown eq2}
        &\sup_{\xi\in\Pi}\left\vert \hat{G}_n^\ast(\xi,\theta_0)\right\vert=O_p\left(n^{-\frac{1}{2}}\right).
    \end{align}
\end{lemma}

\begin{proof}[Proof of Lemma \ref{lemma.projG unknown}]
    The step of proof is identical to that of Lemma \ref{lemma.projG} except for using Lemma \ref{lemma.convergence of main term unknown} and Lemma \ref{lemma.ULLN of main term unknown} instead of Lemma \ref{lemma.convergence of main term known} and Lemma \ref{lemma.ULLN of main term} and is therefore omitted. It is noted that $\int K_{\epsilon,2}(t)\,dt=\psi_2(0)=0$, implying that the estimation of the deconvolution kernel does not influence the bias of the proposed $\hat{G}_n(\xi,\theta_0)$. Since similar arguments apply to the proofs of subsequent lemmas, we omit the detailed steps unless otherwise noted.
\end{proof}

\begin{lemma}\label{lemma.projdelta known}
    Suppose that Assumption \ref{ass.D} holds, together with either Assumption \ref{ass.O} for the ordinary smooth case or Assumption \ref{ass.S} for the supersmooth case,
    \begin{align}\label{lemma.projdelta known eq1}
        &\left\vert \Delta^{-1}_n(\hat{\theta}_n) - \Delta^{-1}(\theta_0)\right\vert = O_p\left(\left\vert \hat{\theta}_n-\theta_0\right\vert\right)+O_p\left(n^{-\frac{1}{2}}\right).
    \end{align}
    For the bootstrap version,
    \begin{align}\label{lemma.projdelta known eq2}
        &\left\vert \Delta_n^\ast(\theta_0)\right\vert = O\left(n^{-\frac{1}{2}}\right).
    \end{align}
\end{lemma}

\begin{proof}[Proof of Lemma \ref{lemma.projdelta known}]
    Following along the same lines as Lemma \ref{lemma.projG}, we designate $(\partial g^2(x;\theta)/\partial\theta_j\partial\theta_k)(\partial g(x;\theta)/\partial\theta_l)$ and $\partial g^2(x;\theta)/\partial\theta_j\partial\theta_k$ as $h(x;\theta)$ in Lemma \ref{lemma.ULLN of main term unknown} and Lemma \ref{lemma.convergence of main term unknown} separately, where $j,k,l$ value from $1$ to $d$. Thus,
    \begin{align*}
        &\left\vert\Delta_n(\hat{\theta}_n) - \Delta(\theta_0)\right\vert = O_p\left(\left\vert \hat{\theta}_n-\theta_0\right\vert\right)+O_p\left(n^{-\frac{1}{2}}\right).
    \end{align*}
    We claim that \eqref{lemma.projdelta known eq1} holds by the boundedness of $\Delta(\theta_0)$ mentioned in Assumption \ref{ass.O} and
    \begin{align}\label{transfer of delta}
        &\Delta^{-1}_n(\hat{\theta}_n) - \Delta^{-1}(\theta_0) = -\Delta^{-1}(\theta_0)\left(\Delta_n(\hat{\theta}_n) - \Delta(\theta_0)\right)\Delta^{-1}_n(\hat{\theta}_n).
    \end{align} 
\end{proof}

\begin{lemma}\label{lemma.projdelta unknown}
    Suppose that Assumption \ref{ass.D} and \ref{ass.D'} hold, along with either Assumption \ref{ass.O} and \ref{ass.O'} for the ordinary smooth case or Assumption \ref{ass.S} and \ref{ass.S'} for the supersmooth case,
    \begin{align}\label{lemma.projdelta unknown eq1}
        &\left\vert \hat{\Delta}^{-1}_n(\hat{\theta}_n) - \Delta^{-1}(\theta_0)\right\vert = O_p\left(\left\vert \hat{\theta}_n-\theta_0\right\vert\right)+O_p\left(n^{-\frac{1}{2}}\right)
    \end{align}
    and 
    \begin{align}\label{lemma.projdelta unknown eq2}
        &\left\vert\hat{\Delta}_n^\ast(\theta_0)\right\vert = O\left(n^{-\frac{1}{2}}\right).
    \end{align}
\end{lemma}

\begin{proof}[Proof of Lemma \ref{lemma.projdelta unknown}]
    Relevant explanations can be found in the proof of Lemma \ref{lemma.projG unknown}.
\end{proof}

\begin{lemma}\label{lemma.main term}
    Suppose that Assumption \ref{ass.D} holds, together with either Assumption \ref{ass.O} for the ordinary smooth case or Assumption \ref{ass.S} for the supersmooth case,
    \begin{align}\label{lemma.main term eq1}
        &\sup_{\xi\in\Pi}\left\vert S_{n}(\xi,\hat{\theta}_n)- S_{n}(\xi,\theta_0)+\left(\hat{\theta}_n-\theta_0\right)^\top G_n(\xi,\theta_0)\right\vert=O_p\left(\left\vert \hat{\theta}_n-\theta_0\right\vert^2\right).
    \end{align}
    For the bootstrap version,
    \begin{align}\label{lemma.main term eq2}
        &\sup_{\xi\in\Pi}\left\vert S^\ast_{n}(\xi,\hat{\theta}_n)- S^\ast_{n}(\xi,\theta_0)+\left(\hat{\theta}_n-\theta_0\right)^\top G^\ast_n(\xi,\theta_0)\right\vert=O_p\left(\left\vert \hat{\theta}_n-\theta_0\right\vert^2\right).
    \end{align}
\end{lemma}

\begin{proof}[Proof of Lemma \ref{lemma.main term}]
    Similar to all lemmas we have proved before, we decompose $S_n(\xi,\hat{\theta}_n)$ as,
    \begin{align*}
        S_n(\xi,\hat{\theta}_n) - S_n(\xi,\theta_0) = &-\frac{1}{n}\sum\limits_{i=1}^n \int\left[g(x;\hat{\theta}_n) - g(x;\theta_0)\right]\mathcal{K}_b\left(\frac{x-W_i}{b}\right)\ee^{\ii x\xi}\,dx\\
        = &-\left(\hat{\theta}_n-\theta_0\right)^\top \left[\frac{1}{n}\sum\limits_{i=1}^n \int\frac{\partial g(x;\theta_0)}{\partial \theta}\mathcal{K}_b\left(\frac{x-W_i}{b}\right)\ee^{\ii x\xi}dx \right] \\
        & - \left(\hat{\theta}_n-\theta_0\right)^\top \left[\frac{1}{n}\sum_{i=1}^n \int \frac{\partial^2 g(x;\tilde{\theta}_2)}{\partial \theta\partial\theta^\top}\mathcal{K}_b\left(\frac{x-W_i}{b}\right)\ee^{\ii x\xi}\,dx\right] \left(\hat{\theta}_n-\theta_0\right),
    \end{align*}
    where $\tilde{\theta}_2$ lies between $\theta_0$ and $\hat{\theta}_n$. Thus, \eqref{lemma.main term eq1} holds by utilization of Lemma \ref{lemma.ULLN of main term unknown} via similar arguments to the proof of \eqref{lemma.projG eq1} and the definition of $G_n(\xi,\theta_0)$.
\end{proof}

\begin{lemma}\label{lemma.main term unknown}
    Suppose that Assumption \ref{ass.D} and \ref{ass.D'} hold, along with either Assumption \ref{ass.O} and \ref{ass.O'} for the ordinary smooth case or Assumption \ref{ass.S} and \ref{ass.S'} for the supersmooth case,
    \begin{align}\label{lemma.main term unknown eq1}
        &\sup_{\xi\in\Pi}\left\vert \hat{S}_{n}(\xi,\hat{\theta}_n)- \hat{S}_{n}(\xi,\theta_0)+\left(\hat{\theta}_n-\theta_0\right)^\top \hat{G}_n(\xi,\theta_0)\right\vert=O_p\left(\left\vert \hat{\theta}_n-\theta_0\right\vert^2\right)
    \end{align}
    and 
    \begin{align}\label{lemma.main term unknown eq2}
        &\sup_{\xi\in\Pi}\left\vert \hat{S}^\ast_{n}(\xi,\hat{\theta}_n)- \hat{S}^\ast_{n}(\xi,\theta_0)+\left(\hat{\theta}_n-\theta_0\right)^\top \hat{G}^\ast_n(\xi,\theta_0)\right\vert=O_p\left(\left\vert \hat{\theta}_n-\theta_0\right\vert^2\right).
    \end{align}
\end{lemma}

\begin{proof}[Proof of Lemma \ref{lemma.main term unknown}]
    See the discussion in the proof of Lemma \ref{lemma.projG unknown}.
\end{proof}

\begin{lemma}\label{lemma.proj main term known}
    Suppose that Assumption \ref{ass.D} holds, together with either Assumption \ref{ass.O} for the ordinary smooth case or Assumption \ref{ass.S} for the supersmooth case,
    \begin{align}\label{lemma.proj main term known eq1}
        &\sup_{\xi\in\Pi}\left\vert M_n(\hat\theta_n)- M_n(\theta_0)+\left[\Delta_n(\theta_0)-\kappa_n(\theta_0)\right]\left(\hat{\theta}_n-\theta_0\right)\right\vert=O_p\left(\left\vert\hat{\theta}_n-\theta_0\right\vert^2\right),
    \end{align}
    where
    \begin{align}\label{lemma.proj main term known kappa}
        &\kappa_n(\theta_0) = \frac{1}{n}\sum\limits_{i=1}^n \int (Y_i-g(x;\theta_0))\frac{\partial^2 g(x;\theta_0)}{\partial \theta\partial \theta^\top} \mathcal{K}_b\left(\frac{x-W_i}{b}\right) dx.
    \end{align}
    For the bootstrap version,
    \begin{align}\label{lemma.proj main term known eq2}
        &\sup_{\xi\in\Pi}\left\vert M_n^\ast(\hat\theta_n)- M_n^\ast(\theta_0)+\left[\Delta_n^\ast(\theta_0)-\kappa_n^\ast(\theta_0)\right]\left(\hat{\theta}_n-\theta_0\right)\right\vert=O_p\left(\left\vert\hat{\theta}_n-\theta_0\right\vert^2\right).
    \end{align}
\end{lemma}

\begin{proof}[Proof of Lemma \ref{lemma.proj main term known}]
    Decomposing $M_n(\hat\theta_n)$ as 
    \begin{align*}
        &M_{n}(\hat{\theta}_n) - M_{n}(\theta_0)  \\
        = &\frac{1}{n}\sum\limits_{i=1}^n\int \left[\left(Y_i-g(x;\hat{\theta}_n)\right)\frac{\partial g(x;\hat{\theta}_n)}{\partial \theta} - \left(Y_i-g(x;\theta_0)\right)\frac{\partial g(x;\theta_0)}{\partial \theta} \right] \mathcal{K}_b\left(\frac{x-W_i}{b}\right)dx\\
        = & -\left[ \frac{1}{n}\sum\limits_{i=1}^n\int \frac{\partial g(x;\theta_0)}{\partial \theta}\left(\frac{\partial g(x;\theta_0)}{\partial \theta}\right)^\top\mathcal{K}_b\left(\frac{x-W_i}{b}\right)dx \right]\left(\hat{\theta}_n-\theta_0\right)\\
        & + \left[ \frac{1}{n}\sum\limits_{i=1}^n \int \left(Y_i-g(x;\theta_0)\right)\frac{\partial^2 g(x;\theta_0)}{\partial \theta\partial \theta^\top} \mathcal{K}_b\left(\frac{x-W_i}{b}\right) dx  \right]\left(\hat{\theta}_n-\theta_0\right)\\
        & + \left(\hat{\theta}_n-\theta_0\right)^\top \left\{ \frac{1}{n}\sum\limits_{i=1}^n \int \frac{\partial^2 \left[\left(Y_i-g(x;\tilde{\theta}_3)\right)\frac{\partial g(x;\tilde{\theta}_3)}{\partial \theta}\right] }{\partial \theta\partial \theta^\top}\mathcal{K}_b\left(\frac{x-W_i}{b}\right)dx \right\}\left(\hat{\theta}_n-\theta_0\right).
    \end{align*}
    Thus, \eqref{lemma.proj main term known eq1} follows by Lemma \ref{lemma.ULLN of main term unknown}, where we designate $(\partial^3g(x;\theta)/\partial\theta_j\partial\theta_k\partial\theta_l) g(x;\theta)$, $(\partial g^2(x;\theta)/\partial\theta_j\partial\theta_k)(\partial g(x;\theta)/\partial\theta_l)$ and $\partial^3g(x;\theta)/\partial\theta_j\partial\theta_k\partial\theta_l$ as $h(x;\theta)$ in Lemma \ref{lemma.ULLN of main term unknown}, where $j,k,l$ value from $1$ to $d$.
\end{proof}

\begin{lemma}\label{lemma.proj main term unknown}
    Suppose that Assumption \ref{ass.D} and \ref{ass.D'} hold, along with either Assumption \ref{ass.O} and \ref{ass.O'} for the ordinary smooth case or Assumption \ref{ass.S} and \ref{ass.S'} for the supersmooth case,
    \begin{align}\label{lemma.proj main term param unknown eq1}
        &\sup_{\xi\in\Pi}\left\vert\hat M_n(\hat\theta_n)-\hat M_n(\theta_0)+\left[\hat{\Delta}_n(\theta_0)-\hat{\kappa}_n(\theta_0)\right]\left(\hat{\theta}_n-\theta_0\right)\right\vert=O_p\left(\left\vert\hat{\theta}_n-\theta_0\right\vert^2\right)
    \end{align}
    and 
    \begin{align}\label{lemma.proj main term param unknown eq2}
        &\sup_{\xi\in\Pi}\left\vert\hat M_n^\ast(\hat\theta_n)-\hat M_n^\ast(\theta_0)+\left[\hat{\Delta}_n^\ast(\theta_0)-\hat{\kappa}_n^\ast(\theta_0)\right]\left(\hat{\theta}_n-\theta_0\right)\right\vert=O_p\left(\left\vert\hat{\theta}_n-\theta_0\right\vert^2\right).
    \end{align}
\end{lemma}

\begin{proof}[Proof of Lemma \ref{lemma.proj main term unknown}]
    Relevant explanations can be found in the proof of Lemma \ref{lemma.projG unknown}.
\end{proof}

\begin{lemma}\label{lemma.power of tsf kernel}
Under Assumptions \ref{ass.D} and \ref{ass.D'}, with addition of Assumption \ref{ass.O} or Assumption \ref{ass.S}, we have,
\begin{align} \label{Power of new kernel with ME}
    &\int b^lx^lK_{\epsilon,2}(x)\,dx=0,\,\,\, l < p.
\end{align}
Meanwhile, for the ordinary smooth case, with the addition of Assumption \ref{ass.O'},
\begin{align} \label{Power of abs new kernel with ME}
    &\int \vert b^lx^lK_{\epsilon,2}(x)\vert\,dx=o_p\left(n^{-\frac{1}{2}}\right),\,\,\, l = p, p+1.
\end{align}
\end{lemma}

\begin{proof}[Proof of Lemma \ref{lemma.power of tsf kernel}]
    First, we have the following equation using the properties of the Fourier transform,  
    \begin{align*}
        &\int \ee^{\ii bx\eta}(bx)^lK_{\epsilon,2}(x)\,dx = (-\ii)^l\left[K^{\text{ft}}(b\eta)\psi_2(\eta)\right]^{(l)}.
    \end{align*}
    Note that \eqref{Decompose of est of decon part2} implies $\psi_2(0) = 0$, Assumption \ref{ass.O}({\romannumeral3}) or Assumption \ref{ass.S}({\romannumeral3}) implies $[K^{ft}(t)]^{(l)}\vert_{t=0} = (\ii)^l\int x^l K(x)\,dx = 0$ for $1\leq l<p$ and $K^{ft}(0) = \int K(x)\,dx = 1$, we have the following result for $l<p$ by setting $\eta = 0$ and thus finish the proof of \eqref{Power of new kernel with ME},
    \begin{align*}
        &\int (bx)^lK_{\epsilon,2}(x)\,dx = (-\ii)^l\left[\psi_2(\eta)\right]^{(l)}\vert_{\eta = 0} = 0.
    \end{align*}   
    
    For the ordinary smooth case, the difference between $l-$th order derivative of $[\hat{f}_\epsilon^{\text{ft}}(t)]^2$ and $l-$th order derivative of $[f_\epsilon^{\text{ft}}(t)]^{2}$ is $O(n^{-1/2}\log(1/b))$, which follows by Assumption \ref{ass.D'}({\romannumeral2}), see Remark $7$ of \cite{kurisu2022uniform} by assigning $d = 1$, $k=l$ and $T_n=1/b$. Subsequently, the order of $\psi^{(l)}_2(t)$ can be given referring to the expression for $f^{ft}_\epsilon$ in Assumption \ref{ass.O}({\romannumeral2}),
    \begin{align*}
        &\sup_{\vert t \vert \leq b^{-1}} \left\vert \psi^{(l)}_2(t) \right\vert = O_p\left(\frac{\log(\frac{1}{b})^2}{n(\inf_{\vert t \vert \leq b^{-1}}\vert f_\epsilon^{\text{ft}}(t) \vert)^4}\right) = O_p\left(n^{-1}b^{-4\alpha}\log(\frac{1}{b})^2\right).
    \end{align*}
    By similar arguments, for the supersmooth case,
    \begin{align*}
        &\sup_{\vert t \vert \leq b^{-1}} \left\vert \psi^{(l)}_2(t) \right\vert = O_p\left(n^{-1}b^{-l}\ee^{4\mu b^{-2}}\log(\frac{1}{b})^2\right).
    \end{align*}
    Note that under Assumption \ref{ass.O'},
    \begin{align*}
        \left\vert b^{l-1}x^{l+2}K_{\epsilon,2}(x)\right\vert & = \left\vert\frac{1}{2\pi b^2}\int \ee^{-\ii t(bx)}\left[K^{ft}(bt)\psi_2(t)\right]^{(l+2)}\,dt\right\vert\\
        &\leq \frac{c_0}{\pi b^3}\sup_{\left\vert t\right\vert\leq b^{-1}}\left\vert\left[K^{ft}(bt)\psi_2(t)\right]^{(l+2)}\right\vert = O_p\left(n^{-1}b^{-(4\alpha+3)}\log(\frac{1}{b})^2\right) = o_p\left(n^{-\frac{1}{2}}\right)
    \end{align*}
    implies $n^{1/2}\vert b^{l-1}x^{l}K_{\epsilon,2}(x)\vert$ decays at a rate of $O(1/x^2)$ as $x$ tends to infinity under Assumption \ref{ass.O'}. In addition, $n^{1/2}\vert b^{l-1}x^{l}K_{\epsilon,2}(x)\vert$ is bounded as $n$ goes to infinity which follows by
    \begin{align*}
        \left\vert b^{l-1}x^lK_{\epsilon,2}(x)\right\vert & = \left\vert\frac{1}{2\pi}\int \ee^{-\ii t(bx)}\left[K^{ft}(bt)\psi_2(t)\right]^{(l)}\,dt\right\vert\\
        &\leq \frac{c_0}{\pi b}\sup_{\left\vert t\right\vert\leq b^{-1}}\left\vert\left[K^{ft}(bt)\psi_2(t)\right]^{(l)}\right\vert = O_p\left(n^{-1}b^{-(4\alpha+1)}\log(\frac{1}{b})^2\right)= o_p\left(n^{-\frac{1}{2}}\right).
    \end{align*}
    Consequently, \eqref{Power of abs new kernel with ME} holds for the ordinary smooth case. 
\end{proof}

\begin{lemma}\label{lemma.power of decon}
Under Assumption \ref{ass.O},
\begin{align} \label{Power of decon without ME Assm O}
    \int \mathcal{K}_\epsilon(x)x^l\ee^{\ii bx\xi}\,dx=I_{\{l\leq\alpha\}}c_l^{os}(\xi)l!b^{-l}+\ii b\sum_{h=0}^\alpha\frac{c_h^{os}(-1)^h\xi^{h+1}}{(h+1)!}\int K^{(h)}(x)x^l\tilde{x}_h^{h+1}\,dx,
\end{align}
holds for $l=0,1,\cdots,p$, where $\bar{x}_h\in(0,x)$ for $h=1,\cdots,\alpha$.

Under Assumption \ref{ass.S},
\begin{align} \label{Power of decon without ME Assm S}
    \int \mathcal{K}_\epsilon(x)x^l\ee^{\ii bx\xi}\,dx=c_l^{ss}(\xi)l!b^{-l}, \qquad \text{for $l\geq0$}.
\end{align}

Suppose that Assumptions \ref{ass.D} and \ref{ass.D'} hold, together with Assumption \ref{ass.O},
\begin{align} \label{Power of decon with ME}
    \sup\limits_{\xi\in\Pi}\left\vert\int (bx)^l\mathcal{K}_{\epsilon,2}(x)\ee^{\ii bx\xi}\,dx\right\vert=o_p\left(n^{-\frac{1}{2}}\right)\quad \text{for }l<p.
\end{align}
For the ordinary smooth case, with the addition of Assumption \ref{ass.O'},
\begin{align} \label{Power of abs decon with ME}
    \int \left\vert (bx)^l\mathcal{K}_{\epsilon,2}(x)\right\vert \,dx=o_p\left(n^{-\frac{1}{2}}\right), \quad \text{for }l=p,p+1.
\end{align}
For the supersmooth case, under Assumption \ref{ass.S} and \ref{ass.S'},
\begin{align}\label{Power of super kernel}
    \sum_{l=0}^\infty\sup\limits_{\xi\in\Pi}\left\vert\int (bx)^l\mathcal{K}_{\epsilon,2}(x)\ee^{\ii bx\xi}\,dx\right\vert=o_p\left(1\right).
\end{align}
\end{lemma}

\begin{proof}[Proof of Lemma \ref{lemma.power of decon}]
    The proofs of \eqref{Power of decon without ME Assm O} and \eqref{Power of decon without ME Assm S} are provided in \cite{dong2022nonparametric}, and the proof of \eqref{Power of decon with ME} follows along the same lines as \eqref{Power of new kernel with ME},
    \begin{align*}
        &\int \mathcal{K}_{\epsilon,2}(x)(bx)^l\ee^{\ii bx\xi}\,dx= \left[\frac{K^{\text{ft}}(b\xi)}{f_\epsilon^{\text{ft}}(\xi)}\psi_2(\xi)\right]^{(l)}.
    \end{align*}
    Similar to the proof in Lemma \ref{lemma.power of tsf kernel}, we have the order of $[\psi_2(t)]^{(l)}$ using the conclusion in \cite{kurisu2022uniform} which also provides
    \begin{align*}
        \left\vert b^{l-1}x^{l+2}\mathcal{K}_{\epsilon,2}(x)\right\vert =& \left\vert\frac{1}{2\pi b^2}\int \ee^{-\ii bxt}\left[\frac{K^{ft}(bt)}{f^{ft}_\epsilon(t)}\psi_2(t)\right]^{(l+2)}\,dt\right\vert\\
        \leq &\frac{c_0}{\pi b^3}\sup_{\left\vert t\right\vert\leq b^{-1}}\left\vert\left[\frac{K^{ft}(bt)}{f^{ft}_\epsilon(t)}\psi_2(t)\right]^{(l+2)}\right\vert= O_p\left(n^{-1}b^{-(5\alpha+3)}\log(\frac{1}{b})^2\right) = o_p\left(n^{-\frac{1}{2}}\right)
    \end{align*}
    under Assumption \ref{ass.O'} and thus \eqref{Power of abs decon with ME} follows. Subsequently, for the supersmooth case,
    \begin{align*}
        \sum_{l=0}^\infty\sup\limits_{\xi\in\Pi}\left\vert\int (bx)^l\mathcal{K}_{\epsilon,2}(x)\ee^{\ii bx\xi}\,dx\right\vert = &\sum_{l=0}^\infty\sup\limits_{\xi\in\Pi}\left\vert \left[\frac{K^{\text{ft}}(b\xi)}{f_\epsilon^{\text{ft}}(\xi)}\psi_2(\xi)\right]^{(l)}\right\vert\\
        \leq & O_p\left(n^{-1}\ee^{-5\mu(1+b^{-1})^2}\log(\frac{1}{b})^2\right)=o_p(1),
    \end{align*}
    where the last equation is assured by Assumption \ref{ass.S'}.
\end{proof}
\end{document}